\documentclass[12pt, oneside]{article}
\usepackage{times}
\usepackage{filecontents}
\usepackage{comment}
\usepackage[moderate]{savetrees}
\usepackage{geometry}
\usepackage[dvipsnames]{xcolor}
\usepackage{graphicx}
\usepackage{setspace}
\usepackage{calrsfs}
\usepackage{amsmath}
\usepackage{amsfonts}
\usepackage{amssymb}
\usepackage{amsthm}
\usepackage{epstopdf,subcaption}

\usepackage[round]{natbib}

\usepackage{sectsty}

\sectionfont{\fontsize{14}{0}\selectfont}
\subsectionfont{\fontsize{13}{0}\selectfont}
\subsubsectionfont{\fontsize{12}{0}\selectfont}

\usepackage{enumerate}

\usepackage[shortlabels]{enumitem}
\usepackage[colorlinks=true,linkcolor=blue]{hyperref}
\hypersetup{colorlinks=true,citecolor=blue}
\usepackage{cleveref}
\usepackage{footnotebackref}

\usepackage{bbm}
\usepackage{ulem}
\newtheorem{theorem}{Theorem}
\newtheorem{lemma}{Lemma}

\newtheorem{corollary}{Corollary}

\newtheorem{claim}{Claim}
\newtheorem{observation}{Observation}

\newtheorem{remark}{Remark}
\newtheorem*{learning-lemma}{Learning Lemma}

\theoremstyle{definition}
\newtheorem{definition}{Definition}

\mathchardef\ordinarycolon\mathcode`\:
\mathcode`\:=\string"8000
\begingroup \catcode`\:=\active
\gdef:{\mathrel{\mathop\ordinarycolon}}
\endgroup

\DeclareMathOperator*{\argmax}{arg\,max}

\newcommand{\snr}{{\psi}}

\begin{document}
	
	\onehalfspacing

	\title{\textbf{Learning from Manipulable Signals}\footnote{This paper originated from two related but independent projects. All authors contributed equally to the current paper. {We are grateful for the helpful comments from Alp Atakan, Costas Cavounidis, Doruk Cetemen, Brett Green, Sam Jindani, Aaron Kolb,  Aditya Kuvalekar, Elliot Lipnowski,  Andrey Malenko, Harry Di Pei, Jo\~{a}o Ramos, Yiman Sun, Philipp Strack, as well as seminar and conference participants at Stanford, MIT Sloan, Boston University, Boston College and the 2020 Econometric Society World Congress.}  \textit{Email:} Ekmekci (mehmet.ekmekci@bc.edu); Gorno (leandro.gorno@fgv.br); Maestri (lucas.maestri@fgv.br); Sun (jiansun@mit.edu); Wei (dwei10@ucsc.edu).}}
	\author{%
		\begin{tabular}{ccccc}
			\textsc{\small Mehmet Ekmekci
			} & \textsc{\small Leandro Gorno
			} & \textsc{\small Lucas Maestri
			} & \textsc{\small Jian Sun
			} & \textsc{\small Dong Wei
			}\\
			{\small \textit{Boston College}} & {\small \textit{FGV EPGE}} & {\small \textit{FGV EPGE}} & {\small \textit{MIT Sloan}} & {\small \textit{UC Santa Cruz}}%
		\end{tabular}
	}
	\date{\today}
	\maketitle
	\thispagestyle{empty}
	
	\begin{abstract}
		We study a dynamic stopping game between a principal and an agent. The agent is privately informed about his type. The principal learns about the agent's type from a noisy performance measure, which can be manipulated by the agent via a costly and hidden action. We fully characterize the unique Markov equilibrium of this game. We find that terminations/market crashes are often preceded by a spike in (expected) performance. Our model also predicts that, due to endogenous signal manipulation, too much transparency can inhibit learning. As the players get arbitrarily patient, the principal elicits no useful information from the observed signal.
		\bigskip
		
	\end{abstract}
	
	\newpage
	\clearpage
	\pagenumbering{arabic} 
	
	\section{Introduction}
	Asymmetric information is pervasive in long-term relationships; meanwhile, learning often takes place during the interactions between different parties. For instance, venture capital (VC) firms face asymmetric information in their investments: startups often have better information about the odds of success of their projects than the investors \citep{brealey1977,chan1983,gompers2004venture}. Moreover, due to the private benefits from receiving continuous funding,\footnote{An extreme example is the former CEO of WeWork, Adam Neumann, who allegedly purchased a corporate jet with the company's money for personal use.} startups are willing to pursue projects that are less viable than what VCs are willing to invest in. VCs, upon agreeing to finance a startup, receive periodical performance reports (subscription growth, number of patents, media and user reviews, etc.) from the startup. These reports may provide information about the viability of the startup. However, the startup may undertake hidden actions to inflate the performance report, tampering with its informativeness. Examples include rideshare platforms who periodically announce their numbers of users and could inflate such statistics by specialized promotions, and Luckin Coffee and Theranos who have been under investigation for fabricating key performance data.
	
	We  analyze learning problems with asymmetric information and hidden actions, and investigate the equilibrium learning dynamics. In our model, a principal (VC) and an agent (startup) are engaged in a relationship that takes place in continuous time. Performance reports are modeled as public signals evolving according to a Brownian motion whose drift depends on the agent's privately-known type and action. If the agent is an \textit{investible} type, then the drift is $\mu>0$; if the agent is a \textit{noninvestible} type, then the drift is $0$ by default, but the agent can take a costly action to boost the drift up to $\mu$. The signals serve only an informational role, and do not affect the principal's payoffs. The principal receives opportunities to terminate the relationship according to a Poisson process, and chooses whether to terminate the relationship whenever such an opportunity arises;\footnote{The Poisson arrival of stopping opportunities captures the frictions in the principal's decision making and implementation, and will technically help us avoid off-path histories in our equilibrium analysis.} she prefers to continue the relationship with the investible type and to terminate the relationship against the noninvestible type.
	
	We study Markov equilibria of this game where the state variable is the public belief that the agent is a noninvestible type. We call the complementary probability, i.e., the probability that the agent is an investible type, the agent's reputation. Our first result establishes the existence and uniqueness of Markov equilibrium. 
	
	In the unique equilibrium, the principal's termination strategy has a cutoff structure --- the principal terminates the relationship if and only if the agent's reputation is sufficiently bad. The agent's equilibrium strategy depends on the magnitude of his discount rate. If his discount rate is greater than a cutoff (i.e., if he does not care much about the future), then he never engages in costly performance boosting. If his discount rate is less than the cutoff (i.e., if he is patient enough), then the agent does not engage in performance boosting when his reputation is very good or very bad, but will do so with intermediate reputation. In particular, the intensity of performance boosting is hump-shaped in the agent's reputation and \textit{peaks} at the principal's termination cutoff.
	
	Our first qualitative finding concerns the relationship between the agent's reputation and the expected performance, measured by the expected drift of the signal from an outsider's perspective. If the agent is so impatient that he never engages in performance boosting, then the expected performance is increasing in the agent's reputation (decreasing in the state variable). However, when the agent is more patient and engages in some performance boosting, the expected performance is non-monotone in the agent's reputation. Starting from an initial good reputation, as the agent's reputation deteriorates, the expected performance first declines, reaching a local minimum, and then it rises, reaching a local maximum precisely at the principal's termination cutoff, and decreases again thereafter (see Figure \ref{fig:nonmonotone}). This finding may help explain why some startups deliver impressive performance reports, such as large sales growth (e.g., Luckin Coffee), extraordinary revenue flow (e.g., Theranos) or rapid expansion (e.g., WeWork), not long before investors pull their funds. It is also consistent with the observation that growing market suspicion and strong (expected) performance can coexist for a period of time.
	
	
	Our second qualitative result concerns the relationship between the amount of information transmission and the transparency of the performance measure. Due to random events such as demand shocks and measurement errors, performance reports are imperfect signals of the agent's type and action,  and we use the signal-to-noise ratio of the process to capture its transparency. In reality, transparency may be determined by the volatility of the product market and may also be affected by how much detail a startup is required to disclose.  
	We show that, due to the agent's endogenous signal manipulation, the principal may be worse off as transparency improves. This result suggests that VCs can sometimes fare better when the startup initially operates in a more volatile market, and that policies that require disclosing too precise information may end up hurting the investors. 
	We also find that exogenous delays in the principal's decision making can sometimes help her, as they facilitate information transmission by reducing the agent's incentive to manipulate the signal.

	Specifically, if the opportunity to terminate the relationship arrives at a rate less than a cutoff (the high-friction case), then in equilibrium, the agent never engages in performance boosting too aggressively  because termination is always unlikely. In this case, as the signal-to-noise ratio grows, the information flow in the principal's optimal stopping problem approaches immediate revelation of the agent's type, which benefits the principal.
	
	On the other hand, if the termination opportunity arrives at a rate greater than the aforementioned cutoff (the low-friction case), then the agent has stronger incentives to engage in performance boosting. We find, perhaps surprisingly, that the principal's payoff is nonmonotone in the signal-to-noise ratio of the performance report, implying that the principal can be worse off when the performance measure becomes more transparent (less noisy). We obtain this result by looking at two extreme cases. At one extreme, if the performance report is independent of the agent's type and action (i.e., uninformative signals), then the principal can never learn about the agent's type and will receive her ``no-information" value. At the other extreme, as the signal-to-noise ratio grows without bound, we show that the principal cannot utilize any information about the agent's type either. Intuitively, in this case the agent will engage in performance boosting aggressively, for otherwise his type would be revealed rapidly. Such aggressive performance boosting is anticipated by the principal, and thus largely reduces the informativeness of the signal. As a result, the principal's equilibrium payoff converges to her ``no-information" value. In contrast to the extreme cases, for intermediate values of the signal-to-noise ratio, the principal will learn some information about the agent's type and get a payoff strictly above her ``no-information" value.
	
	Finally, we investigate the equilibrium outcomes as players get arbitrarily patient. We find a strong manifestation of the ratchet effect in the patient limit of our model. Since the principal cannot commit to refraining from using future information against the agent, a patient agent will engage in performance boosting with almost full intensity in order to maintain his reputation. In the limit, no useful information is revealed, and the principal's lack of commitment hurts her in the most extreme way.\footnote{This result holds if both players get arbitrarily patient at the same rate, or if the agent gets patient at a faster rate than does the principal.}

	While the leading application of our model is the VC-startup relationships, we believe that the economic forces identified by our analysis are relevant in other scenarios, such as voter-politician, manager-worker and purchaser-supplier relationships, where learning with asymmetric information is a critical aspect. On the technical side, our choice of modeling this game in continuous time enables us to obtain semi-closed-form expressions that describe the key equilibrium properties.\footnote{Besides papers reviewed below, recent works that exploit the tractability of continuous-time methods include \cite{demarzo2012}, \cite{bonatti2016}, \cite{ortner2017}, \cite{cisternas2018res} and \cite{varas2020}, among others.} However, our most substantive technical contributions lie in the asymptotic analysis, namely Theorems \ref{t:noisehelps} and \ref{t:patientlimit}, where the lack of a fully closed-form solution presents additional challenges. To deal with them, we establish a new Learning Lemma (see Claim \ref{lem:learninglemma2} and Lemma OA.5) that allows us to measure how frequently those beliefs under which the agent's mimicking intensity is low are visited. This result, while interesting on its own and useful in other settings, 
	enables us to tackle discontinuities which arise in the limit of equilibria of our model. At the end of Section \ref{sec:tranparency}, we explain in more detail how our proofs combine the Learning Lemma with the semi-closed-form solution to obtain these limit results.

	\medskip
	\noindent\textbf{Related Literature. }
	Our paper is most closely related to the reputation literature and the literature on dynamic games with stopping decisions. 
	
	Most of the reputation literature --- starting with \cite{Kreps:1982p524} and \cite{Milgrom:1982p371}, and later generalized by \cite{Fudenberg:1989p83,Fudenberg:1992p39} and most recently by \cite{pei2018reputation} --- investigates whether and how much a long-lived informed player can benefit from its private information in repeated games played against myopic opponents.  The focus is typically on the case where the informed party is arbitrarily patient, and on bounding the informed player's equilibrium payoffs.\footnote{There are also papers that bound equilibrium payoff of the informed player with long-lived uninformed players, e.g., \cite{Schmidt:1993p328,CrippsThomas97,celetani1996,ataekmek,atakan15}.} In contrast, our analysis fully characterizes (Markov) equilibrium behavior for all discount rates, and we uncover new qualitative features that the equilibrium dynamics exhibit.\footnote{Studies on reputation dynamics include \cite{mailath2001,phelan2006public,liu2011,ekmekci2011,liu2013,liu2014}. However, these papers do not share similar equilibrium dynamics or qualitative results that we obtain partly because they look at repeated moral hazard games and/or the uninformed parties are myopic.}
	
	\cite{faingold2011} study reputation effects in games played in continuous time with one long-lived informed player against myopic opponents, and they characterize the set of sequential equilibria. Unlike \cite{faingold2011}, the uninformed player in our game is forward-looking and can terminate the game. More importantly, the termination payoffs depend on the informed player's type, creating \textit{interdependence} of payoffs between the players \citep[similar to][]{pei2018reputation} and thus making their characterization not applicable to our model.
	
	
	There is a growing interest in dynamic games with stopping decisions. \cite{gry2020}, \cite{daley2012}, \cite{kolb2015,kolb2019}, \cite{dilme2019}, \cite{ekme2019} and \cite{sun2020} all study stopping games with two long-lived players, where the uninformed party receives information over time and obtains type-dependent payoffs. \cite{gry2020} introduces type dependence into a real options model wherein the principal's payoff depends on both an exogenously evolving state and the agent's type; in his model, the agent's action is perfectly monitored implying that the principal's inference problem is of a different nature.  In \cite{daley2012}, \cite{kolb2015} and \cite{dilme2019}, the informed player makes the stopping decision while in our paper such decision is made by the uninformed player. This makes the players' incentives in our model quite different from theirs. In \cite{kolb2019}, the agent can only influence the information process by irreversibly changing his type, while in our model the agent can directly manipulate the signal, with his type being persistent. Besides, the qualitative results on the equilibrium dynamics in our paper do not have a counterpart in these papers. \cite{ekme2019} study a similar setting in discrete time, and focus solely on the limiting case with arbitrarily patient players. 
	While our Theorem \ref{t:patientlimit} is a continuous-time analogue of their main finding, we obtain much richer equilibrium dynamics for any fixed discount rate and develop alternative techniques for conducting asymptotic analysis. 
	Finally, \cite{sun2020} studies dynamic censorship with Poisson news, wherein the agent can decide whether to show or hide the bad news after privately observing its realization.\footnote{In \cite{sun2020}, the equilibrium censoring intensity is monotone in the agent's reputation while in our model, the intensity of performance boosting is non-monotone. Besides, his analysis focuses on the welfare implications of the censoring activity, while we examine the welfare effects of better transparency and the ratchet effect at the patient limit.}

	\cite{jackson2016} and \cite{kuvalekar2020} also study dynamic games (in discrete and continuous time, respectively) between two long-run players with stopping decisions. However, the nature of uncertainty and agent's actions in their models are quite different from ours. Specifically, both papers look at a career-concern type of model with symmetric information between the two players, while the agent's actions affecting the signal process are \textit{costless} to the agent and \textit{observable} to the principal. By contrast, in our model the agent has private information about his type, and his action is \textit{costly} and \textit{hidden}. This necessarily makes the principal's inference problem more delicate, as she has to form a conjecture about the agent's action which need coincide with the agent's actual strategy in equilibrium.  Moreover, in our model the agent's trade-off is between improving his reputation and saving the mimicking cost, while in their models the agent is optimizing over the speed of learning (i.e., variance, rather than drift, of the belief process). \cite{orlov2020} also consider a dynamic setting with stopping decisions and symmetrically informed players, and they study the agent's optimal information disclosure policy in a persuasion game.

	\section{Model}
	
	\subsection{Players, types, actions, and information flow}
	A principal (she) and an agent (he), both risk-neutral, interact in continuous time $t\in\left[0,\infty\right)$. 
	
	At any time $t$, an exogenous stopping opportunity arrives according to a Poisson process \(\{J_t\}_{t \geq 0}\) with rate \(\lambda>0\). When the said opportunity arrives, the principal chooses whether to \textit{continue} or irreversibly \textit{stop} the game. The Poisson arrival of stopping opportunities captures the frictions in the principal's decision making and implementation.\footnote{Technically, this assumption ensures that there is no off-equilibrium history/belief. In Remark \ref{rm:1} we discuss what happens as the frictions vanish, i.e. as $\lambda\to \infty$.
		None of our main results requires the frictions to be significant: they hold either for all $\lambda$, or when $\lambda$ is sufficiently large.}
	
	The agent can be one of two types, denoted by $\theta$: an \textit{investible} type ($\theta=I$), or a \textit{noninvestible} type ($\theta=NI$). The agent's type is his private information. From the principal's viewpoint, the initial probability that the agent is a noninvestible type is $p_0\in(0,1)$.
	
	There is a public signal \(\{X_t\}_{t \geq 0}\) that evolves over time. If the agent is an investible type, then the public signal evolves according to the process: 
	\begin{equation*}
		dX_t = \mu dt+ \sigma dB_t,
	\end{equation*}
	where \(\{B_t\}_{t \geq 0}\) is a standard Brownian motion. Without loss, we assume that $\mu>0$ and $\sigma>0$, and we define the \textit{signal-to-noise ratio} $\snr$ of the process as $\snr\equiv \mu/\sigma$. If the agent is a noninvestible type, he chooses an $\alpha_t\in[0,1]$ at any time $t$ when the game has not stopped yet. In this case, his choice controls the drift of the public signal process:
	\begin{equation*}
		dX_t = \mu \alpha_t dt+\sigma dB_t.
	\end{equation*}
	
	The model assumes that the investible type does not have any action choice, and the evolution of the public signal is exogenous conditional on this type (always having a drift of $\mu$). Meanwhile, the noninvestible type chooses a \textit{mimicking intensity}, which can be interpreted as the probability with which the noninvestible type acts the same as the investible type.
	In our leading application of VC investments, we can interpret the public signal as performance reports from the startup and the mimicking action taken by the noninvestible type as \textit{performance boosting}.

	\subsection{Strategies}
	The investible type of the agent does not have an action choice. A strategy for the noninvestible type is a stochastic process \(\{\alpha_t\}_{t \geq 0}\), which takes values in $[0,1]$ and is progressively measurable with respect to the filtration generated by \( \{B_t\}_{t \geq 0}\). Let $\mathcal{A}$ be the set of strategies for the agent.
	
	A strategy for the principal is a stochastic process \(\beta \equiv\{\beta_t\}_{t \geq 0}\), progressively measurable with respect to the filtration generated by \(\{X_t,J_t\}_{t \geq 0}\), which represents the probability with which the principal takes the stopping action conditional on the arrival of a stopping opportunity. Let $\mathcal{B}$ be the set of strategies for the principal.\footnote{We note that the principal only observes the public signal, while the agent knows his own past actions, and thus can recover \(\{B_t\}_{t \geq 0}\) by removing the drift term.}
	
	Given a strategy profile \(\left(\alpha,\beta\right)\) and a prior $p_0$, the principal updates her belief about the agent's type using Bayes' rule, and we let  \(\{p_t\}_{t \geq 0}\) denote the belief process defined by
	\begin{equation}\label{eq:ptdef}
		p_t := \mathbb{P}\left\{\theta = NI \phantom{.}\middle|\phantom{.} \{X_s\}_{s \leq t}\right\}.
	\end{equation}
	Note that the belief process $p_t$, conditional on a continuing relationship, is determined by the strategy of the agent and not affected by the strategy of the principal or the arrival of stopping opportunities.

	\subsection{Payoffs}
	If the game is stopped, the agent receives his outside option which we normalize to $0$. If the game is not yet stopped, the noninvestible agent receives a flow payoff that depends on his action, 
	$u+(1-\alpha_t)c$, where $u>0$ and $c>0$.\footnote{Interpreting the noninvestible type as choosing between mimicking ($A_t=1$) or not ($A_t=0$) and $\alpha_t$ as the probability of taking the mimicking action, we can think of the noninvestible type's flow payoff as defined by $u+\mathbbm{1}_{\{A_t=0\}}c$.}$^,$\footnote{The flow payoff of the investible type in the relationship is always some positive constant, say, $u+c$.} That is, if the noninvestible agent does (not) mimic the investible type then his flow payoff in the relationship is $u$ (resp., $u+c$); thus, $c$ is the flow cost of mimicking. For a given strategy profile,  \(\left(\alpha, \beta\right)\), the expected discounted payoff of the noninvestible agent at time \(t\) is given by
	\begin{equation*}
		U_1(t,
		\alpha, \beta) := \mathbb{E}\left\{\int_t^\mathbb{T} e^{-r_1 (\tau-t)} r_1 \left[u+(1-\alpha_\tau) c\right] d\tau \phantom{.}\middle|\phantom{.} \theta = NI,\left\{B_s\right\}_{s \leq t} \right \},
	\end{equation*}
	where $\mathbb{T}$ is the random time at which the game stops and the expectation is taken over $\mathbb{T}$. This expression can be simplified to
	\[
	U_1(t,
	\alpha, \beta) := \mathbb{E}\left\{\int_t^{\infty} e^{-\Lambda_1(t,\tau,\beta)} r_1 \left[u+(1-\alpha_{\tau}) c\right] d\tau \phantom{.}\middle|\phantom{.} \theta = NI,\left\{B_s\right\}_{s \leq t} \right\},
	\]
	where we define the discounting exponent (taking into account the agent's discount rate $r_1$ and the termination probability)
	\[
	\Lambda_1(t,\tau,\beta) := \int_t^\tau (r_1 + \lambda \beta_s)ds.
	\]
	
	The principal's flow payoff does not depend on the agent's action or the public signal,\footnote{This  assumption seems reasonable in our leading example of venture capital investments, wherein an investor's payoff is mainly driven by the viability (type) of the startup rather than its performance in the initial financing period, while the initial performance is still informative to the investor about the startup's type.} and we normalize her flow payoff to zero. However, the principal receives a lump-sum payoff of $w_{NI}>0$ if the game stops against a noninvestible type, and $w_I<0$ if the game stops against an investible type. That is, relative to continuing the relationship, the principal prefers stopping against a noninvestible type but dislikes terminating an investible type.  Thus, given a strategy profile  $(\alpha,\beta)$, the expected discounted payoff of the principal at time $t$ is given by
	\[
	U_2(t,
	\alpha,\beta) := \mathbb{E}\left\{\int_t^\infty e^{-\Lambda_2(t,\tau,\beta)} \lambda \beta_\tau \left(\mathbbm{1}_{\{\theta = NI\}} w_{NI}+\mathbbm{1}_{\{\theta = I\}} w_I\right)d\tau \phantom{.}\middle|\phantom{.} \left\{X_s\right\}_{s \leq t}  \right\},
	\]
	where we define the discounting exponent (taking into account the principal's discount rate $r_2$ and the termination probability)
	\[
	\Lambda_2(t,\tau,\beta) := \int_t^\tau (r_2 + \lambda \beta_s)ds.
	\]
	Note that $U_2(t,\alpha,\beta)$ is calculated conditional on the stopping opportunity \textit{not} arriving (or having been forgone) at time $t$.
	
	\section{Discussion of Model Assumptions}
	The essential ingredients of our model are the following:
	
	\begin{enumerate}
		
		\item  The agent has private information about his type, and wants to stay in the relationship for as long as possible.
		
		\item  The principal faces a learning problem about the agent's type; she prefers to terminate against the nonivestible type and to continue with the investible type.
		
		\item The noninvestible type can manipulate the drift of a noisy signal at a cost to mimic the investible type's performance.\footnote{Our equilibrium characterization in Theorem \ref{uniqueness} is robust to introducing some strategic behavior to the investible type. See Remark \ref{rm:2} for a discussion.}
		
		\item The signal only serves an informational role and is payoff-irrelevant to the principal.
	\end{enumerate}
	In addition to these assumptions, we adopt a normalization of flow payoffs and outside options to simplify the exposition. 
	Below, we present an alternative but equivalent formulation, which fits better with our VC examples and provides a foundation to the principal's simplified payoff structure.
	
	Suppose that the principal's outside option is independent of the agent's type and is equal to 0. By continuing the relationship the principal incurs a flow cost equal to $b>0$. There is a \textit{revealing event} that arrives according to a Poisson process with rate $\delta$, independent of the agent's type and the signal process. When the event arrives, the game ends delivering a lump-sum payoff to the principal. This payoff is equal to $\pi_I>0$ if the agent is investible and $\pi_{NI}=0$ otherwise. The flow cost represents the continuous financial inputs that the VC contributes to the startup. The revealing event corresponds to the VC's realization of the startup's profitability (type), and the ensuing type-dependent lump-sum payoffs correspond to the value of the startup to the VC upon learning its type. As in the original formulation, the principal can terminate the relationship whenever a stopping opportunity arrives. The arrival follows a Poisson process at rate $\hat{\lambda}$, and is independent of the revealing event. These intermittent stopping opportunities  capture the frictions that are inevitable in a VC's decision making and implementation. For example, the withdrawal of funding may be decided only through board meetings which are called upon once in a while; moreover, a VC that wants to liquidate its shares in a startup may have to wait some time until a buyer shows up.
	
	The (noninvestible) agent's flow payoff is $\hat{u}$ if he engages in performance boosting and  $\hat{u}+\hat{c}$ otherwise. The agent receives a payoff of 0 when the relationship ends, either because the principal terminates it or the revealing event occurs.\footnote{We could also assume that the investible type gets a positive lump-sum reward when the revealing event occurs, but this will not change the game in any way because the revealing event is out of everyone's control.} The discount rates of the agent and the principal are $\hat{r}_1$ and $\hat{r}_2$, respectively.
	
	This formulation is strategically equivalent to the benchmark model with a type-dependent outside option  for the principal. The equivalence is achieved through the following transformation of parameters, which can be verified by standard calculations. The agent's flow payoffs are identical across the two formulations, i.e., $u=\hat{u}$, $c=\hat{c}$, and so is the arrival rate of the principal's stopping opportunity, $\lambda=\hat{\lambda}$. The implied discount rates are augmented by the arrival rate of the revealing event, i.e., $r_1=\hat{r}_1+\delta$, $r_2=\hat{r}_2+\delta$. And finally, the principal's type-dependent outside options are given by
	\[    w_I = \frac{\hat{r}_2b-\delta\pi_{I}}{\hat{r}_2+\delta},\quad w_{NI}=\frac{\hat{r}_2b}{\hat{r}_2+\delta}.
	\]
	As long as $\pi_I>\frac{\hat{r}_2}{\delta}b$ (i.e., if the reward to the principal upon learning that the agent is investible is large enough), we have $w_I<0<w_{NI}$, as in our baseline model. 
	
	Finally, our model focuses solely on the adverse-selection aspect of the problem, by assuming that performance reports do not directly affect the principal's payoff. This seems reasonable in contexts such as VC-startup relationships, wherein an investor's payoff is mostly driven by the viability (type) of the startup rather than its performance in the initial financing period, though the initial performance is informative about the startup's type.\footnote{Moreover, we have verified that our equilibrium characterization still holds even after allowing for some dependence of the principal's flow payoff on the agent's action. This analysis is not included in the manuscript and is readily available upon request.}

	\section{Equilibrium Characterization}
	\subsection{Equilibrium Concept}
	An \textbf{equilibrium} is a strategy profile \((\alpha,\beta)\) such that 
	\begin{eqnarray*}
		U_1(t,
		\alpha,\beta) &\geq& U_1(t,
		\tilde{\alpha},\beta), \\
		U_2(t,
		\alpha,\beta) &\geq& U_2(t,
		\alpha,\tilde{\beta}),
	\end{eqnarray*}
	for all alternative strategies \(\tilde{\alpha}\in\mathcal{A}\) and \(\tilde{\beta}\in \mathcal{B}\), almost surely for all $t\geq 0$.
	
	Let
	\footnote{A function $f:(0,1)\to [0,1]$ is piecewise Lipschitz if there exist $n\in \mathbb{N}$ and $0= x_1<x_2<...<x_n=1$, such that for each $i\in \{1,...,n-1\}$, there exists a Lipschitz function $f_i$ on $[x_i,x_{i+1}]$ such that $f_i(p)=f(p)$ for all $p\in(x_i,x_{i+1})$.}
	\[
	\mathcal{P}:=\{f:(0,1)\to[0,1],f\text{ is right-continuous and piecewise Lipschitz}\}.
	\] 
	Recall that the belief process defined in \eqref{eq:ptdef} is determined by the agent's strategy. We say that a strategy $\alpha$ of the agent is \textit{Markovian} if there exists \textit{policy function} $a\in \mathcal{P}$ such that $\alpha_t=a(p_t)$ for all $t\geq 0$. An equilibrium \((\alpha,\beta)\) is \textbf{Markovian} if there exist policy functions \(a,b\in \mathcal{P}\) such that \(\alpha_t = a(p_t)\) and \(\beta_t = b(p_t)\) for all \(t\geq 0\). In this case, we say that the policy profile $(a,b)\in \mathcal{P}^2$ is \textit{induced} by \((\alpha,\beta)\). 
	
	Given a Markovian equilibrium $(\alpha,\beta)$, let $\text{SP}(\alpha)$ be the set of posteriors reached on the equilibrium path.\footnote{Consider a Markovian equilibrium, $(\alpha,\beta)$, and the underlying probability space $\left( \Omega ,\mathfrak{F,}%
		\mathbb{P}\right) $.  For each $p\in \left( 0,1\right) ,$ we
		define $\Phi \left( p\right) :=\left\{ \omega \in \Omega :\exists t\leq 
		\mathbb{T}\text{ such that }p_{t}(\omega )=p\right\} $, where $\mathbb{T}$ is the equilibrium stopping time. The belief span, $\text{SP}(\alpha)$,
		is the set of all $p$ such that $\mathbb{P}\left( \Phi \left( p\right)
		\right) >0$. Because $p_t$ in a continuing relationship depends only on the agent's strategy $\alpha$ and the principal's stopping opportunity may not arrive at any $t$, the belief span is also solely determined by $\alpha$. Consequently, this notion of belief span can be defined for any (Markovian) strategy of the agent.} 
	
	\begin{lemma}\label{lem:fullspan}
		Any Markovian equilibrium $(\alpha,\beta)$ with an induced policy profile $(a,b)$ satisfies  $i)$ $\sup_{p\in (0,1)}a(p)<1$, and $ii)$ $\text{SP}(\alpha)=(0,1)$.
	\end{lemma}
	
	Lemma \ref{lem:fullspan} says that in any Markovian equilibrium the agent's action is always bounded away from full mimicking and that every posterior belief is reached with positive probability. Intuitively, if the noninvestible type is expected to choose $a=1$ at some belief $\hat{p}$, then the signal from that time on becomes uninformative, making $\hat{p}$ an absorbing state. But if the belief is not moving, the noninvestible type's best reply at state $\hat{p}$ is to choose $a=0$, which violates the (implicit) requirement that the agent's equilibrium action must coincide with the principal's conjecture about his action. In fact, one can show that an equilibrium policy function $a(\cdot)$ must be bounded away from $1$, and thus the variance of the belief process is always bounded away from $0$.\footnote{This result holds for any fixed value of the parameters. The lower bound on the volatility of the belief process, $1-\sup_{p\in (0,1)} a(p)$, depends on players' discount rates, the arrival rate of the Poisson process, and the
		signal to noise ratio. Therefore, the informativeness of the public signal can get arbitrarily
		close to zero in some cases (e.g., when the signal to noise ratio is very high or the agent
		is very patient, as we will see in Sections \ref{sec:tranparency} and \ref{sec:patience}).} Together with the Poisson arrival of stopping opportunities, this makes all interior beliefs reachable on the equilibrium path.
	
	Given a Markovian equilibrium $(\alpha,\beta)$, the continuation payoff at time $t$ depends only on the public belief $p_t$. Hence, we define the value function of the (noninvestible) agent as
	\[
	V(p) := \mathbb{E}\left\{U_1(t,\alpha,\beta)\phantom{.}\middle| \phantom{.} p_t = p, \theta = NI\right\}
	\]
	and the value function of the principal as
	\[
	W(p) := \mathbb{E}\left\{U_2(t,\alpha,\beta)\phantom{.}\middle| \phantom{.} p_t = p\right\}
	\]
	for every \(p \in (0,1)\).

	We say that a value function is \textbf{regular} if it is continuously differentiable everywhere, and twice continuously differentiable everywhere except perhaps at a finite number of points. We say that a Markovian equilibrium \((\alpha,\beta)\) is \textbf{smooth} if the associated value functions are regular and the agent's policy function $a(\cdot)$ is Lipschitz. We refer to smooth Markovian equilibria simply as \textbf{Markov equilibria}.\footnote{We emphasize that ``smoothness" is built into our definition of the term ``Markov equilibrium." We do not look for non-smooth Markovian equilibria in this paper, and any claim about equilibrium uniqueness does not rule out the possibility of non-smooth Markovian equilibria.} Moreover, when there is no confusion, we denote a Markov equilibrium by the policy profile $(a,b)$ that it induces.
	
	\subsection{Characterization}
	
	We first introduce some terminology to define properties of policy functions for the principal and the agent. Recall that the state variable $p$ is the principal's belief that the agent is noninvestible.
	\begin{definition}
		The policy function $b\in \mathcal{P}$ for the principal has a \textbf{cutoff structure} if there exists $\tilde{p}\in[0,1]$ such that $b(p)=0$ for $p<\tilde{p}$, and $b(p)=1$ for $p>\tilde{p}$. We refer to $\tilde{p}$ as the cutoff belief of $b$.
	\end{definition}
	
	\begin{definition}
		The policy function $a\in \mathcal{P}$ for the (noninvestible) agent is \textbf{fully separating} if $a(p)=0$ for all $p\in(0,1)$. 
	\end{definition}

	\begin{definition}\label{def:humpshaped}
		The policy function $a\in \mathcal{P}$ for the (noninvestible) agent is \textbf{hump-shaped} if $a$ is continuous and there are cutoffs $0<p_{L}<p^*<p_{R}<1$ such that $a(p)=0$ for $p\leq p_L$, strictly increasing on $(p_L,p^*)$, strictly decreasing on $(p^*,p_R)$, and $a(p)=0$ for $p\geq p_R$. 
	\end{definition}

	\begin{theorem}\label{uniqueness}
		There always exists a unique Markov equilibrium $(a,b)$. In this equilibrium, $b$ has a cutoff structure with some cutoff belief $p^*\in (0,1)$. Moreover, there exists $r^*>0$ such that
		\begin{enumerate}
			\item If $r_1 \geq r^*$, then $a$ is fully separating.
			\item If $r_1<r^*$, then $a$ is hump-shaped and is maximized at $p^*$.
		\end{enumerate}
	\end{theorem}
	
	Theorem \ref{uniqueness} characterizes the structure of the unique Markov equilibrium. First, the principal uses a cutoff strategy. This follows from the type-dependent stopping payoff of the principal, and the absence of flow payoffs.
	
	Second, the noninvestible agent's behavior depends on his discount rate. If he is impatient (i.e., with a high discount rate), then he never mimics the investible type, because he always finds the saving of the mimicking cost to outweigh the benefit of having a better reputation. A richer dynamics opens up if the agent is patient (i.e., with a low discount rate). In this case, his behavior can be described by three reputation phases: good, medium and bad, as depicted in Figure \ref{fig:agentbestreply}. Both in the good and the bad reputation phases, the noninvestible type does not mimic the investible type at all, but for different reasons: when his reputation is good ($p<p_L$), the relationship is highly stable, so the noninvestible agent gains little from further improving his reputation by mimicking; when his reputation is bad ($p>p_R$), termination is so imminent that the noninvestible agent gives up building reputation. In the intermediate phase, however, the noninvestible type first starts to mimic more often as his reputation worsens in order to slow down the principal's learning. We interpret this as a \textit{``scramble-to-rescue" effect}: the agent increases his mimicking intensity (before $p^*$) as the relationship gets less stable. After certain point, he gradually gives up as the relationship becomes doomed. His mimicking intensity is highest at belief $p^*$ when the principal's action switches  from continuing the relationship to termination.\footnote{Even though the principal's optimal action switches at $p^*$, the agent's mimicking intensity does not immediately drop to $0$ right after the belief passes $p^*$. This is because the relationship can only be terminated when a stopping opportunity arrives, leaving some hope for the agent to rebuild his reputation and avoid termination.}
	
	
	That the agent's mimicking intensity reaches its peak around $p^*$ can be understood from an equilibrium perspective. In general, the agent's incentive to manipulate depends on: i) how responsive the belief is to signal realizations; ii) how sensitive the principal's decision is to belief changes. If the agent manipulates more \textit{in equilibrium}, then the principal believes that the signal is less informative and her belief is less responsive to the signal realizations, which reduces the agent's incentive to manipulate through (i). When the principal's decision is very sensitive to the agent's reputation (namely, at reputations around the cutoff $p^*$), the agent's manipulation has to be high in order to partially neutralize the effect of such sensitivity, and ensure that the marginal cost of manipulation equals the marginal benefit.
	
	\begin{figure}[h!]
		\centering
		\includegraphics[width=0.8\linewidth]{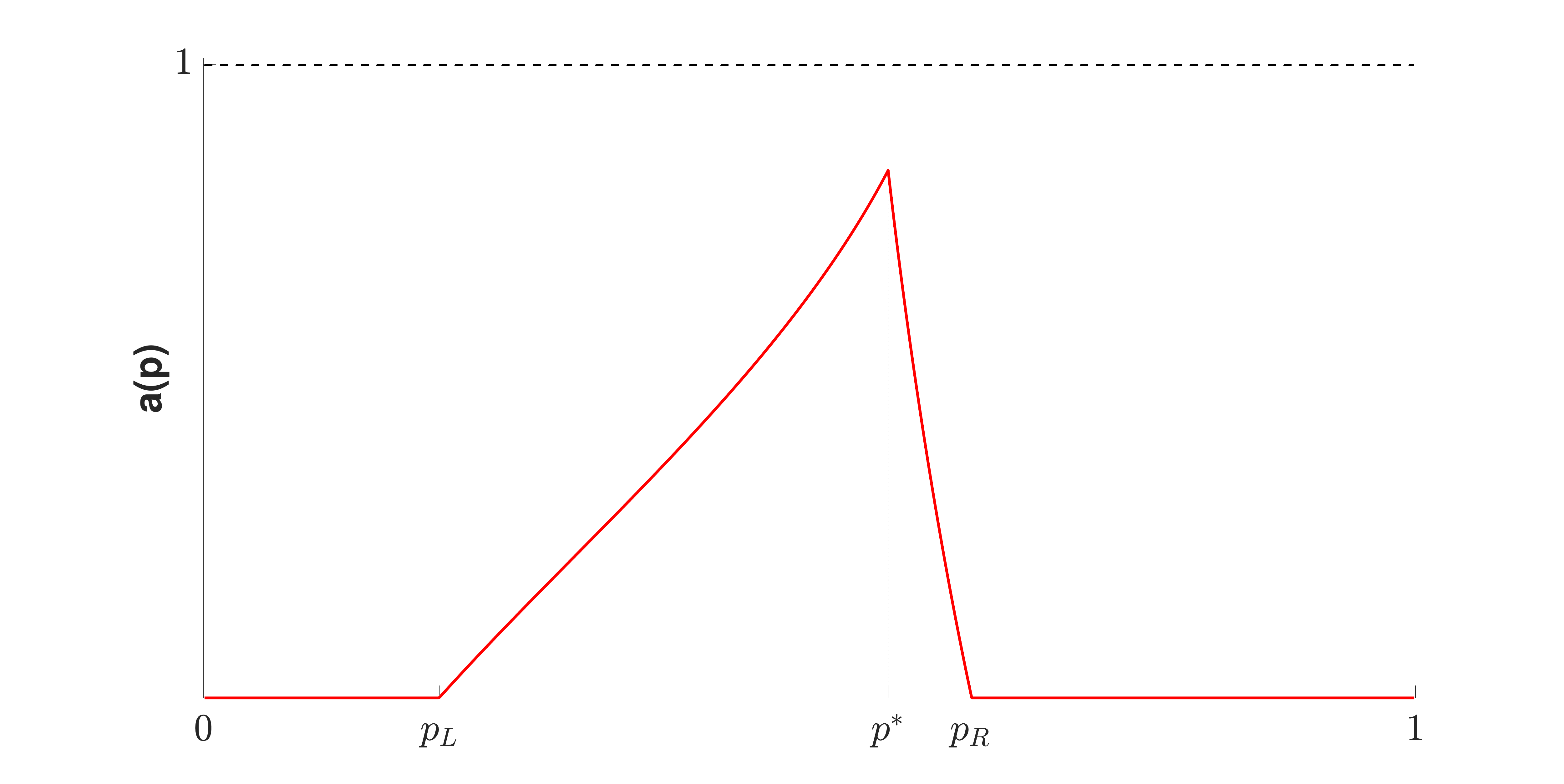}
		\caption{Agent's Equilibrium Policy Function When $r_1<r^*$. This figure is plotted under the following parameter values: $r_1=0.5$, $r_2=0.5$, $\lambda=2$, $\snr=1.5$, $u=1$, $c=1$, $w_{NI}=1$, $w_I=-1$. In equilibrium, $p_L\approx 0.195$, $p^*\approx 0.565$, $p_R\approx 0.633$.}
		\label{fig:agentbestreply}
	\end{figure}
	
	Finally, the cutoff discount rate, $r^*$, can be characterized in closed form. Specifically, $r^*$ is the unique solution to the following equation:
	\begin{equation}\label{eq:r*}
		r^*(\sqrt{1+8r^*/\snr^2} +\sqrt{1+8(r^*+\lambda)/\snr^2}) + \lambda(\sqrt{1+8r^*/\snr^2}+1)=4\lambda\left(\frac{u}{c}+1\right).
	\end{equation}
	Some comparative statics results are readily obtained from \eqref{eq:r*}. In particular, $r^*$ increases with $\lambda$, $\snr$ and $\frac{u}{c}$. This is intuitive, because the noninvestible type will have a higher incentive to mimic if: \textit{i)} the stopping opportunity arrives more frequently and thus the relationship is less stable; \textit{ii)} the signal-to-noise ratio is higher and thus a manipulation of signal is more profitable; \textit{iii)} mimicking is relatively less costly. We also note that $r^*$ does not depend on the principal's payoff parameters ($r_2$, $w_{NI}$ and $w_{I}$).
	
	\begin{remark}\label{rm:1}
		As the stopping frictions vanish (i.e., as $\lambda\to\infty$), the cutoff discount rate $r^*$ converges from below to a finite number $\bar{r}$. For any fixed $r<\bar{r}$, the agent's equilibrium policy function converges to one that resembles Figure \ref{fig:agentbestreply} for $p<p^*$ and is equal to $0$ for all $p>p^*$, that is, $p_R^*$ converges to $p^*$, creating a discontinuity at $p^*$. This is intuitive because in the limit the agent's incentive when $p<p^*$ is similar to what we explained before, but once the belief is above $p^*$ the agent expects the relationship to be terminated in the next instant regardless of what he does, and thus he should choose $a=0$ to save the mimicking cost.\footnote{Technically, if $\lambda$ is set to $\infty$, i.e., if we literally allow the principal to terminate the relationship whenever she wants, for that game additional refinement is needed to preserve equilibrium uniqueness, because once $p>p^*$ the agent expects the relationship to be terminated right away in which case his action choice in that instant has no payoff consequence. See \cite{kuvalekar2020} for a detailed discussion and a refinement that will select the limiting policy function we described in the limiting game. In our model, the Poisson arrival of stopping opportunities helps us avoid such complications and obtain equilibrium uniqueness for every finite $\lambda$ without additional refinements.}
	\end{remark}
	\begin{remark}\label{rm:2}
		The assumption that the investible type does not have an action choice makes him a ``commitment/behavioral" type in the sense of the reputation literature. However, many of the insights of this paper are robust to certain forms of strategic behavior. For instance, if we allow the investible type to costlessly choose a drift, then the unique equilibrium we characterize in Theorem \ref{uniqueness} remains an equilibrium in this modified game.\footnote{The equilibrium is no longer unique though. For example, both types choosing zero drift and the principal ignoring the signal is always an equilibrium in that case.}
	\end{remark}
	
	Below we describe our approach to proving Theorem \ref{uniqueness}. Finding a Markov equilibrium amounts to finding a policy profile $(a,b)$, and a conjecture that the principal holds about the agent's strategy such that: \textit{i)} the principal's conjecture determines her interpretation of public signal histories into her beliefs about the agent's type; \textit{ii)} the principal's policy $b$ is optimal given her conjecture about the agent's strategy; \textit{iii)} the agent's policy $a$ is optimal given $b$ and the principal's conjecture; \textit{iv)} the principal's conjecture coincides with $a$.
	
	Specifically, any Markov equilibrium $(a,b)$ satisfies the optimality conditions stated below.\footnote{The optimality conditions \eqref{eq:optimality-principal} and \eqref{eq:agent-optimality} restrict the players to maximize their payoffs over Markov controls in $\mathcal{P}$. This is for expository purposes and is without loss. In the proof of Theorem \ref{uniqueness} we verify that the equilibrium strategies are mutual best replies among all strategies in $\mathcal{B}$ and $\mathcal{A}$.} 
	
	\underline{Principal's Optimality:}
	
	\begin{equation}\label{eq:optimality-principal}
		b \in \argmax_{\tilde{b}\in \mathcal{P}}\hat{W}(p,a,\tilde{b}), 
	\end{equation}
	where
	\[
	\hat{W}(p,a,\tilde{b}):=\mathbb{E}\{e^{-r_2\nu}\left(\mathbbm{1}_{\{\theta = NI\}} w_{NI}+\mathbbm{1}_{\{\theta = I\}} w_I\right)\},
	\]
	where $p_0=p$, $\nu$ is the time when the game stops, controlled by both $\tilde{b}$ and $\{J_t\}_{t\geq 0}$, and the evolution of \(\{p_t\}_{t \geq 0}\) is given by the SDE \citep[for a formal derivation, see, e.g.,][]{bolten1999}
	\begin{equation}\label{eq:law-of-p}
		dp_t =  -\snr (1-a_t) \gamma(p_t) d\tilde{B}_t,
	\end{equation}
	In (\ref{eq:law-of-p}), $\snr$ is the signal-to-noise ratio parameter, $a_t$ is a function of $p_t$, \(\gamma : [0,1] \to \mathbb{R}_+\) is defined by \(\gamma(p):= p(1-p)\), and \((\tilde{B}_t)_{t \geq 0}\) is the innovation process associated with the filtering of the principal, i.e.,
	\begin{equation}\label{eq:law-of-tilde-B}
		d\tilde{B}_t := \frac{dX_t - \mu (p_ta_t+1-p_t)}{\sigma} = \frac{dX_t}{\sigma}-  \snr (p_ta_t+1-p_t) dt.
	\end{equation}
	
	The optimality condition (\ref{eq:optimality-principal}) requires that $b$  maximizes the principal's payoff when the agent is using policy $a$.
	
	\underline{Agent's Optimality:}

	\begin{equation}\label{eq:agent-optimality}
		a\in \argmax_{\tilde{a}\in \mathcal{P}}\hat{V}(p,\tilde{a},b;a),
	\end{equation}
	and 
	\[
	\hat{V}(p,\tilde{a},b;a):=\mathbb{E}\left\{\int_0^{\nu} e^{-r_1\tau} \left\{r_1 \left[(1-\tilde{a}(p_\tau)) c +  u\right]\right\}d\tau\right\}
	\]
	where $p_0=p$, $\nu$ is the time when the game stops, and the evolution of \(\{p_t\}_{t \geq 0}\) is given by substituting $dX_t=\mu \tilde{a}_tdt+\sigma dB_t$ into equations (\ref{eq:law-of-p}) and (\ref{eq:law-of-tilde-B}). Specifically, from the noninvestible type's perspective, the belief process satisfies:
	\begin{equation}\label{eq:dptagent}
		dp_t=\snr^2(1-a_t)[1-\tilde{a}_t-p_t(1-a_t)]\gamma(p_t)dt-\snr(1-a_t)\gamma(p_t)dB_t.
	\end{equation}
	In a Markov equilibrium $(a,b)$, the principal has a conjecture about the agent's behavior, which determines  how she interprets any history of signal realizations into her belief about the agent's type. This conjecture has to coincide with the agent's policy $a$ in equilibrium. If the agent contemplates a deviation from the equilibrium, this would not affect the processes in equations (\ref{eq:law-of-p}) and (\ref{eq:law-of-tilde-B}) (which jointly describe the dependence of beliefs on the public history), but would affect the process that governs the evolution of $X_t$ (public histories). The necessary condition (\ref{eq:agent-optimality}) requires that the agent does not have a profitable deviation from his equilibrium policy function $a$, when the principal conjectures that the agent is using this policy function.

	We now build on the implications of the necessary conditions outlined above. We first show that in any Markov equilibrium, the principal's policy function has a cutoff structure: she terminates the relationship if and only if the agent's reputation is bad enough. We then show that the agent's equilibrium policy function must be either fully separating (i.e., never mimicking) or hump-shaped. Finally, the existence and uniqueness of Markov equilibrium follow from a fixed-point argument.
	
	Let $R(p):=pw_{NI}+(1-p)w_I$ be the principal's expected payoff if the relationship is terminated at belief $p$. Define $p^{**}:=R^{-1}(0)>0$ and $p_H:=R^{-1}\left(\frac{\lambda}{r_2+\lambda}w_{NI}\right)<1$.
	
	\begin{lemma}\label{lemma:principalbestreply}
		If $(a,b)$ is a Markov equilibrium, then $b$ has a cutoff structure with a cutoff belief $p^*\in [p^{**},p_H]$.
	\end{lemma}

	To prove this result, we utilize the optimality condition in (\ref{eq:optimality-principal}). Observe that the equilibrium value function, $W(p)$, is such that $W(p):=\hat{W}(p,a,b)$. Then, the principal's value and policy functions must satisfy the following HJB equation:
	\begin{equation}\label{eq:HJBprincipal}
		r_2 W(p) = \max_{\tilde{b} \in [0,1]} \left\{\frac{1}{2} \snr^2 [1-a(p)]^2 \gamma(p)^2 W''(p) + \lambda \tilde{b} \left[R(p) - W(p)\right] \right\}.
	\end{equation}
	It is clear that \(b(p) = 0\) whenever \(R(p)<W(p)\), and \(b(p)=1\) whenever \(R(p)>W(p)\). In the proof, we show that these functions have a unique intersection point. Moreover, because terminating the relationship when $p<p^{**}$ gives the principal a negative payoff, and because the stopping opportunity arrives only once in a while which bounds her payoff from waiting by $\frac{\lambda}{r_2+\lambda}w_{NI}$, the principal's optimal stopping threshold must be between $p^{**}$ and $p_H$. See Figure \ref{fig:optimalcutoff} for an illustration.
	
	\begin{figure}[ht]
		\centering
		\includegraphics[width=0.8\linewidth]{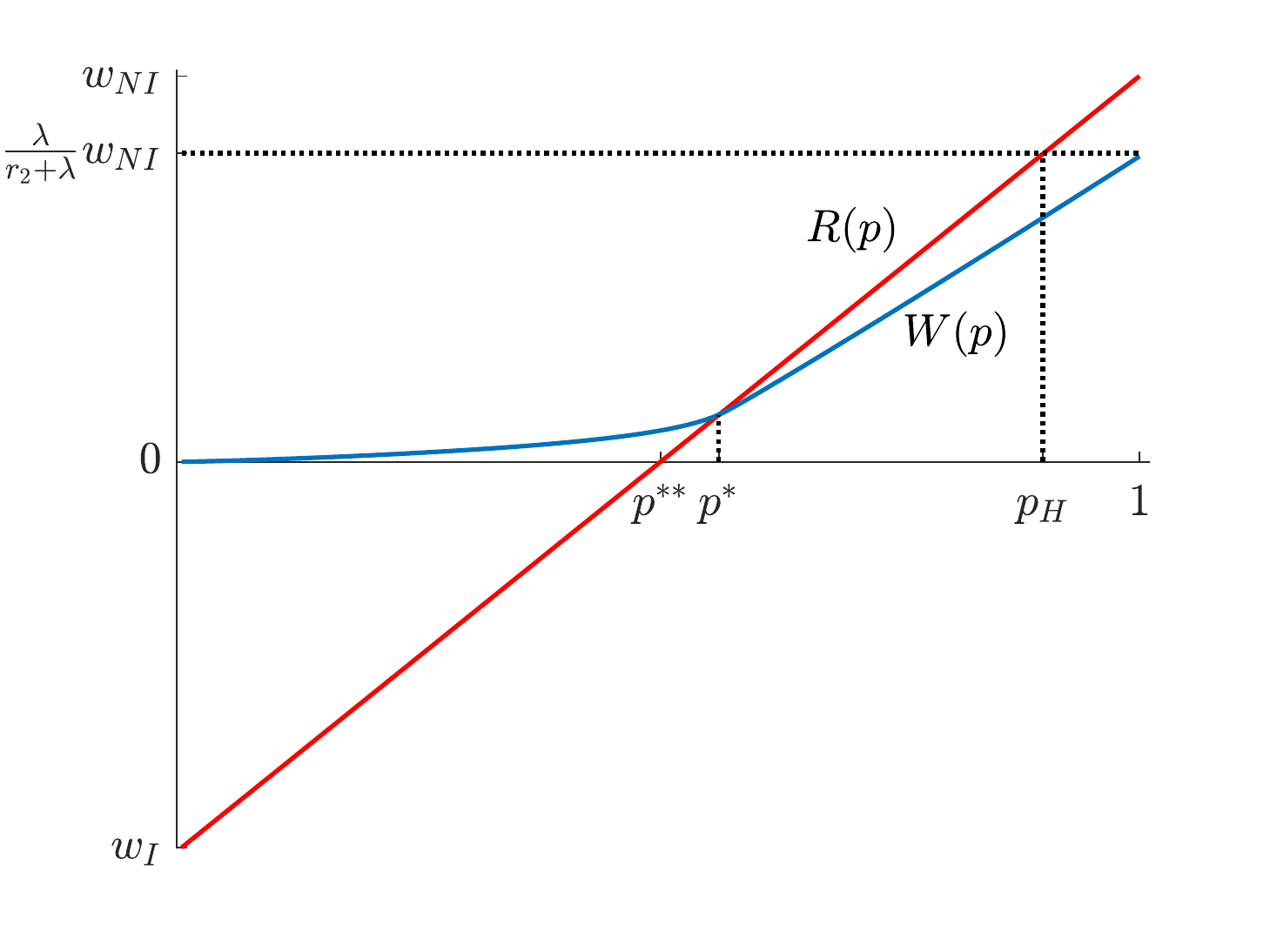}
		\caption{Principal's Equilibrium Cutoff. This figure is plotted under the following parameter values: $r_1=0.5$, $r_2=0.5$, $\lambda=2$, $\snr=1.5$, $u=1$, $c=1$, $w_{NI}=1$, $w_I=-1$. In equilibrium, $p^{**}= 0.5$, $p^*\approx 0.565$, $p_H=0.9$.}
		\label{fig:optimalcutoff}
	\end{figure}

	We now turn to the agent's behavior. 
	
	\begin{lemma}\label{lemma:optimality-agent}
		Suppose $b\in \mathcal{P}$ is a cutoff policy function for the principal with cutoff belief $p^*$. Then, there is a unique policy function  $a \in \mathcal{P}$ for the agent such that i) $\hat{V}(p,a(p),b(p);a(p))$ is a regular function of $p$, ii) $a$ is Lipschitz and $\sup_{p\in (0,1)}a(p)<1$, and iii) $a$ satisfies (\ref{eq:agent-optimality}).  Moreover, this unique policy function is fully separating if $r_1\geq r^*$, and is hump-shaped if $r_1< r^*$.
	\end{lemma}
	
	The proof of Lemma \ref{lemma:optimality-agent} is more involved. This is because finding a solution to program (\ref{eq:agent-optimality}) is akin to finding a fixed point: the policy $a$ for the agent is optimal when the principal holds the conjecture $a$. Nonetheless, we are able to characterize its unique solution in closed form.
	
	Intuitively, if the agent is impatient, the short-run incentives determine his behavior, and the noninvestible type will never pay the cost to mimic the investible type, leading to a fully separating policy function. If the agent is patient, full separation can no longer be part of an equilibrium. This is because the fully separating policy function, if conjectured by the principal, generates opportunities to build a reputation rather fast; and when the agent cares enough about the future, it will give strict incentives to the noninvestible type to mimic.
	
	What is the dynamics of the agent's mimicking intensity when he is patient? When the  public belief $p$ is very small, it takes a long time for the belief to increase all the way up to the termination cutoff. Hence, the limited benefit of further improving reputation cannot justify the mimicking cost. As a result, $a(p)=0$ for low $p$. When $p$ is very large, it takes so long for the agent to regain his reputation that he simply gives up. Consequently, $a(p)=0$ for high $p$. For intermediate $p$, the agent's short-run temptation and long-run benefits are more balanced, so that $a(p)\in(0,1)$.  Specifically, after the good reputation phase, for $p\in (p_L,p^*)$, the noninvestible type starts to mimic more often as his reputation worsens (i.e., the scramble-to-rescue effect). Such an incentive peaks at $p^*$ where the principal's action is most sensitive to a change in belief. After that, for $p\in (p^*,p_R)$, the noninvestible type gradually gives up restoring his reputation as termination becomes more imminent.

	\section{Non-Monotonicity of Expected Performance}\label{sec:nonmonotone}

	We saw in Theorem \ref{uniqueness} that the noninvestible type will engage in performance boosting whenever he is sufficiently patient, in which case $a(\cdot)$ peaks at the termination cutoff $p^*$. This has been referred to as the ``scramble-to-rescue" effect: the agent increases his mimicking intensity (before $p^*$) as the relationship gets less stable.
	
	What are the implications of this effect on the observables? An outsider (the principal or a modeler) does not see the agent's type or action, but can observe his performance, such as subscription growth or progress reports. In our model, the expected performance at time $t$ is given by  \[EP_t:=\frac{E[dX_t]}{dt}=\mu \left(\underbrace {1-p_t}_{\text{investible}}+\underbrace{p_ta(p_t)}_{\text{noninvestible}} \right).\]
	
	Holding constant $a<1$, the expected performance decreases with $p$. We call this the \textit{belief effect}: if the agent is more likely to be noninvestible, then the expected performance is lower. This is the entire story if the equilibrium is fully separating, in which case $a(p)=0$ everywhere and $EP(p)=\mu (1-p)$.
	
	However, if the equilibrium is not fully separating, the noninvestible type's mimicking intensity $a$ is no longer constant: it is increasing below $p^*$ due to the scramble-to-rescue effect. Hence, whether the expected performance increases or decreases with the public belief depends on which of the two effects is stronger. The following theorem characterizes the evolution of the expected performance when the stopping opportunity arrives sufficiently fast.
	
	\begin{theorem}\label{t:zigzag-shape}
		Fixing all parameters of the model other than $r_1$ and $\lambda$, there exists $\bar{\lambda}$ such that for all $\lambda>\bar{\lambda}$, $EP(p)$ is non-monotone whenever the equilibrium is not fully separating (i.e., whenever $r_1<r^*$). In particular, $EP(p)$ is
		
		\begin{itemize}
			\item strictly decreasing for $p\in[0,\underline{p})$, where $\underline{p}$ is  in $[p_L,p^*)$;
			
			\item strictly increasing for $p\in (\underline{p},p^*)$;
			
			\item strictly decreasing for $p \in (p^*,1]$. 
		\end{itemize}
	\end{theorem}
	
	Theorem \ref{t:zigzag-shape} shows that if the arrival rate of the stopping opportunity is large enough,  the scramble-to-rescue effect will dominate the belief effect when the public belief is less than but close to the termination cutoff $p^*$. As a result, the expected performance reaches a local maximum at $p^*$ whenever the agent's equilibrium policy function is hump-shaped (see Figure \ref{fig:nonmonotone}).\footnote{The lower bound on $\lambda$ is not crucial for this qualitative predication. In fact, even if $\lambda$ is small, we can show that the expected performance is either decreasing or has the shape described in Theorem \ref{t:zigzag-shape} (see Lemma \ref{lem:EPpossibilities}). Moreover, for any $\lambda$, we can find a $\bar{r}$ (less than $r^*$) such that the expected performance is non-monotone whenever $r_1<\bar{r}$. The only difference for small $\lambda$ is that we cannot say definitively what will happen when $r_1\in (\bar{r},r^*)$.}$^,$\footnote{The sudden drop of expected performance after $p^*$ carries over to the time domain. Specifically, because $EP(p)$ is locally maximized (thus concave) at $p^*$, one can show that the stochastic process $EP_t$ is a local supermartingale when $p_t=p^*$.}
	
	\begin{figure}[ht]
		\centering
		\includegraphics[width=0.8\linewidth]{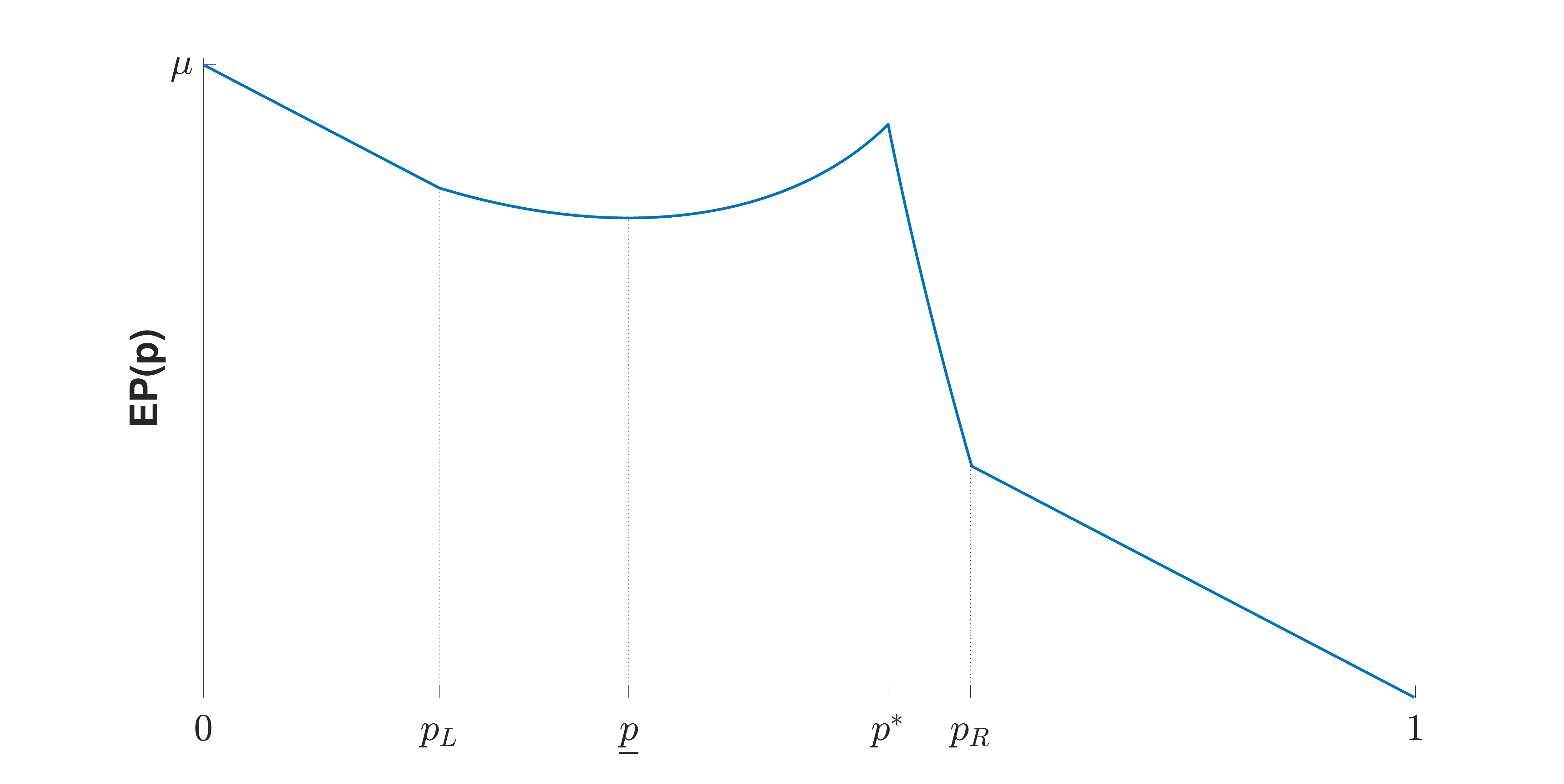}
		\caption{Agent's Expected Performance When $\lambda>\bar{\lambda}$ and $r_1<r^*$. This figure is plotted under the following parameter values: $r_1=0.5$, $r_2=0.5$, $\lambda=2$, $\snr=1.5$, $u=1$, $c=1$, $w_{NI}=1$, $w_I=-1$.}
		\label{fig:nonmonotone}
	\end{figure}

	In the context of our applications, Theorem \ref{t:zigzag-shape} offers an empirical prediction of our model: when performance boosting is expected to happen, terminations are preceded by a spike in expected performance. This seems consistent with a number of famous cases of corporate failure, such as Theranos, Luckin Coffee and WeWork: there were periods of time during which market suspicions about their business models grew, meanwhile the companies kept performing strongly and/or expanding aggressively prior to the crashes of their market values.\footnote{For example, in the third quarter of 2019, Luckin Coffee reported a 470.1\%  increase in the total items sold from 7.8 million in the same quarter of 2018. Its stock price was slashed by 75\% in April 2020, following suspicion and then admission of fabricating sales data. Likewise, before scandals started to unravel, Theranos falsely claimed in 2014 that the company had annual revenues of \$100 million, a thousand times more than the actual figure of \$100,000. In the case of WeWork, the company once had expanded to over 86 cities in 32 countries, despite growing suspicion about its profitability. However, in September 2019, the companied delayed its IPO, followed by a 90\% slash in valuation and enormous layoffs.
	}

	\section{Environments with Low Volatility / High Transparency}\label{sec:tranparency}
	The noise in the performance measure may come from various random events such as temporary demand shocks, measurement errors, etc., making $X_t$ only an imperfect signal of the agent's type. We say that a performance measure is more transparent if it is less affected by the noise component, and we can use the signal-to-noise ratio of the process to capture its \textit{transparency}. 
	In reality, transparency may be determined by the intrinsic volatility of the product market; it may also be affected by how much detail about the market or the project that a startup is required to disclose in a performance report. Indeed, one may expect that, other things equal, disclosing more details about the objective market conditions can help an investor better understand the numbers in the report. 
	
	In this section, we investigate the question: Do improvements in transparency always benefit the principal?
	We show that, under some parameters, improving transparency can inhibit learning and hurt the principal, due to the agent's endogenous response through signal manipulation. We also find that frictions in the principal's decision making (which cause decision delays) can sometimes help her, as they serve as a commitment of not terminating the relationship too quickly and thus reduce the agent's incentive to manipulate the signal.
	
	
	Recall that the signal-to-noise ratio parameter is defined by $\snr = \mu/\sigma$. In what follows, we will fix all parameters other than $\snr$, and analyze the principal's payoff as $\snr$ increases. Hence, we make explicit the dependence of any variable or function on $\snr$.\footnote{As is clear from the belief processes \eqref{eq:law-of-p} and \eqref{eq:dptagent} and HJBs \eqref{eq:HJBprincipal} and \eqref{eq:HJBagent}, in equilibrium $\mu$ and $\sigma$ always affect the players' incentives and payoffs through $\snr$. So writing equilibrium objects as functions of $\snr$ (while dropping $\mu$ and $\sigma$) is without loss.}
	
	As a benchmark, suppose that the principal never receives any information about the agent's type. Recall that $R(p)=pw_{NI}+(1-p)w_I$ and $p^{**}=R^{-1}(0)$. In this case, the principal would continue the relationship if $p<p^{**}$, and she would terminate the relationship at the first stopping opportunity if $p>p^{**}$. This leads to the following \textbf{``no-information value function"} for the principal:  
	\[
	\underline{W}(p):=\frac{\lambda}{r_2+\lambda}\max\{0,pw_{NI}+(1-p)w_I\}=\frac{\lambda}{r_2+\lambda}\max\{0,R(p)\}.
	\]
	Note that when $\snr$ is near $0$, the signal process is close to pure noise regardless of the agent's action. So we have the following observation.
	
	\begin{observation}
		$\lim_{\snr\to 0}W(p_0;\snr)=\underline{W}(p_0)$ and $\lim_{\snr\to 0}p^*(\snr)=p^{**}$.
	\end{observation}

	Now suppose that the agent's type is exogenously and immediately revealed to the principal. In this case, the principal will obtain her highest possible payoff for each belief, summarized by her \textbf{``full-information value function"}:
	\[
	\overline{W}(p):=\frac{\lambda}{r_2+\lambda}\left[p\mbox{ }\max\left\{0,w_{NI}\right\}+(1-p)\max\left\{0,w_{I}\right\}\right]=\frac{\lambda}{r_2+\lambda}pw_{NI}.
	\]
	Note that the principal's equilibrium payoff is always strictly between $\underline{W}(p_0)$ and  $\overline{W}(p_0)$. This is because some learning will take place in equilibrium (as the noninvestible type never fully mimics), but the agent's type is not immediately revealed (as $\snr<\infty$).
	\begin{observation}
		For all $\snr\in (0,\infty)$ and $p_0\in (0,1)$, $W(p_0;\snr)\in (\underline{W}(p_0),\overline{W}(p_0))$.
	\end{observation}

	The next theorem characterizes the limiting behavior of the principal's payoff as $\snr$ goes to infinity. 

	\begin{theorem}\label{t:noisehelps} 
		Letting $\tilde{\lambda}:=r_1\left(\frac{c}{u}\right)$,
		\begin{enumerate}
			\item If $\lambda<\tilde{\lambda}$, then 
			$\lim_{\snr\to\infty}||W(\cdot;\snr)-\overline{W}(\cdot)||_{\infty}=0$ and $\lim_{\snr\to\infty}p^*(\snr)=1$.
			
			\item If $\lambda>\tilde{\lambda}$, then 
			$\lim_{\snr\to\infty}||W(\cdot;\snr)-\underline{W}(\cdot)||_{\infty}=0$ and $\lim_{\snr\to\infty}p^*(\snr)=p^{**}$.
		\end{enumerate}
	\end{theorem}
	
	{This result is heavily driven by the agent's equilibrium behavior.} When the stopping opportunity arrives slowly ($\lambda<\tilde{\lambda}$), the noninvestible type's incentives to mimic are not strong. Intuitively, the relationship is relatively stable from the agent's viewpoint because, due to the lack of stopping opportunity, it will take a long time for the relationship to end even if the principal has decided to terminate it. As a result, the agent's equilibrium action is bounded away from ``full mimicking" for all $\snr$.\footnote{For sufficiently small $\lambda$, the noninvestible type does not mimic at all (i.e., $a(p)=0$ for all $p$), and the principal's problem becomes identical to a standard two-armed bandit problem. For relatively large $\lambda$ (still less than $\tilde{\lambda}$), some mimicking appears in equilibrium, but $a(\cdot)$ is still uniformly bounded away from $1$ for all $\snr$.} As $\snr$ grows without bound, the public signal becomes increasingly informative about the agent's type, and in the end, the agent's type is almost immediately revealed. Thus, the principal can afford to wait until being very certain that the agent is noninvestible, and her equilibrium value function converges to $\overline{W}$ (see Figure \ref{fig:Wconverge}, left panel).

	On the other hand, when the stopping opportunity arrives fast ($\lambda>\tilde{\lambda}$), the noninvestible type has stronger incentives to mimic the investible type. In particular, as $\snr$ increases without bound, the equilibrium mimicking intensity at the termination cutoff converges to 1. The speed of this convergence is so fast that the variance of the belief process vanishes there (i.e., $p^{*}$ becomes an almost absorbing state). Meanwhile, the equilibrium policy function $a(\cdot)$  converges to 1 also for $p<p^{*}$, and it converges to a function that is strictly less than 1 for $p>p^{*}$. In both of these regions, the principal will learn some information about the agent's type from the public signal. However, this information is not useful (payoff-relevant) for the principal since it does not lead to an action change. Hence, the principal's termination cutoff converges to $p^{**}$ and her equilibrium value function converges to $\underline{W}$---as if no information would ever arrive (see Figure \ref{fig:Wconverge}, right panel).
	
	\begin{figure}[ht]
		\centering
		\begin{subfigure}[t]{.49\linewidth}
			\includegraphics[width=1\linewidth]{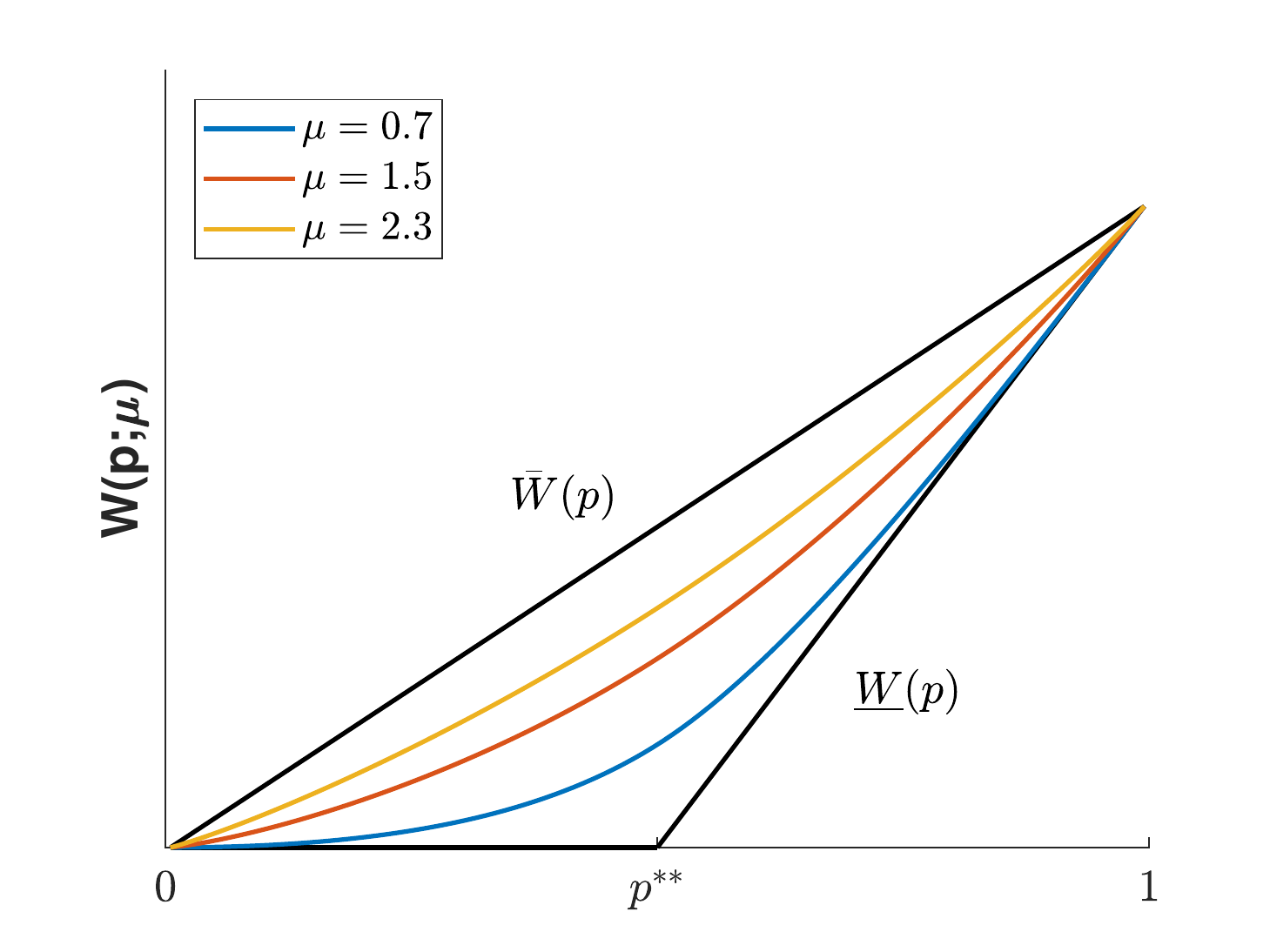}
			\caption{$\lambda<\tilde{\lambda}$}
		\end{subfigure}
		\begin{subfigure}[t]{.49\linewidth}
			\includegraphics[width=1\linewidth]{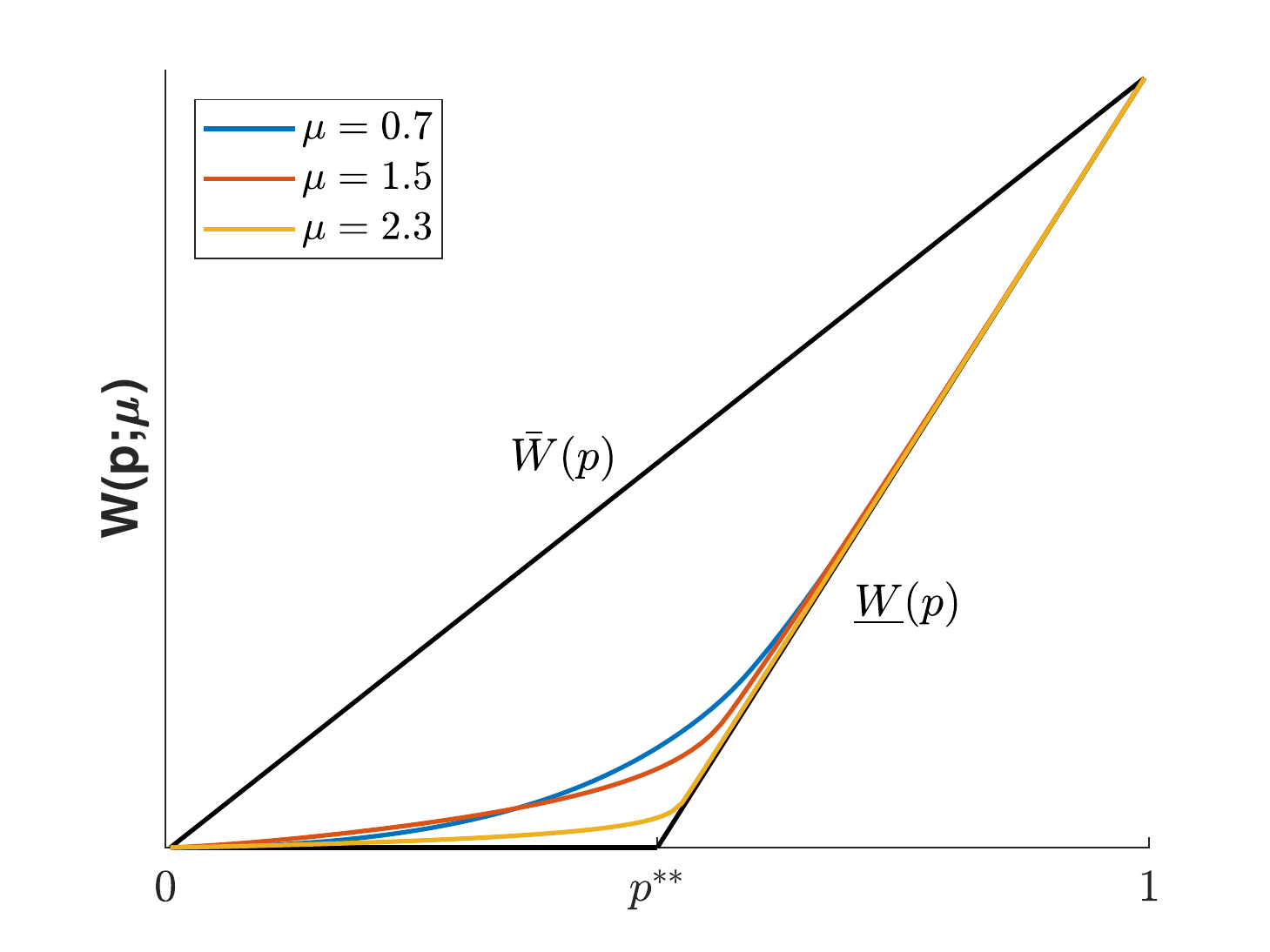}
			\caption{$\lambda>\tilde{\lambda}$}
		\end{subfigure}
		\caption{Convergence of the Principal's Equilibrium Value Function. This figure is plotted under the following parameter values: $r_1=0.5$, $r_2=0.5$, $u=1$, $c=1$, $w_{NI}=1$, $w_I=-1$ (so that $p^{**}=0.5$, $\tilde{\lambda}=0.25$); $\lambda=0.1$ for (a), $\lambda=2$ for (b).}
		\label{fig:Wconverge}
	\end{figure}
	
	To better understand the limiting equilibrium dynamics when $\lambda>\tilde{\lambda}$, note from equation \eqref{eq:law-of-p} that the diffusion coefficient of the principal's belief process at time $t$ is proportional to $\snr(1-a_t)$. For any fixed $p$, we can write it as $$
	\underbrace{\snr}_{\text{direct effect}} [1-\underbrace{a(p;\snr)}_{\text{equilibrium effect}}].$$ As $\snr$ increases, the \textit{direct effect} through the multiplier accelerates information revelation while the \textit{equilibrium effect} through the agent's strategy slows down learning. Its limit depends on the value of $p$; in particular, one can show that 
	\begin{align*}
		&\lim_{\snr\rightarrow\infty}\snr[1-a(p;\snr)] \in (0,\infty), \text{ for all }p\in (0,p^{**}),\\
		&\lim_{\snr\rightarrow\infty}\snr[1-a(p;\snr)]=\infty, \text{ for all }p\in (p^{**},1),\\
		&\lim_{\snr\rightarrow\infty}\snr[1-a(p^*(\snr);\snr)] = 0.
	\end{align*}
	This suggests the following limiting equilibrium dynamics. If the prior belief is above $p^{**}$, there is an immediate split of belief to either $1$ or very close to $p^{**}$ and then it (almost) stops moving. If the prior belief is below $p^{**}$, then learning takes place gradually but becomes slower and slower as the posterior approaches $p^{**}$. In both cases the principal's learning does not stop despite the agent's extreme manipulation, but because her posterior never moves across $p^{**}$, the principal's action never changes with the information she learns, and so payoff-wise it is as if no information ever arrives.
	
	The next two corollaries follow immediately from Theorem \ref{t:noisehelps}.
	\begin{corollary}\label{cor:noisehelps}
		For any $\lambda>\tilde{\lambda}$ and $p_0\in (0,1)$, there exist $\snr_1,\snr_2$ such that $\snr_1>\snr_2$ and $W(p_0;\snr_1)<W(p_0;\snr_2)$.
	\end{corollary}
	Corollary \ref{cor:noisehelps} illustrates that a more transparent performance measure (higher $\snr$) can sometimes reduce the principal's payoff. 
	This result indicates that VCs can sometimes fare better when the startup operates in a more volatile environment. Moreover, policies that require the startup to disclose too precise information may end up hurting the investors. 
	Intuitively, when the signal-to-noise ratio is very high, the mimicking incentives of the noninvestible type can be so strong that the public signal provides little useful information to the principal about the agent's type. In that case, introducing more noise to the public signal can lessen such perverse incentives of the noninvestible type, which leads to a more informative equilibrium signal process, benefiting the principal.\footnote{The possibility that better monitoring/more transparency may hurt a principal or relationship has appeared in other settings, such as career-concern models \citep{holmstrom1999,dewatripont1999}, contracting in insurance markets \citep{hirshleifer1971,schlee2001}, contracting with moral hazard \citep{zhu2020}, and dynamic team production \citep{Cetemen2020,Ramos2020}. In our model, this effect shows up for a different reason: better monitoring may give stronger incentives to the agent to engage in performance boosting, which depresses the informativeness of the public signals.} This result also suggests that there is usually a strictly positive but finite level of transparency that maximizes the principal's value.
	
	\begin{corollary}\label{cor:frictionhelps}
		For some large $\snr$ and any $p_0\in (0,1)$, there exists $\lambda_1$, $\lambda_2$ such that  $\lambda_1>\lambda_2$ and $W(p_0;\snr,\lambda_1)<W(p_0;\snr,\lambda_2)$.
	\end{corollary}
	Corollary \ref{cor:frictionhelps} demonstrates that more frictions in the principal's decision making (i.e., less frequent arrivals of the stopping opportunity) may sometimes improve the principal's equilibrium payoff. Intuitively, such frictions instill some commitment to not terminating the relationship too soon in the principal's behavior, which, similar to before, weakens the noninvestible type's incentive to manipulate the signals. In turn, this indirect effect through the agent allows the signal process to provide more information to the principal and possibly increases her equilibrium payoff. Because of this, the principal does not always have an incentive to reduce the frictions that prevent prompt decisions.

	\subsection*{Explanation of the Proof}
	For interested readers, we provide below an explanation of our proof of Theorem \ref{t:noisehelps} and highlight our technical contributions. This subsection can be skipped without interrupting the flow of the paper.
	
	It is technically difficult to analyze players' behavior and payoffs as $\snr \rightarrow \infty $ because
	incentives change discontinuously at the cutoff $p^{\ast }$ in the limit
	problem. This discontinuity makes the analysis of the semi-closed-form solution quite elusive and obliges to develop a new set of methods to deal with asymptotics. In particular, our analysis leverages three key elements:
	\begin{enumerate}
		\item[(i)] A Learning Lemma, which stipulates a ``zero-one law" between the agent's mimicking behavior and the principal's learning;
		\item[(ii)] The semi-closed-form solution of the unique equilibrium;
		\item[(iii)] The martingale property of the belief process.
	\end{enumerate}
	
	Let us begin by stating the Learning Lemma.
	\begin{learning-lemma}[Claim \ref{lem:learninglemma2}]
		Fix any prior \(p_{0}\in \left( 0,1\right)\),   any $r_1>0$, and some open interval $S=(\underline{s},\bar{s})$ such that $0<\underline{s}<p_0<\bar{s}<1$. For each \(\snr >0\), let \(a_{\snr}\left( \cdot \right)\) be the policy function induced by some Markovian strategy $\alpha_\snr$. Take any \(\varepsilon >0\) and let \(\mathbb{\bar{T}}\) be 
		the first time at which $p_t$ reaches the boundary of $S$. 
		Then we have:
		\[
		\limsup_{\snr\uparrow \infty}\mbox{ }\mathbb{E}^{NI}\left\{ r_1\int_{0}^{\mathbb{\mathbb{\bar{T}}}}e^{-r_1t}\mathbb{I}_{\left\{a _{\snr}\left( p_{t}\right) \leq 1-\varepsilon \right\} }dt\right\} =0.
		\]
	\end{learning-lemma}
	The lemma says that as $\snr\rightarrow \infty$, little time will be spent on those beliefs under which the noninvestible type is not mimicking intensively (i.e. those $p$ such that $a_\snr(p)<1-\epsilon$). In other words, for any interval of beliefs, either the noninvestible type mimics (almost) fully, or the principal's learning is extremely fast so that the posterior moves out of the interval in (almost) no time.\footnote{The Learning Lemma also holds conditional on the investible type, 
		and (thus) it holds unconditionally from the principal's perspective.} 
	This Learning Lemma is an important input to the proof of our asymptotic results, but it may also be of independent interest itself. Indeed, it holds for \textit{any} Markovian strategy, not only the equilibrium one, and a variation of this lemma will be used to prove our result in the next section.\footnote{This Learning Lemma resonates with the learning result in \cite{cripps2004}. Their result states that when the monitoring technology is fixed, in the long run, either the agent's type is revealed or the agent's behavior converges to full mimicking. Our Learning Lemma (proved in a continuous-time setting) requires the noise in monitoring to get arbitrarily small, but its relevance is not limited to the long run.}
	
	We now provide an explanation of our proof of Theorem 3 in the low-friction case where $\lambda >r_{1}\left( \frac{u}{c}\right)$. The proof proceeds in three steps.
	
	In \textbf{Step 1}, we show that $u$ is an upper bound for the agent's payoff at any belief $p\in (0,1)$ as $\snr \rightarrow \infty $, and this bound is tight at beliefs below the termination cutoff. Since the agent's value function $v(p)
	$ is decreasing, it suffices to prove the result for $p$ small. As the agent
	loses reputation on average, we can use our semi-closed-form solution to prove that for $%
	p<p^{\ast }$ the agent's payoff in a mixing region is bounded above by 
	\[
	u+(1-a(p))c.
	\]%
	We then note that $u$ is an upper bound for the agent's payoff at $p^{\ast}$ because the mimicking
	intensity at $p^{\ast}$ converges to $1$. Indeed, if this was not the case
	then the agent would be revealing information very fast as $\snr \rightarrow
	\infty ,$ yielding a payoff no greater than $\frac{\lambda }{\lambda +r_{1}}%
	(u+c)$ which, in the low-friction case, is smaller than $u$ by assumption. The bound at $p^{\ast }$ is then used to bound the agent's limit payoff at any (fixed) $p\in \left( 0,p^{\ast }\right) $. Indeed, the Learning Lemma tells us that either the agent mimics with (almost) full intensity, or the posterior reaches $p^*$ (almost) immediately, and in both cases his payoff is no greater than $u$. It turns out that this upper bound is tight and the agent's payoff converges to $u$ as $\snr \rightarrow \infty $ at any belief in $\left( 0,p^{\ast }\right)$. Indeed, should his limit payoff fall short of $u$, the noninvestible type would have too strong mimicking incentives which, in turn, would be incompatible with the indifference condition in the mixing region.\footnote{The martingale property of the belief process implies that by fully mimicking the investible type, the agent can lengthen the expected duration of the game while guaranteeing a flow payoff of $u$.}

	In \textbf{Step 2}, we prove that starting from any belief below $p^*$, the \textit{expected discounted duration} of the game, $1-\mathbb{E}^{NI}\left[ e^{-r_1\mathbb{T}}\right]$, converges to $1$ as $\snr
	\rightarrow \infty$; roughly speaking, this means that the game tends to continue for a very long time. To establish this, recall first that the Learning Lemma tells us that during the game (almost) all time will be spent on beliefs where the agent (almost) fully mimics. This implies that at almost all time the agent is getting a flow payoff of $u$. But since the agent's value converges to $u$, as shown in Step 1, the expected duration of the game has to converge to $1$.   
	We remark that since $\mathbb{E}^{NI}\left[ e^{-r_1\mathbb{T}}\right] \rightarrow 0$ and the termination region is more likely reached by the noninvestible type (due to the martingale property of the belief process), we must also have $\mathbb{E}^{I}\left[e^{-r_1\mathbb{T}}\right]\rightarrow 0$ (and hence $\mathbb{E}^{{}}\left[e^{-r_2\mathbb{T}}\right] \rightarrow 0$).
	
	In \textbf{Step 3}, we establish the convergence of the principal's equilibrium cutoff and value function as $\snr \rightarrow \infty$. Recall that we always have $p^{\ast }\geq p^{\ast \ast }$, as the option value of waiting is always nonnegative. Assume that in the limit we had $p^{\ast }>$ $p^{\ast \ast }.$ Then take a posterior $p\in \left( p^{\ast \ast },p^{\ast }\right) $ and recall from Step 2 that $\mathbb{E}\left[ e^{-r_2\mathbb{T}}\right] \rightarrow 0$. Thus the principal's payoff must converge to zero at that posterior (because the principal's payoff from staying in the relationship forever is normalized to $0$). However, $p\in \left(
	p^{\ast \ast },p^{\ast }\right) $ implies that $R(p)>0,$ and hence the principal could profitably deviate by terminating the relationship in the next opportunity. So we conclude that $\lim_{\snr\to\infty}p^*(\snr)\rightarrow p^{**}$. Finally, Step 2 also implies that $W(p;\snr)\rightarrow 0$ for every $p<p^{\ast \ast}$. Since the limit value function $\lim_{\snr \rightarrow \infty}W(p;\snr )$ is convex, it agrees with $\underline{W}(p)$ in $\left[0,p^{\ast \ast }\right],$ and $\underline{W}(p)$ is affine in $\left[p^{\ast \ast },1\right] ,$ it follows that $\lim_{\snr \rightarrow \infty}W(p;\snr )=\underline{W}(p)$.

	\section{Information at the Patient Limit}\label{sec:patience}
	
	We now investigate the equilibrium outcomes as players get arbitrarily patient. The purpose is to see more clearly the role of patience in the principal's incentives to wait for more information and in the agent's incentives to engage in performance boosting.
	
	First, consider an extreme case where $r_1$ is constant and $r_2$ goes to 0 (i.e., the principal gets arbitrarily more patient than the agent). In this case, the agent's mimicking intensity (which is independent of $r_2$) stays bounded away from 1 everywhere, implying that the public signal always reveals some information about the agent's type. As the principal gets more patient, her marginal cost of waiting for new information becomes lower. Consequently, a patient principal will terminate the relationship only when $p$ is very high. Indeed, the termination cutoff converges to 1 and the principal's payoff converges to $\overline{W}$.
	
	Next, consider the other extreme case where $r_2$ is constant and  $r_1$ goes to 0 (i.e., the agent gets arbitrarily more patient than the principal). As the agent gets more patient, he cares more about staying in the relationship for long and less about the instantaneous mimicking cost. Thus, the noninvestible type has stronger incentives to mimic the investible type, and the equilibirum mimicking intensity approaches one at and below the termination cutoff. In the limit, the outcome is similar to the case for large $\snr$ and $\lambda$ characterized in the previous section. That is, it is as if no information ever arrives, with the termination cutoff converging to $p^{**}$ and the principal's value function converging to $\underline{W}$.

	What happens in between the two extreme cases? Take a sequence $\{r_{1,n},r_{2,n}\}_n$ of discount rates such that $r_{i,n}\to 0$ for both $i=1,2$ and  $\lim_n\frac{r_{2,n}}{r_{1,n}}=\chi\in(0,\infty)$. Consider a sequence of games along which all other parameters are fixed, and let $\{W_n,V_n\}_n$ be the corresponding sequence of value functions for the principal and the agent, respectively. The following theorem displays their limits.
	
	\begin{theorem}\label{t:patientlimit}
		$W_n(\cdot)$ converges uniformly to $\max\{0,R(\cdot)\}$,  
		and $V_n(\cdot)$ converges pointwise to $V^*(\cdot)$ which satisfies
		\[
		V^*(p):=\begin{cases} u, &\text{ if }p<p^{**}\\ 0, &\text{ if }{p>p^{**}} \end{cases}.
		\]
	\end{theorem}
	
	Theorem \ref{t:patientlimit} shows that if both players get arbitrarily patient at comparable rates, then it is as if the principal does not receive any information.\footnote{Recall that the ``no-information'' value function $\underline{W}(\cdot) = \frac{\lambda}{r_2+\lambda}\max\{0,R(\cdot)\}$, so $\max\{0,R(\cdot)\}$ is the limiting ``no-information'' value function as $r_2$ tends to $0$.} Along the sequence, both the agent's incentives to mimic and the principal's resolve to wait for more information get stronger, but it turns out that the former effect dominates the latter.
	
	We view this result as a strong manifestation of the ratchet effect in the patient limit of our model. Since the principal cannot commit to not using future information against the agent, the noninvestible type will engage in performance boosting with almost full intensity in order to maintain his reputation. In the end, no useful information is ever revealed, and the principal's lack of commitment hurts her in the most extreme way. In our applications, this result suggests that the use of other instruments, such as some form of commitment (e.g., setting a deadline and/or grace period), additional screening devices (e.g., performance-based investment levels and/or salaries), or huge fines that increase the expected cost of performance boosting, may be necessary to help the principal get more information.

	\section{Concluding Remarks}
	
	In this paper, we study a stopping game with asymmetric information where the performance measures that reflect the fundamental can be manipulated by an agent at a cost. Despite the model being stylized, we obtain rich equilibrium dynamics. Our model illustrates that inflated performance can coexist with growing suspicion about a project's viability. Our analysis also implies that too much transparency may hinder the principal's ability to learn, by encouraging excessive performance boosting. This result suggests that some noise in the monitoring technology may be beneficial for the principal. Furthermore, we find an extreme form of ratchet effect in the patient limit, precluding any useful learning. This happens because the principal lacks the commitment to refrain from using the information obtained during the relationship against the agent, giving a highly patient noninvestible type strong incentives to boost performance and maintain his reputation.
	
	Several ways to extend our analysis are worth mentioning. While our main focus is the adverse selection problem, in some settings moral hazard is a prominent issue. Thus, it would be interesting to allow the agent's action to directly influence the principal's payoff. Relatedly, in that setting the principal might want to use history-dependent flow payoffs to reward/punish the agent.
	Another possibility is to expand the choice set of the principal by allowing her to elevate the ``status" of the relationship, such as promoting the agent or upgrading the terms of financing. Finally, optimal contracting in this setting remains an open problem. 
	We leave all these aspects as interesting directions for future research.
	
	\small
	\singlespacing
	\bibliography{references}
	\bibliographystyle{aer}
	
	\newpage
	\appendix
	\section{Appendix: Proofs}
	\noindent\textbf{Remark.} Because only the noninvestible type of the agent has an active choice to make, whenever there is no confusion we simply refer to the ``noninvestible-type agent" as the agent.
	\subsection{Equilibrium Characterization: Toward a Proof of Theorem \ref{uniqueness}}\label{app:equilibrium}
	To establish Theorem \ref{uniqueness}, we use the results that the equilibrium belief process must have full support (Lemma \ref{lem:fullspan}), and that the principal's equilibrium strategy must have a cutoff structure (Lemma \ref{lemma:principalbestreply}). These two lemmas are proved in the Online Appendix. The main proof characterizes the agent's (pseudo-)best reply to any cutoff termination rule (Lemma \ref{lemma:optimality-agent}). Finally, we prove equilibrium existence and uniqueness using a fixed point argument.

	\subsubsection{Proof of Lemma \ref{lemma:optimality-agent}}
	In light of Lemma \ref{lemma:principalbestreply}, let us fix a cutoff termination rule of the principal. We define a new state variable $Z_t:=\log\frac{p_t}{1-p_t}$, which is a strictly increasing transformation of $p_t$. Note that $Z_t$ is defined on $(-\infty,\infty)$. 
	
	Given the principal's conjecture $a(\cdot)$ about the agent's policy function and the agent's actual policy function $\tilde{a}(\cdot)$, the law of motion of $p_t$ is given by \eqref{eq:dptagent}. By It\^{o}'s lemma, the law of motion of $Z_t$ is
	\begin{equation}\label{eq:dZt}
		dZ_t=\snr^2(1-a_t)\left[1-\tilde{a}_t-\tfrac{1}{2}(1-a_t)\right]dt-\snr(1-a_t)dB_t.
	\end{equation}
	
	Now, suppose that the principal uses a particular cutoff policy function $b$ with cutoff belief $p^*\in (0,1)$. Suppose also that the noninvestible type's policy function $a$ satisfies the conditions in Lemma \ref{lemma:optimality-agent}: $a(\cdot)$ is Lipschitz, $\sup_{p\in (0,1)}a(p)<1$, and it satisfies \eqref{eq:agent-optimality}, i.e., $a\in \argmax_{\tilde{a}\in \mathcal{P}}\hat{V}(p,\tilde{a},b;a)$; moreover, the resulting $V(p)=\hat{V}(p,a(p),b(p);a(p))$ is regular. \textbf{For brevity, we call any $a(\cdot)$ that satisfies these conditions a \textit{pseudo-best reply} to $b(\cdot)$.}\footnote{We call such $a(\cdot)$ pseudo-best reply because $b(\cdot)$ by itself does not lead to a well-defined strategy of the principal; the principal's interpretation of the observed signal into her posterior belief depends on (her conjecture of) the agent's strategy. The equilibrium condition that the principal's conjecture coincides with the agent's actual strategy is imposed as part of the definition of a pseudo-best reply.} Lipschitz continuity of $a(\cdot)$ implies that, for any control in $\mathcal{P}$ or $\mathcal{A}$,  the controlled process $p_t$ or $Z_t$ in the agent's problem always admits a unique strong solution.

	Let $Z_t$ be the new state variable and define $v(z):=V\left(\frac{e^z}{1+e^z}\right)$ and $z^*:=\log\frac{p^*}{1-p^*}$. Because we work with $Z_t$ most of the time in this appendix, \textit{we will write $a(z)$ to mean $a\left(\frac{e^z}{1+e^z}\right)$ whenever there is no confusion.} The HJB for the agent is\footnote{Since $v$ is regular, $v$ is $C^2$ except at possibly finite points. This HJB holds on any interval over which $v$ is $C^2$.}
	\begin{equation}\label{eq:HJBagent}
		\left[r_1+b(z)\lambda\right]v(z)=\max_{\tilde{a}\in [0,1]} r_1[u+(1-\tilde{a})c]+\snr^2[1-a(z)]\left[1-\tilde{a}-\tfrac{1}{2}(1-a(z))\right]v'(z)+\tfrac{1}{2}\snr^2[1-a(z)]^2v''(z).
	\end{equation}
	
	The following sequence of claims establishes some necessary properties of any pseudo-best reply $a(\cdot)$. 
	
	\begin{claim}\label{cl:aoptimal}
		$a(z)=1-\frac{r_1c}{\max\left\{r_1c,-\snr^2 v'(z)\right\}}$, for all $z\in (-\infty,\infty)$. 
	\end{claim}
	\begin{proof}
		Since the RHS of \eqref{eq:HJBagent} is affine in the choice variable, optimality requires that, for almost every $z\in \mathbb{R}$,\footnote{Since $\sup_{z\in \mathbb{R}}a(z)<1$ by definition of a pseudo-best reply (and by Lemma \ref{lem:fullspan}), the law of motion \eqref{eq:dZt} implies that  the distribution of $Z_t$ has full support for any $t>0$, i.e., $\mbox{supp}(Z_t)=\mathbb{R},\forall t>0$.}
		\begin{equation*}
			a(z)\begin{cases}
				=0, &\text{ if }r_1c+\snr^2[1-a(z)]v'>0\\
				\in [0,1] &\text{ if }r_1c+\snr^2[1-a(z)]v'=0\\
				=1, &\text{ if }r_1c+\snr^2[1-a(z)]v'<0\\
			\end{cases}.
		\end{equation*}
		This implies that, for almost every $z\in \mathbb{R}$, we have
		\begin{equation}\label{eq:aoptimal}
			a(z)=1-\frac{r_1c}{\max\left\{r_1c,-\snr^2 v'(z)\right\}}.
		\end{equation}
		Since both sides of \eqref{eq:aoptimal} are continuous in $z$ (recall that $a(\cdot)$ is Lipschitz by assumption, and $v(\cdot)$ is $C^1$ by assumption), we conclude that \eqref{eq:aoptimal} must hold for every $z\in \mathbb{R}$.\footnote{This is because any continuous function that is $0$ almost everywhere is equal to $0$ everywhere.}
	\end{proof}
	
	\begin{claim}\label{cl:La=1}
		Fix any $z_1<z_2\leq z^*$ and suppose that  $a(z)=0$ for all $z\in (z_1,z_2)$. Then,
		\begin{equation}\label{eq:vLa=1}
			v(z)=u+c + A_1 e^{\xi_L z} + A_2 e^{\xi'_L z}, \forall z\in (z_1,z_2)
		\end{equation}
		for some $A_1,A_2\in\mathbb{R}$, where $\xi_L>0>\xi'_L$ are the two roots of the characteristic equation $ \xi^2 +\xi = 2r_1/\snr^2.$
	\end{claim}
	\begin{proof}
		Since $b(z)=0$ and $a(z)=0$ for all $z\in (z_1,z_2)$, equation \eqref{eq:HJBagent} becomes
		\begin{equation*}
			r_1v(z)= r_1(u+c)+\frac{1}{2}\snr^2[v'(z)+v''(z)].
		\end{equation*}
		It is easy to verify that its general solution is given by \eqref{eq:vLa=1}.
	\end{proof}
	
	\begin{claim}\label{cl:Ra=1}
		Fix any $ z^*\leq z_1<z_2$ and suppose that $a(z)=0$ for all $z\in (z_1,z_2)$. Then,
		\begin{equation}\label{eq:vRa=1}
			v(z)=\frac{r_1}{r_1+\lambda}(u+c) + B_1 e^{\xi_R z} + B_2 e^{\xi'_R z}, \forall z\in (z_1,z_2)
		\end{equation}
		for some $B_1,B_2\in\mathbb{R}$, where $\xi_R<0<\xi'_R$ are the two roots of the characteristic equation $ \xi^2 +\xi = 2(r_1+\lambda)/\snr^2.$
	\end{claim}
	\begin{proof}
		Since $b(z)=1$ and $a(z)=0$ for all $z\in (z_1,z_2)$, equation \eqref{eq:HJBagent} becomes
		\begin{equation*}
			(r_1+\lambda)v(z)= r_1(u+c)+\frac{1}{2}\snr^2[v'(z)+v''(z)].
		\end{equation*}
		It is easy to verify that its general solution is given by \eqref{eq:vRa=1}.
	\end{proof}
	
	Now, let us denote by $\Phi$ and $\phi$ the CDF and PDF of the standard normal distribution, respectively.

	\begin{claim}\label{cl:La<1}
		Fix any $z_1<z_2\leq z^*$ and suppose that  $a(z)\in (0,1)$ for all $z\in (z_1,z_2)$. Then,
		\begin{equation}\label{eq:vLa<1}
			v(z)=u+\sqrt{\kappa_L}\Phi^{-1}(C_1e^{z}+C_2),\forall z\in (z_1,z_2)
		\end{equation}
		and
		\begin{equation}\label{eq:aL<1}
			a(z)=1+\frac{\sqrt{2r_1}}{\snr}\frac{\phi\left(\Phi^{-1}(C_1e^z+C_2)\right)}{C_1e^z}
		\end{equation}
		for some $C_1<0$ and $C_2\in \mathbb{R}$, where 
		$
		\kappa_L:=\tfrac{r_1c^2}{2\snr^2}.
		$
		
		Moreover, $a(z)$ is strictly increasing, or strictly decreasing, or first strictly decreasing and then strictly increasing on $(z_1,z_2)$. 
	\end{claim}
	
	\begin{proof}
		Fix any $z_1<z_2\leq z^*$ such that $a(z)\in (0,1)$  for all $z\in (z_1,z_2)$. Claim \ref{cl:aoptimal} implies that
		\begin{equation}\label{eq:a<1}
			a(z)=1+\frac{r_1c}{\snr^2v'(z)}, \forall z\in (z_1,z_2).
		\end{equation}
		Substituting \eqref{eq:a<1} into \eqref{eq:HJBagent} and setting $b(z)=0$, we have
		\begin{equation}\label{eq:ODEvLa<1}
			v(z)=u+\kappa_L\frac{v''(z)-v'(z)}{v'(z)^2}.
		\end{equation}
		It is easy to verify that its general solution is given by \eqref{eq:vLa<1}, and that the resulting $a(\cdot)$ implied by \eqref{eq:a<1} is given by \eqref{eq:aL<1}. Moreover, since $a(z)\in (0,1)$, we must have $v'(z)<0$, i.e., $C_1<0$.
		
		To analyze the monotonicity of $a(\cdot)$ on $(z_1,z_2)$, we first establish the following equality which links $a'(z)$ to $a(z)$: 
		\begin{equation}\label{eq:a+a'}
			1- a(z)-a'(z) = 2\left(\frac{v(z)-u}{c}\right).
		\end{equation}
		By \eqref{eq:a<1},
		\[
		1-a(z) = - \frac{r_1c}{\snr^2v'(z)}.
		\]
		Differentiating this expression, we obtain
		\[
		-a'(z) = \frac{r_1c }{\snr^2}\frac{v''(z)}{v'(z)^2} = - \frac{v''(z)}{v'(z)} [1-a(z)].
		\]
		Recall, from the agent's HJB \eqref{eq:ODEvLa<1} in this case, that
		\[
		\frac{v''(z)}{v'(z)} = 1 + \frac{\left[v(z)-u\right] v'(z)}{\kappa_L}  = 1 - 2\left(\frac{v(z)-u}{c}\right)\frac{1}{1-a(z)}, 
		\]
		where the second equality follows from \eqref{eq:a<1} and $\kappa_L=\frac{r_1c^2}{2\snr^2}$. Equation \eqref{eq:a+a'} then follows immediately.
		
		Equation \eqref{eq:a+a'} implies that
		\begin{enumerate}
			\item[(i)] If $a'(\tilde{z})<0$ for some $\tilde{z}\in (z_1,z_2)$, then $a'(z)<0$ for all $z\in (z_1,\tilde{z})$. 
			\item[(ii)] There does not exist an interval $I\subseteq (z_1,z_2)$ s.t. $a'(z)=0$ for all $z\in I$.
			\item[(iii)] If $a(\tilde{z})\geq 0$ for some $\tilde{z}\in (z_1,z_2)$, then $a(\cdot)$ is strictly increasing on $(\tilde{z},z_2)$.
			
		\end{enumerate}
		
		To see (i), suppose that $a'(\tilde{z})<0$ for some $\tilde{z}\in (z_1,z_2)$. Since $a(\cdot)$ given by \eqref{eq:aL<1} is a smooth function on $(z_1,z_2)$, we can define 
		$\underline{z}=\inf\left\{z\in [z_1,\tilde{z}):a'(\cdot)|_{(z,\tilde{z}]}<0\right\}.
		$ 
		Result (i) is proved if $\underline{z}=z_1$. Suppose (for a contradiction) that $\underline{z}>z_1$. Continuity of $a'$ implies that $a'(\underline{z})=0$. Moreover, since $a'(z)<0$ for all $z\in (\underline{z},\tilde{z}]$, we have $a(\underline{z})>a(\tilde{z})$. Then,
		\[
		2\left(\frac{v(\underline{z})-u}{c}\right)=1-a(\underline{z})-a'(\underline{z})<1-a(\tilde{z})-a'(\tilde{z})=2\left(\frac{v(\tilde{z})-u}{c}\right)
		\]
		where the equalities follow from \eqref{eq:a+a'} and the strict inequality follows from $a(\underline{z})>a(\tilde{z})$ and $a'(\underline{z})=0>a'(\tilde{z})$. But this is a contradiction to $v(\underline{z})>v(\tilde{z})$ because $C_1<0$ and $v$ is strictly decreasing on $(z_1,z_2)$.

		To see (ii), suppose (for a contradiction) that there exists an interval $I\subseteq (z_1,z_2)$ s.t. $a'(z)=0$ for all $z\in I$. Then, the LHS of \eqref{eq:a+a'} is constant on $I$ while the RHS is strictly decreasing, a contradiction.
		
		To see (iii), suppose that $a'(\tilde{z})\geq 0$ for some $\tilde{z}\in (z_1,z_2)$. Then, we must have $a'(z)\geq 0$ for all $z\in (\tilde{z},z_2)$, for otherwise we would reach a contradiction to (i). Further, take any $z_3,z_4$ s.t. $\tilde{z}\leq z_3<z_4\leq z_2$. Since $a'(z)\geq 0$, we know that $a(z_3)\leq a(z_4)$. But this inequality must be strict, for otherwise $a(z_3)=a(z)=a(z_4)$ for all $z\in (z_3,z_4)$ contradicting Result (ii). So, $a(\cdot)$ is strictly increasing on  $(\tilde{z},z_2).$
		
		Finally, to establish the monotonicity of $a(\cdot)$, suppose first that $a'(z)\geq 0$ for all $z\in (z_1,z_2)$. The same argument for (iii) above proves that $a$ must be strictly increasing on $(z_1,z_2)$. Suppose now that $a'(\tilde{z})< 0$ for some $\tilde{z}\in (z_1,z_2)$. Let $Z\subseteq (z_1,z_2)$ be the largest interval containing $\tilde{z}$ s.t. $a'(z)<0$ for all $z\in Z$. Result (i) immediately implies that $\inf Z = z_1$. If $\sup Z = z_2$, then $a(\cdot)$ is strictly decreasing on $(z_1,z_2)$. If $\sup Z <z_2$, continuity of $a'$ implies that $a'(\sup Z) = 0$. Then, Result (iii) implies that $a(\cdot)$ is strictly increasing on $(\sup Z, z_2)$. In summary, $a(\cdot)$ is either strictly increasing, or strictly decreasing, or first strictly decreasing and then strictly increasing on $(z_1,z_2)$. 
	\end{proof}
	
	\begin{claim}\label{cl:Ra<1}
		Fix any $z^*\leq z_1<z_2\ $ and suppose that $a(z)\in (0,1)$ for all $z\in (z_1,z_2)$. Then,
		\begin{equation}\label{eq:vRa<1}
			v(z)=\frac{r_1}{r_1+\lambda}u+\sqrt{\kappa_R}\Phi^{-1}(D_1e^{z}+D_2),\forall z\in (z_1,z_2)
		\end{equation}
		and
		\begin{equation}\label{eq:aR<1}
			a(z)=1+\frac{\sqrt{2(r_1+\lambda)}}{\snr}\frac{\phi\left(\Phi^{-1}(D_1e^z+D_2)\right)}{D_1e^z}
		\end{equation}
		for some $D_1<0$ and $D_2\in \mathbb{R}$, where 
		$
		\kappa_R:=\tfrac{r_1^2c^2}{2(r_1+\lambda)\snr^2}.
		$
		
		Moreover, $a(z)$ is strictly increasing, or strictly decreasing, or first strictly decreasing and then strictly increasing on $(z_1,z_2)$.
	\end{claim}
	\begin{proof}
		The idea of this proof is completely analogous to that of Claim \ref{cl:La<1}. Fix any $z^*\leq z_1<z_2$ such that $a(z)\in (0,1)$  for all $z\in (z_1,z_2)$. In this case, equation \eqref{eq:a<1} still holds. Substituting it into \eqref{eq:HJBagent} and setting $b(z)=1$, we have
		\[
		v(z)=\frac{r_1}{r_1+\lambda}u+\kappa_R\frac{v''(z)-v'(z)}{v'(z)^2}.
		\]
		It is easy to verify that its general solution is given by \eqref{eq:vRa<1}, and that the resulting $a$ implied by \eqref{eq:a<1} is given by \eqref{eq:aR<1}. Moreover, since $a(z)\in (0,1)$, we must have $v'(z)<0$, i.e., $D_1<0$.
		
		Analogous to \eqref{eq:a+a'}, the following equation links $a(z)$ to $a'(z)$ in this case:
		\begin{equation}\label{eq:a+a'R}
			1- a(z)-a'(z) =2\left(\frac{v(z)-\frac{r_1}{r_1+\lambda}u}{c}\right)\left(\frac{r_1+\lambda}{r_1}\right). 
		\end{equation}
		The proof of equation \eqref{eq:a+a'R} and the monotonicity of $a(\cdot)$ is along the same lines as in the proof of Claim \ref{cl:La<1}, and is thus omitted.
	\end{proof}
	
	\begin{claim}\label{cl:zbounded}
		If \(a(\tilde{z})>0\), then there are \(z_L, z_R \) such that \(z_L < \tilde{z} < z_R\) and \(a(z_L)=a(z_R)=0\). 
	\end{claim}
	\begin{proof}
		Let \(a(\tilde{z})>0\) and suppose, seeking a contradiction, that \(a(z)>0\) for all \(z>\tilde{z}\). Then, for all such \(z\), we would have (by Claim \ref{cl:aoptimal})
		$
		v'\left( z\right) =-\frac{r_1c}{\snr^2 [1-a \left( z\right)] }\leq -\frac{r_1c}{\snr^2}.
		$ 
		Taking the limit for arbitrary large $z$, we obtain \(\lim_{z \to +\infty}v\left(z\right) =-\infty\), a contradiction as \(v\) is always nonnegative. Similarly, if \(a(z)>0\) for all \(z<\tilde{z}\), then \(\lim_{z \to -\infty}v\left(z\right) =+\infty\), which contradicts that \(v\) is bounded above by \(u+c\).
	\end{proof}
	
	\begin{claim}\label{cl:zconnected}
		For any  \(z_1,z_2 \) such that either \(z_1<z_2\leq z^*\) or \(z^*\leq z_1<z_2\), \(a(z_1)=a(z_2)=0\) implies \(a(z)=0\) for all $z\in [z_1,z_2]$. 
	\end{claim}
	
	\begin{proof}
		Fix any $z_1<z_2\leq z^*$ such that $a(z_1)=a(z_2)=0$. Suppose (for a contradiction) that there exists $\tilde{z}\in (z_1,z_2)$ s.t. $a(\tilde{z})\in (0,1)$. Let $Z$ be the largest interval containing $\tilde{z}$ such that $a(z)\in (0,1)$ for all $z\in Z$. Obviously, $z_1\leq \inf Z<\sup Z\leq z_2$, and $a(\inf Z)=a(\sup Z)=0$ because $a(\cdot)$ is continuous. By Claim \ref{cl:La<1}, $a(\cdot)$ is strictly increasing, or strictly decreasing, or first strictly decreasing and then strictly increasing on $Z$. Since $a(\inf Z)=0$, $a(\cdot)$ can only be strictly increasing on $Z$, but this contradicts the continuity of $a(\cdot)$ at $\sup Z$. An analogous argument which invokes Claim \ref{cl:Ra<1} establishes same result for any $z^*\leq z_1<z_2$ such that $a(z_1)=a(z_2)=0$.
	\end{proof}

	\begin{corollary}\label{cor:constantorhump}
		One of the following must hold for any pseudo-best reply $a(\cdot)$:
		\begin{itemize}
			\item $a(z)=0$ for all $z\in \mathbb{R}$;
			\item $a(z)$ is hump-shaped (and maximized at $z^*$).
		\end{itemize}
	\end{corollary}
	\begin{proof}
		Suppose that $a(\cdot)$ is not always equal to $0$. Then there exists $\tilde{z}$ s.t. $a(\tilde{z})>0$. Let $Z$ be the largest interval containing $\tilde{z}$ such that $a(z)>0$ for all $z\in Z$. Let $z_L=\inf Z$ and $z_R=\sup Z$. By Claim \ref{cl:zbounded}, $-\infty<z_L<z_R<\infty$. By continuity of $a$, $a(z_L)=a(z_R)=0$. Then we must have $z^*\in (z_L,z_R)$, for otherwise we would reach a contradiction to Claim \ref{cl:zconnected}. Moreover, for any $z\in (z_L,z_R)^c$, we must have $a(z)=0$, for otherwise we can construct another $Z'$ which also contains $z^*$ such that $Z'-Z\neq \emptyset$, contradicting the maximality of $Z$. Since $a(z_L)=0$,  Claim \ref{cl:La<1} implies that $a(\cdot)$ must be strictly increasing on $(z_L,z^*)$. Similarly, since $a(z_R)=0$, Claim \ref{cl:Ra<1} implies that $a$ must be strictly decreasing on $(z^*,z_R)$. 
		
		In summary, if $a(\cdot)$ is not always equal to $0$, then there exist $z_L<z_R$ s.t. $a(\cdot)=0$ on  $(-\infty,z_L)\cup (z_R,\infty)$, $a(\cdot)$ is strictly increasing on $(z_L,z^*)$, and strictly decreasing on $(z^*,z^R)$; that is, $a(\cdot)$ is hump-shaped and maximized at $z^*$.
	\end{proof}
	
	\begin{claim}\label{cl:ashaper}
		$a(\cdot)$ is hump-shaped if and only if $r_1<r^*$, where $r^*$ is the unique solution to \eqref{eq:r*}.
	\end{claim}
	\begin{proof}
		\textbf{(``Only if" part)} Suppose first that $a(\cdot)$ is hump-shaped. We will show that this implies $r_1<r^*$. By Definition \ref{def:humpshaped}, there exist $z_L,z_R$ s.t. $-\infty<z_L<z^*<z_R<\infty$ such that $a(z)=0$ on $(-\infty,z_L]\cup [z_R,\infty)$ and $a(z)\in (0,1)$ on $(z_L,z_R)$. Let 
		\begin{equation*}
			v^*:=v(z^*),\quad
			v_L:=v(z_L),\quad
			v_R:=v(z_R).
		\end{equation*}
		
		We now calculate the undetermined coefficients in Claims \ref{cl:La=1} to \ref{cl:Ra<1} ($A_1,A_2,B_1,B_2,C_1,C_2,D_1,D_2$ and $z_R,z_L,v_R,v_L$) in various parts of the agent's policy and value, as functions of $v^*$ and model parameters. First, consider $z<z^*$. As $z\to-\infty$, the agent's value function is given in Claim \ref{cl:La=1} by \eqref{eq:vLa=1}. Because $v(\cdot)$ is bounded, we must have 
		\begin{equation}\label{eq:A2}
			A_2=0.
		\end{equation}
		Note that Claim \ref{cl:aoptimal} implies that $r_1c=-\snr^2v'(z_L)$, that is, 
		$v'(z_L)=-\frac{r_1c}{\snr^2}.$ 
		By Claims \ref{cl:La=1} and \ref{cl:La<1}, the value function $v(\cdot)$ must satisfy the following value-matching and smooth-pasting conditions at $z_L$ and $z^*$:
		\begin{align*}
			u+c+A_{1}e^{\xi_{L}z_{L}}\phantom{.}&=\phantom{.}v_{L}, \tag{value-matching at $z_L$}\\
			u+\sqrt{\kappa_{L}}\Phi^{-1}\left(C_{1}e^{z_{L}}+C_{2}\right)\phantom{.}&=\phantom{.}v_{L},\tag{value-matching at $z_L$}\\
			u+\sqrt{\kappa_{L}}\Phi^{-1}\left(C_{1}e^{z^{*}}+C_{2}\right)\phantom{.}&=\phantom{.}v^{*}, \tag{value-matching at $z^*$}\\
			A_{1}\xi_{L}e^{\xi_{L}z_{L}}\phantom{.}&=\phantom{.}-\tfrac{r_1c}{\snr^2},\tag{smooth-pasting at $z_L$}\\
			\tfrac{\sqrt{\kappa_{L}}C_{1}e^{z_{L}}}{\phi\left(\frac{v_L-u}{\sqrt{\kappa_{L}}}\right)}\phantom{.}&=\phantom{.}-\tfrac{r_1c}{\snr^2}. \tag{smooth-pasting at $z_L$}
		\end{align*}
		These five conditions can uniquely pin down the undetermined vector $(v_L,z_L,A_1,C_1,C_2)$ as:
		\begin{align}
			v_{L}&=u+c-\frac{r_{1}c}{\xi_{L}\snr^{2}},\label{eq:vL}\\
			e^{z_{L}}&=e^{z^{*}}\left[\frac{\frac{\sqrt{2r_{1}}}{\snr}\phi\left(\frac{v_{L}-u}{\sqrt{\kappa_{L}}}\right)}{\frac{\sqrt{2r_{1}}}{\snr}\phi\left(\frac{v_{L}-u}{\sqrt{\kappa_{L}}}\right)+\Phi\left(\frac{v_{L}-u}{\sqrt{\kappa_{L}}}\right)-\Phi\left(\frac{v^{*}-u}{\sqrt{\kappa_{L}}}\right)}\right],\label{eq:zL}\\
			A_{1}&=-\frac{r_{1}c}{\xi_{L}\snr^{2}}e^{-\xi_Lz_L}
			\label{eq:A1}\\
			C_{1}&=e^{-z^{*}}\left[\Phi\left(\frac{v^{*}-u}{\sqrt{\kappa_{L}}}\right)-\Phi\left(\frac{v_{L}-u}{\sqrt{\kappa_{L}}}\right)-\frac{\sqrt{2r_{1}}}{\snr}\phi\left(\frac{v_{L}-u}{\sqrt{\kappa_{L}}}\right)\right],\label{eq:C1}\\
			C_{2}&=\Phi\left(\frac{v_{L}-u}{\sqrt{\kappa_{L}}}\right)+\frac{\sqrt{2r_{1}}}{\snr}\phi\left(\frac{v_{L}-u}{\sqrt{\kappa_{L}}}\right).\label{eq:C2}
		\end{align}
		
		Now, consider $z>z^*$. As $z\to\infty$, the agent's value function is given in Claim \ref{cl:Ra=1} by \eqref{eq:vRa=1}. Because $v(\cdot)$ is bounded, we must have 
		\begin{equation}\label{eq:B2}
			B_2=0.
		\end{equation} 
		Note that Claim \ref{cl:aoptimal} implies that $r_1c=-\snr^2v'(z_R)$, that is, 
		$v'(z_R)=-\frac{r_1c}{\snr^2}.$ 
		By Claims \ref{cl:Ra=1} and \ref{cl:Ra<1}, the value function $v(\cdot)$ must satisfy the following value-matching and smooth-pasting conditions at $z_R$ and $z^*$:
		\begin{align*}
			\left(u+c\right)\frac{r_{1}}{r_{1}+\lambda}+B_{1}e^{\xi_{R}z_{R}}\phantom{.}&=\phantom{.}v_{R}, \tag{value-matching at $z_R$}\\
			\tfrac{r_{1}}{r_{1}+\lambda}u+\sqrt{\kappa_{R}}\Phi^{-1}\left(D_{1}e^{z_{R}}+D_{2}\right)\phantom{.}&=\phantom{.}v_R, \tag{value-matching at $z_R$}\\
			\tfrac{r_{1}}{r_{1}+\lambda}u+\sqrt{\kappa_{R}}\Phi^{-1}\left(D_{1}e^{z^{*}}+D_{2}\right)\phantom{.}&=\phantom{.}v^{*}, \tag{value-matching at $z^*$}\\
			B_{1}\xi_{R}e^{\xi_{R}z_{R}}\phantom{.}&=\phantom{.}-\tfrac{r_1c}{\snr^2},\tag{smooth-pasting at $z_R$}\\
			\tfrac{\sqrt{\kappa_{R}}D_{1}e^{z_{R}}}{\phi\left(\frac{v\left(z_{R}\right)-\frac{r_{1}}{r_{1}+\lambda}u}{\sqrt{\kappa_{R}}}\right)}\phantom{.}&=\phantom{.}-\tfrac{r_1c}{\snr^2}. \tag{smooth-pasting at $z_R$}
		\end{align*}
		These five conditions can uniquely pin down the undetermined vector $(v_R,z_R,B_1,D_1,D_2)$ as:
		\begin{align}
			v_{R}&=\frac{r_{1}}{r_{1}+\lambda}\left(u+c\right)-\frac{r_{1}c}{\xi_{R}\snr^{2}},\label{eq:vR}\\
			e^{z_{R}}&=e^{z^{*}}\left[\frac{\frac{\sqrt{2\left(r_{1}+\lambda\right)}}{\snr}\phi\left(\frac{v_{R}-\frac{r_{1}}{r_{1}+\lambda}u}{\sqrt{\kappa_{R}}}\right)}{\frac{\sqrt{2\left(r_{1}+\lambda\right)}}{\snr}\phi\left(\frac{v_{R}-\frac{r_{1}}{r_{1}+\lambda}u}{\sqrt{\kappa_{R}}}\right)+\Phi\left(\frac{v_{R}-\frac{r_{1}}{r_{1}+\lambda}u}{\sqrt{\kappa_{R}}}\right)-\Phi\left(\frac{v^{*}-\frac{r_{1}}{r_{1}+\lambda}u}{\sqrt{\kappa_{R}}}\right)}\right],\label{eq:zR}\\
			B_{1}&=-\frac{r_{1}c}{\xi_{R}\snr^{2}}e^{-\xi_{R}z_R},\label{eq:B1}\\
			D_{1}&=e^{-z^{*}}\left[\Phi\left(\frac{v^{*}-\frac{r_{1}}{r_{1}+\lambda}u}{\sqrt{\kappa_{R}}}\right)-\Phi\left(\frac{v_{R}-\frac{r_{1}}{r_{1}+\lambda}u}{\sqrt{\kappa_{R}}}\right)-\frac{\sqrt{2\left(r_{1}+\lambda\right)}}{\snr}\phi\left(\frac{v_{R}-\frac{r_{1}}{r_{1}+\lambda}u}{\sqrt{\kappa_{R}}}\right)\right],\label{eq:D1}\\
			D_{2}&=\Phi\left(\frac{v_{R}-\frac{r_{1}}{r_{1}+\lambda}u}{\sqrt{\kappa_{R}}}\right)+\frac{\sqrt{2\left(r_{1}+\lambda\right)}}{\snr}\phi\left(\frac{v_{R}-\frac{r_{1}}{r_{1}+\lambda}u}{\sqrt{\kappa_{R}}}\right)\label{eq:D2}.
		\end{align}
		
		Given $v^*$ and model parameters, equations \eqref{eq:vLa=1} through \eqref{eq:aL<1}, \eqref{eq:vRa<1} and \eqref{eq:aR<1} with coefficients given by \eqref{eq:A2} through \eqref{eq:D2} fully determine the agent's policy function $a(\cdot)$ and his value function $v(\cdot)$ on $\mathbb{R}$. Since $a(z)\in (0,1)$ on $(z_L,z_R)$, Claim \ref{cl:aoptimal} implies that $v'(z)<0$ on $(z_L,z_R)$, which in turn implies that $v_L>v_R$. Note that $v_L$ and $v_R$ are given by \eqref{eq:vL} and \eqref{eq:vR}, both of which are independent of $v^*$. In particular, $v_L>v_R$ amounts to 
		$   u+c-\frac{r_{1}c}{\xi_{L}\snr^{2}}>\frac{r_{1}}{r_{1}+\lambda}\left(u+c\right)-\frac{r_{1}c}{\xi_{R}\snr^{2}}$. 
		Straightforward calculation shows that this is equivalent to 
		\begin{equation*}
			r_1(\sqrt{1+8r_1/\snr^2} +\sqrt{1+8(r_1+\lambda)/\snr^2}) + \lambda(\sqrt{1+8r_1/\snr^2}+1)<4\lambda\left(\tfrac{u}{c}+1\right),
		\end{equation*}
		that is, $r_1<r^*$ (see condition \eqref{eq:r*}). Therefore, we have shown that
		``$   a(\cdot) \text{ is hump-shaped } \Rightarrow r_1<r^*$", 
		so the ``only if" part of the claim is proved.
		
		\textbf{(``If" part)} Suppose now that $a(\cdot)$ is not hump-shaped. We will show that this implies $r_1\geq r^*$. By Corollary \ref{cor:constantorhump}, we know that $a(z)=0$ for all $z\in \mathbb{R}$. Then, by Claims \ref{cl:La=1} and \ref{cl:Ra=1}, $v(\cdot)$ is given by \eqref{eq:vLa=1} for $z<z^*$ and by \eqref{eq:vRa=1} for $z>z^*$. We now pin down the undetermined coefficients ($A_1,A_2,B_1,B_2$) as functions of model parameters. Since $v(\cdot)$ is bounded as $z\to-\infty$ or $z\to+\infty$, we must have
		\begin{align}
			A_2&=0,\label{eq:A2a=1}\\
			B_2&=0.\label{eq:B2a=1}
		\end{align}
		Also, the value function $v(\cdot)$ must satisfy the following value-matching and smooth-pasting condition at $z^*$:
		\begin{align*}
			u+c+A_{1}e^{\xi_{L}z^{*}}=\tfrac{r_{1}}{r_{1}+\lambda}\left(u+c\right)+B_{1}e^{\xi_{R}z^{*}},\tag{value-matching at $z^*$}\\
			A_{1}\xi_{L}e^{\xi_{L}z^{*}}=B_{1}\xi_{R}e^{\xi_{R}z^{*}}.\tag{smooth-pasting at $z^*$}
		\end{align*}
		These two conditions uniquely pin down $(A_1,B_1)$ as:
		\begin{align}
			A_{1}&=\tfrac{\xi_R}{\xi_L-\xi_R}\tfrac{\lambda}{r_1+\lambda}(u+c)e^{-\xi_{L}z^{*}},\label{eq:A1a=1}\\
			B_{1}&=\tfrac{\xi_L}{\xi_L-\xi_R}\tfrac{\lambda}{r_1+\lambda}(u+c)e^{-\xi_{R}z^{*}}.\label{eq:B1a=1}
		\end{align}
		Since $a(z)=0$ for all $z\in \mathbb{R}$, by Claim \ref{cl:aoptimal} we must have $r_1c\geq -\snr^2v'(z)$ for all $z\in \mathbb{R}$. In particular, this should hold at $z^*$: 
		$ r_1c\geq -\frac{\xi_R\xi_L}{\xi_L-\xi_R}\frac{\lambda}{r_1+\lambda}(u+c)\snr^2$.
		Straightforward calculation shows that this is equivalent to
		\begin{equation*}
			r_1(\sqrt{1+8r_1/\snr^2} +\sqrt{1+8(r_1+\lambda)/\snr^2}) + \lambda(\sqrt{1+8r_1/\snr^2}+1)\geq 4\lambda\left(\tfrac{u}{c}+1\right),
		\end{equation*}
		that is, $r_1\geq r^*$  (see condition \eqref{eq:r*}). Therefore, we have shown that
		``$a(\cdot) \text{ is not hump-shaped } \Rightarrow r_1\geq r^*$",
		so the ``if" part of the claim is also proved.
	\end{proof}
	
	In Claim \ref{cl:ashaper}, we find that if $r_1<r^*$, then $a(\cdot)$ is hump-shaped, in which case we can express all coefficients and cutoffs in the agent's policy and value functions in closed form with respect to $v^*:=v(z^*)$ and model parameters. However, $v^*$ itself is an endogenous object that needs to be determined. Our next claim, proved in the Online Appendix, paves the final way for establishing Lemma \ref{lemma:optimality-agent}.
	
	\begin{claim}\label{cl:v*unique}
		Suppose that $r_1<r^*$, and let $v_L$ and $v_R$ be given by \eqref{eq:vL} and \eqref{eq:vR}, respectively. Define $a_-^*,a_+^*:[v_R,v_L]\rightarrow \mathbb{R}$ by
		\begin{align}
			a_-^*(x)&:=1-\frac{\frac{\sqrt{2r_{1}}}{\snr}\phi\left(\frac{x-u}{\sqrt{\kappa_{L}}}\right)}{\frac{\sqrt{2r_{1}}}{\snr}\phi\left(\frac{v_{L}-u}{\sqrt{\kappa_{L}}}\right)+\Phi\left(\frac{v_{L}-u}{\sqrt{\kappa_{L}}}\right)-\Phi\left(\frac{x-u}{\sqrt{\kappa_{L}}}\right)},\label{eq:a*-}\\
			a_+^*(x)&:=1-\frac{\frac{\sqrt{2\left(r_{1}+\lambda\right)}}{\snr}\phi\left(\frac{x-\frac{r_{1}}{r_{1}+\lambda}u}{\sqrt{\kappa_{R}}}\right)}{\frac{\sqrt{2\left(r_{1}+\lambda\right)}}{\snr}\phi\left(\frac{v_{R}-\frac{r_{1}}{r_{1}+\lambda}u}{\sqrt{\kappa_{R}}}\right)+\Phi\left(\frac{v_{R}-\frac{r_{1}}{r_{1}+\lambda}u}{\sqrt{\kappa_{R}}}\right)-\Phi\left(\frac{x-\frac{r_{1}}{r_{1}+\lambda}u}{\sqrt{\kappa_{R}}}\right)}.\label{eq:a*+}
		\end{align}
		Then, $a_-^*(\cdot)$ is strictly decreasing on $[v_R,v_L]$ with $a_-^*(v_R)\in (0,1)$ and $a_-^*(v_L)=0$; $a_+^*(\cdot)$ is strictly increasing on $[v_R,v_L]$ with $a_+^*(v_R)=0$ and $a_+^*(v_L)\in (0,1)$. 
	\end{claim}
	\begin{proof}
		See Online Appendix.
	\end{proof}

	\begin{proof}[Proof of Lemma \ref{lemma:optimality-agent}]
		Suppose first that $r_1\geq r^*$. By Corollary \ref{cor:constantorhump} and Claim \ref{cl:ashaper}, any pseudo-best reply $a(\cdot)$ must be such that $a(z)=0$ for all $z\in \mathbb{R}$. Obviously, such function is unique. To verify that $a(z)=0$ for all $z\in \mathbb{R}$ is indeed a solution to \eqref{eq:agent-optimality}, note that we have shown in the proof of Claim \ref{cl:ashaper} that, together with the $v(\cdot)$ given by \eqref{eq:vLa=1} and \eqref{eq:vRa=1} and the coefficients given by \eqref{eq:A2a=1} through \eqref{eq:B1a=1}, it satisfies the agent's HJB equation \eqref{eq:HJBagent}.\footnote{Everywhere except at $z^*$ where $v''$ does not exist.} Since $v(\cdot)$ is bounded, we have
		$\lim_{t\rightarrow\infty}e^{-r_{1}t}\mathbb{E}\left[v\left(z_{t}\right)1_{\left\{ \tau\geq t\right\} }\right]=0$, 
		where $\tau$ is the stopping time when the relationship is terminated. Then by \citet[Theorem 3.3.5]{ross2018}, $a(z)=1$ for all $z\in \mathbb{R}$ is indeed a solution to \eqref{eq:agent-optimality}. 
		In addition, $v(\cdot)$ is regular because the functions given by \eqref{eq:vLa=1} and \eqref{eq:vRa=1} are smooth, and value-matching and smooth-pasting conditions are imposed at $z^*$.
		
		Suppose now that $r_1<r^*$. By Claim \ref{cl:ashaper}, any psuedo-best reply $a(\cdot)$ is hump-shaped. We first show that such function, if it exists, must be unique. By Claims \ref{cl:La=1} to \ref{cl:Ra<1} and the proof of the ``only if" part of Claim \ref{cl:ashaper}, such a policy function and the associated value function must satisfy \eqref{eq:vLa=1} through \eqref{eq:aL<1}, \eqref{eq:vRa<1} and \eqref{eq:aR<1} with coefficients given by \eqref{eq:A2} through \eqref{eq:D2}, given the value $v^*:=v(z^*)$. 
		
		Recall that the functions $a_-^*$ and $a_+^*$ are defined in \eqref{eq:a*-} and \eqref{eq:a*+}, respectively. By Claim \ref{cl:La<1} and equations \eqref{eq:C1} and \eqref{eq:C2}, it is easy to verify that
		$   \lim_{z\uparrow z^*}a(z;v^*)=a_-^*(v^*).$ 
		Similarly, by Claim \ref{cl:Ra<1} and equations \eqref{eq:D1} and \eqref{eq:D2}, it is easy to verify that
		$   \lim_{z\downarrow z^*}a(z;v^*)=a_+^*(v^*).$ 
		Since $a(\cdot)$ is continuous at $z^*$, $v^*$ must satisfy
		\begin{equation}\label{eq:determinev*}
			a_-^*(v^*)=a_+^*(v^*).
		\end{equation}
		By Claim \ref{cl:v*unique}, there is a unique  $v^*\in (v_R,v_L)$ satisfying \eqref{eq:determinev*}, rendering the unique (candidate) policy function. Finally, take such unique $v^*$, and let $a(\cdot)$ and $v(\cdot)$ be defined by \eqref{eq:vLa=1} through \eqref{eq:aL<1}, \eqref{eq:vRa<1} and \eqref{eq:aR<1} with coefficients given by \eqref{eq:A2} through \eqref{eq:D2}. Exactly the same verification argument as in the case of $r_1\geq r^*$ confirms that $a(\cdot)$ indeed solves \eqref{eq:agent-optimality}. In addition, $v(\cdot)$ is regular because the functions given by \eqref{eq:vLa=1}, \eqref{eq:vRa=1}, \eqref{eq:vLa<1} and \eqref{eq:vRa<1} are smooth, and value-matching and smooth-pasting conditions are imposed at $z_L$, $z^*$ and $z_R$.
	\end{proof}
	
	\subsubsection{Proof of Theorem \ref{uniqueness}}
	Lemmas \ref{lemma:principalbestreply} and \ref{lemma:optimality-agent} establish the unique \textit{structure} of Markov equilibria. To prove Theorem \ref{uniqueness}, we still need an argument for equilibrium existence and uniqueness.

	\begin{proof}[Proof of Theorem \ref{uniqueness}]
		Suppose first that $r_1\geq r^*$. By Lemma \ref{lemma:optimality-agent}, the agent's pseudo-best reply to any cutoff termination rule satisfies that $a(p)=0$ for all $p\in (0,1)$. In fact, the verification theorem we invoke in proving Lemma \ref{lemma:optimality-agent} tells us that such $a(\cdot)$ satisfies the agent's optimality condition \eqref{eq:agent-optimality} in a stronger sense, even if we allow him to maximize over all strategies in $\mathcal{A}$ instead of over Markov controls in $\mathcal{P}$. On the other hand, given this Markovian strategy of the agent under which the belief span is $(0,1)$, the proof of Lemma \ref{lemma:principalbestreply} (in the Online Appendix) can be used verbatim to show that the principal has a unique best reply whose policy function $b$ admits a cutoff $p^*\in (0,1)$. Hence, such $(a,b)$ is the unique Markov equilibrium in this case.
		
		Suppose now that $r_1<r^*$. Lemmas \ref{lemma:principalbestreply}  and \ref{lemma:optimality-agent} imply that any Markov equilibrium $(a,b)$ must be such that $a$ is hump-shaped and $b$ has a cutoff structure. We now show that there exists a unique Markov equilibrium. 
		
		To show that a Markov equilibrium exists, note that the pseudo-best reply of the noninvestible agent can be described by a function \(\varphi_1\) that maps a conjecture \(\tilde{p}^*\) about the cutoff \(p^*\) used by the principal into a policy function $a(\cdot)$ defined on (0,1), the probability-domain version of the policy function in \(z-\)space constructed in Claim \ref{cl:ashaper}'s proof. The function \(\varphi_1\) is continuous.\footnote{This follows from Lemma 
			OA.1 in the Online Appendix, i.e., the translation invariance of the agent's problem.} Moreover, the best reply of the principal can be described by a function \(\varphi_2\) that maps any Markov strategy \({\alpha}\) of the noninvestible agent into a unique cutoff \(p^*\), and Lemma \ref{lemma:principalbestreply} also tells us that $p^*\in [p^{**},p_H]\subset (0,1)$.\footnote{Lemma \ref{lemma:principalbestreply} is stated for an equilibrium. However, its proof is applicable to the principal's best reply to any Markovian strategy of the agent.}
		
		
		Define the composition \(\varphi := \varphi_2 \cdot \varphi_1\) that maps each conjecture \(\tilde{p}^*\) into the associated optimal cutoff \(p^*\).\footnote{Note that $SP[\varphi_1(\tilde{p}^*)]=(0,1)$ for all $\tilde{p}^*\in (0,1)$, because the pseudo-best reply delivered by Lemma \ref{lemma:optimality-agent} (derived in Claim \ref{cl:ashaper}'s proof) is always bounded away from $1$ by a positive number and thus the diffusion coefficient of the belief process is bounded away from $0$.} The mapping \(\varphi\) satisfies
		\[
		\varphi(\tilde{p}^*) = \argmax_{p' \in [p^{**},p_H]} \hat{M}(p,p',\tilde{p}^*),
		\] 
		where 
		\[
		\hat{M}(p,p', \tilde{p}^*) := r_2 \int_0^{+\infty} \int_{p'}^1 e^{-r_2 t} R(q) d\Gamma\left(t,q\middle|\tilde{p}^*, p\right)
		\]
		and
		\(\Gamma\) is the joint probability measure of getting the first Poisson shock in the stopping region \([p',1)\) at time \(t\) and state \(\tilde{p}\), when the prior belief at time zero is \(p\). Notice that \(\varphi(p^*)\) is the unique solution to this maximization problem and is independent of the prior due to its Markovian nature. Moreover, \(\hat{M}(p, p', \tilde{p}^*)\) is jointly continuous in \((p', \tilde{p}^*)\). Since the choice space is compact and the objective is continuous in both the choice variable \(p'\) and the ``parameter'' \(\tilde{p}^*\), the Maximum Theorem implies that \(\varphi\) is continuous.
		
		Since \(\varphi\) is a continuous function mapping from $(0,1)$ to $[p^{**},p_H]$ such that $\liminf_{p\to 0}[\varphi(p)-p]\geq p^{**}>0$ and $\limsup_{p\to 1}[\varphi(p)-p]\leq p_H-1<0$, the intermediate value theorem implies that there must be \(p^* \in (0,1)\) such that \(\varphi(p^*)=p^*\). Taking any such $p^*$, it is easy to verify that $\left({a}^*:=\varphi_1(p^*),b^*:=1_{\{p_t\geq p^*\}}\right)$ represents a Markov equilibrium. This proves equilibrium existence.

		To show that the Markov equilibrium is unique, let \(({a},b)\) be a Markov equilibrium in which the principal uses a threshold \(p^*\), associated with a likelihood \(z^*\). Consider the payoff of the principal when deviating to a different threshold \(z'\) when the state is \(z\).
		\[
		M(z,z',z^*):= p(z) \mathbb{E}\left\{e^{- r_2 T_{NI}(z,z',z^*)}\right\} w_{NI} + (1-p(z)) \mathbb{E}\left\{e^{- r_2 T_I(z,z',z^*)}\right\} w_I,
		\]
		where \(T_\theta(z,z',z^*)\) is the random time of occurrence of the first Poisson shock that arrives while the state lies in the stopping interval \([z',+\infty)\), provided the initial state is \(z\) and the dynamics is conditioned on type \(\theta \in \{I,NI\}\). By the conditional translation invariance property proved in Lemma 
		OA.2 of the Online Appendix,
		\[
		\mathbb{E}\left\{e^{- r_2 T_\theta(z,z',z^*)}\right\} = \mathbb{E}\left\{e^{- r_2 T_\theta(0,z'-z,z^*-z)}\right\}
		\]
		for all \(z,z',z^* \in \mathbb{R}\). The FOC of the principal is
		\[
		\frac{\partial M(z,z',z^*)}{\partial z'} = p(z) D_{NI}(z,z',z^*) w_{NI} + (1-p(z)) D_I(z,z',z^*) w_I = 0,
		\]
		where we define
		$D_\theta(z,z',z^*) := \frac{\partial \mathbb{E}\left\{e^{- r_2 T_\theta(z,z',z^*)}\right\}}{\partial z'}$
		for each \(\theta \in \{I,NI\}\), \(z,z',z^* \in \mathbb{R}\). Note that this condition should hold for every \(z \in \mathbb{R}\). In equilibrium, the principal's choice of \(z'\) must coincide with $z^*$, so the FOC becomes:
		$p(z) D_{NI}(z,z^*,z^*) w_{NI} + (1-p(z)) D_I(z,z^*,z^*) w_I = 0,$ 
		which can be rewritten as:
		\[
		p(z) =  \frac{D_C(z,z^*,z^*) w_C}{D_C(z,z^*,z^*) w_C - D_S(z,z^*,z^*) w_S}.
		\]
		Evaluating the limit from below as \(z \uparrow z^*\), we have
		\[
		p^* = p(z^*) =  \frac{D_C(z^*,z^*,z^*) w_C}{D_C(z^*,z^*,z^*) w_C - D_S(z^*,z^*,z^*) w_S}  =  \frac{D_C(0,0,0) w_C}{D_C(0,0,0) w_C - D_S(0,0,0) w_S},
		\]
		where the last equality follows from Lemma 
		OA.2 in the Online Appeneix, i.e., the conditional translation invariance of the principal's payoff function. Since the RHS is independent of \(p^*\), we conclude that there can be at most one value \(p^*\) consistent with a Markov equilibrium, establishing the uniqueness claim.
	\end{proof}
	
	\begin{corollary}\label{cor:convexity}
		The following hold:
		\begin{enumerate}
			\item $W(\cdot)$ is (weakly) increasing, nonnegative and convex on $(0,1)$, and it satisfies $\lim_{p\to 0}W(p)=0$ and $\lim_{p\to 1}W(p)=\frac{\lambda}{r_2+\lambda}w_{NI}$.
			\item  $v(\cdot)$ is strictly decreasing on $\mathbb{R}$, and it satisfies $\lim_{z\to -\infty}v(z) = u+c$ and $\lim_{z\to \infty}v(z) = \frac{r_1}{r_1+\lambda}(u+c)$. Moreover, $v(\cdot)$ is concave on $(-\infty,z^*)$ and convex on $(z^*,\infty)$.
		\end{enumerate}
	\end{corollary}
	
	\begin{proof}
		See Online Appendix.
	\end{proof}

	\subsection{Expected Performance: Toward a Proof of Theorem \ref{t:zigzag-shape}}\label{app:EP}
	In this section, we prove Theorem \ref{t:zigzag-shape} which is about the non-monotonicity of the expected performance. 
	
	Given a Markov equilibrium $(a,b)$ (where the equilibrium policy functions are defined on the $z-$space), let $v$ be the agent's value function, $z^*$ be the principal's termination cutoff, and recall that the agent's expected performance is given by
	\begin{equation}\label{eq:EPz}
		EP(z)=\snr \left[1-(1-a(z))p(z)\right],
	\end{equation}
	where 
	$p(z)=\frac{e^z}{1+e^z}.$
	Our analysis in this section fixes all model parameters, except $r_1$ and/or $\lambda$.\footnote{In Section \ref{sec:nonmonotone}, we defined $EP_t=\mu \left[1-(1-a_t)p_t\right]$. To ease notation, here we divide the original expression by $\sigma$, which is completely equivalent because for this exercise $\mu$ and $\sigma$ are fixed numbers.}

	\begin{lemma}\label{lem:EPpossibilities}
		If $r_1\geq r^*$, then $EP(\cdot)$ is strictly decreasing on $\mathbb{R}$. If $r_1<r^*$, then either $EP(\cdot)$ is strictly decreasing on $\mathbb{R}$, or $EP(\cdot)$ is
		\begin{itemize}
			\item strictly decreasing for $z<\underline{z}$, where $\underline{z}$ is some number in $[z_L,z^*)$;
			
			\item strictly increasing for $z\in (\underline{z},z^*)$;
			
			\item strictly decreasing for $z>z^*$. 
		\end{itemize}
	\end{lemma}
	\begin{proof}
		If $r_1\geq r^*$, Theorem \ref{uniqueness} tells us that $a(z)=0$ for all $z\in \mathbb{R}$. Thus, $EP(z)=\snr[1-p(z)]$, which is strictly decreasing on $\mathbb{R}$.
		
		If $r_1<r^*$, by Theorem \ref{uniqueness} the agent's equilibrium policy function $a$ is hump-shaped, with cutoffs denoted by $z_L$ and $z_R$ such that $a(z)>0$ if and only if $z\in (z_L,z_R)$. Obviously, $EP(\cdot)$ is strictly decreasing on $(-\infty,z_L)$ and on $(z^*,\infty)$, because on each of these intervals $a(\cdot)$ is weakly decreasing and $p(\cdot)$ is strictly increasing in $z$. We now focus on the monotonicity of $EP(\cdot)$ on $(z_L,z^*)$.
		
		First, analogous to \eqref{eq:a+a'}, we establish the following equality which links $EP(z)$ to $EP'(z)$:
		\begin{equation}\label{eq:EP+EP'}
			\snr-EP(z)-\frac{EP'(z)}{p(z)}=2\snr\left(\frac{v(z)-u}{c}\right).
		\end{equation}
		To see this, note first that since $a(z)>0$ on $(z_L,z^*)$,  by equation \eqref{eq:EPz} and Claim \ref{cl:aoptimal} we have
		\begin{equation}\label{eq:mu-EP}
			\snr-EP(z) = \snr [1-a(z)]p(z)=- \frac{r_1c}{\snr }\frac{p(z)}{v'(z)}.
		\end{equation}
		Differentiating this expression, we obtain
		\[
		-EP'(z) = -\frac{r_1c }{\snr}\left(-\frac{v''(z)p(z)}{v'(z)^2}+\frac{p(z)[1-p(z)]}{v'(z)}\right) =\snr[1-a(z)]p(z) \left(- \frac{v''(z)}{v'(z)}+1-p(z)\right).
		\]
		where the first equality follows from $p'(z)=p(z)[1-p(z)]$ and the second equality follows from \eqref{eq:mu-EP}. Recall, from the agent's HJB \eqref{eq:ODEvLa<1} in this case, that
		\[
		\frac{v''(z)}{v'(z)} = 1 + \frac{\left[v(z)-u\right] v'(z)}{\kappa_L}  = 1 - 2\left(\frac{v(z)-u}{c}\right)\frac{1}{1-a(z)}, 
		\]
		where the second equality follows from Claim \ref{cl:aoptimal} and $\kappa_L=\frac{r_1c^2}{2\snr^2}$. Thus,
		\begin{align*}
			-\frac{EP'(z)}{p(z)}&=\snr[1-a(z)]\left[2\left(\frac{v(z)-u}{c}\right)\frac{1}{1-a(z)}-p(z)\right]\\
			&=2\snr\left(\frac{v(z)-u}{c}\right)- \snr [1-a(z)]p(z)\\
			&= 2\snr\left(\frac{v(z)-u}{c}\right) - \left[\snr-EP(z)\right],
		\end{align*}
		where the last equality follows from \eqref{eq:mu-EP}. Equation \eqref{eq:EP+EP'} then follows immediately.
		
		From equation \eqref{eq:EP+EP'}, we can apply the same argument as that after equation \eqref{eq:a+a'} to show that $EP(\cdot)$ must be either strictly increasing, or strictly decreasing, or first strictly decreasing and then strictly increasing on $(z_L,z^*)$, a property that echoes what we have shown for the policy function $a(\cdot)$ in Claim \ref{cl:La<1}. Since we know that $EP(\cdot)$ is strictly decreasing on $(-\infty,z_L)$ and on $(z^*,\infty)$, the result in the lemma follows.
	\end{proof}
	
	\begin{corollary}\label{cor:EPmonotone}
		$EP(\cdot)$ is non-monotone if and only if $r_1<r^*$ and $EP'_-(z^*)>0$.
	\end{corollary}
	
	\begin{lemma}\label{lem:EPmonotone}
		$EP(\cdot)$ is non-monotone if and only if $r_1<r^*$ and $[1-a(z^*)]p(z^*)>2\left(\frac{v^*-u}{c}\right)$, where $v^*:=v(z^*)$.
	\end{lemma}
	\begin{proof}
		Since $a(z)>0$ on $(z_L,z^*]$, by Claim \ref{cl:aoptimal} and equation \eqref{eq:vLa<1} we have
		\begin{equation}\label{eq:1-aprop}
			1- a(z) = -\frac{r_1c}{\snr^2 v'(z)}\propto \frac{\phi\left(\frac{v(z)-u}{\sqrt{\kappa_L}}\right)}{e^z},\forall z\in (z_L,z^*].
		\end{equation}
		Then,
		\begin{align*}
			&EP_-'(z^*)>0\\
			\iff\quad&\left.\tfrac{d}{dz}\left[(1-a(z))p(z)\right]\right|_{z=z_-^*}<0\tag{by \eqref{eq:EPz}}\\
			\iff\quad& \left.-a'(z)p(z)+[1-a(z)]p(z)[1-p(z)]\right|_{z=z_-^*}<0\tag{because $p'(z)=p(z)[1-p(z)]$}\\
			\iff\quad &\left.[1-a(z)]p(z) \left[-\tfrac{a'(z)}{1-a(z)}+1-p(z)\right]\right|_{z=z_-^*}<0\\
			\iff\quad &\left.[1-a(z)]p(z)\left[\tfrac{d\ln [1-a(z)]}{dz}+1-p(z)\right]\right|_{z=z_-^*}<0\\
			\iff\quad & \left.\tfrac{d\ln [1-a(z)]}{dz}+1-p(z)\right|_{z=z_-^*} < 0 \tag{because $a(z)<1$ for all $z$}\\
			\iff\quad & \left.\tfrac{d}{dz}\left[\ln\phi\left(\tfrac{v(z)-u}{\sqrt{\kappa_L}}\right)-z\right]+1-p(z)\right|_{z=z_-^*}<0\tag{by \eqref{eq:1-aprop}}\\
			\iff\quad & \left(-\tfrac{v^*-u}{\sqrt{\kappa_L}}\right)\left(\tfrac{v'(z^*)}{\sqrt{\kappa_L}}\right)-p(z^*)<0 \tag{because $d\ln\phi(x)/dx=-x$}\\
			\iff\quad & \tfrac{v^*-u}{\kappa_L}<-\tfrac{p(z^*)}{v'(z^*)}\tag{because $-v'(z^*)>0$}\\
			\iff\quad & 2\left(\tfrac{v^*-u}{c}\right)<[1-a(z^*)]p(z^*)\tag{by $\kappa_L = \frac{r_1c^2}{2\snr^2}$ and  \eqref{eq:1-aprop} }
		\end{align*}
		By Corollary \ref{cor:EPmonotone}, the result follows.
	\end{proof}
	
	\begin{corollary}\label{cor:EPmonotone2}
		$EP(\cdot)$ is non-monotone if $r_1<r^*$ and $v^*<u$.
	\end{corollary}
	
	For a given $\lambda$, recall that $r^*(\lambda)$ is the unique solution to $\eqref{eq:r*}$. It is easy to see that there exists a unique $\lambda_1>0$ such that 
	\begin{equation}\label{eq:lambda1}
		r^*(\lambda_1)=\snr^2.
	\end{equation}
	Consequently, $r^*(\lambda)>\snr^2$ if and only if $\lambda>\lambda_1$.
	
	The following lemma deals with the agent's discount rates that are close to $r^*(\lambda)$.
	\begin{lemma}\label{lem:lambda1}
		If $\lambda>\lambda_1$ and $\snr^2< r_1<r^*(\lambda)$, then $EP(\cdot)$ is non-monotone.
	\end{lemma}
	\begin{proof}
		Recall that, in the proof of Claim \ref{cl:ashaper}, we have shown that if $a(\cdot)$ is hump-shape, then $v_L:=v(z_L)$ and $v_R:=v(z_R)$ are calculated in \eqref{eq:vL} and \eqref{eq:vR}, respectively. Moreover, $v_L>v_R$ if (and only if) $r_1<r^*(\lambda)$, and since $v'(\cdot)<0$, we have $v^*\in (v_R,v_L)$.
		
		From \eqref{eq:vL}, it is easy to verify that $v_L<u$ if and only if $r_1>\snr^2$. Therefore, if $\lambda>\lambda_1$ and $\snr^2< r_1<r^*(\lambda)$, we have $v_R<v^*<v_L<u$. By Corollary \ref{cor:EPmonotone2}, $EP(\cdot)$ is non-monotone.
	\end{proof}
	
	What about $r_1\in (0,\mbox{ }\snr^2]$? 
	For $\lambda>\lambda_1$ and $r_1\leq \snr^2$, recall that the functions $a_-^*(\cdot;r_1),a_+^*(\cdot;r_1,\lambda):[v_R,v_L]\rightarrow\mathbb{R}$ are defined by \eqref{eq:a*-} and \eqref{eq:a*+}, respectively. Recall also, from the proof of Lemma \ref{lemma:optimality-agent}, that $v^*\in (v_R,v_L)$ is the unique solution to $a^*_-(x;r_1)=a^*_+(x;r_1,\lambda)$. 
	
	\begin{claim}\label{cl:v*<u}
		If $\lambda>\lambda_1$ and $r_1\leq \snr^2$, then $u\in [v_R,v_L]$. Moreover, $v^*<u$ if and only if $a_-^*(u;r_1)<a_+^*(u;r_1,\lambda)$.
	\end{claim}
	\begin{proof}
		Suppose that $\lambda>\lambda_1$ and $r_1\leq \snr^2$. First, by definition of $\lambda_1$ in  \eqref{eq:lambda1}, we have $r^*(\lambda)>\snr^2\geq r_1$, so the equilibrium $a(\cdot)$ is hump-shaped and $v_L>v_R$. Substituting the expressions of $\xi_L$ and $\xi_R$ into \eqref{eq:vL} and \eqref{eq:vR}, we can rewrite $v_L$ and $v_R$ as
		\begin{align*}
			v_L(r_1)&=u+c\left(1-\frac{\sqrt{1+8r_1/\snr^2}+1}{4}\right),\\
			v_R(r_1,\lambda)&=\frac{r_1}{r_1+\lambda}\left[u+c\left(1+\frac{\sqrt{1+8(r_1+\lambda)/\snr^2}-1}{4}\right)\right].
		\end{align*}
		It is easy to verify that $v_L$ is strictly decreasing in $r_1$ and $v_R$ is strictly increasing in $r_1$. Thus, for $\lambda>\lambda_1$ and $r_1\leq \snr^2$,
		\begin{equation*}
			v_R(r_1,\lambda)<v_R(r^*(\lambda),\lambda)=v_L(r^*(\lambda))<v_L(\snr^2)\leq v_L(r_1)
		\end{equation*}
		where the inequalities follow from the monotonicity of $v_L$ and $v_R$ in $r_1$, and the equality follows from the definition of $r^*$. Note that $v_L(\snr^2)=u$, so we have $u\in [v_R,v_L]$. 
		
		The fact that $u\in [v_R,v_L]$ implies that $a_-^*(u;r_1)$ and $a_+^*(u;r_1,\lambda)$ are well-defined. By Claim \ref{cl:v*unique}, the function $g:[v_R,v_L]\to \mathbb{R}$ defined by
		$    g(x;r_1,\lambda):=a_-^*(x;r_1)-a_+^*(x;r_1,\lambda)
		$ 
		is strictly decreasing in $x$, and $v^*$ is the unique zero point of $g(\cdot;r_1,\lambda)$ on $[v_R,v_L]$. Therefore,
		\begin{equation*}
			v^*<u \iff g(u;r_1,\lambda)<0 \iff a_-^*(u;r_1)<a_+^*(u;r_1,\lambda),
		\end{equation*}
		as desired.
	\end{proof}
	
	The next lemma deals with the case where the agent's discount rate $r_1$ is small.
	\begin{lemma}\label{lem:lambda2}
		There exist $\lambda_2\geq \lambda_1$ and $\underline{r}>0$, such that if $\lambda>\lambda_2$ and $0<r_1<\underline{r}$, then $EP(\cdot)$ is non-monotone.
	\end{lemma}
	
	\begin{proof}
		See Online Appendix.
	\end{proof}
	
	We note that the bound $\underline{r}$ obtained in Lemma \ref{lem:lambda2} is a fixed number independent of $\lambda$. We now turn to the last case where $r_1\in [\underline{r},\snr^2]$.
	\begin{claim}\label{cl:a+bound}
		There exist $\lambda_3'\geq\lambda_1$ and $A> 1$ such that if $\lambda>\lambda_3'$, then
		\begin{equation}\label{eq:a+bound}
			a_+^*(u;r_1,\lambda)>1-A\exp\left(-\frac{\snr^2}{c^2r_1^2}\frac{\lambda^2}{r_1+\lambda}u^2\right),\forall r_1\in [\underline{r},\snr^2].
		\end{equation}
	\end{claim}
	\begin{proof}
		See Online Appendix.
	\end{proof}

	
	\begin{lemma}\label{lem:lambda3}
		There exists $\lambda_3\geq\lambda_1$ such that if $\lambda>\lambda_3$ and $\underline{r}\leq r_1\leq \snr^2$, then $EP(\cdot)$ is non-monotone.
	\end{lemma}
	\begin{proof}
		Let $\lambda_3'\geq \lambda_1$ and $A> 1$ be delivered  by Claim \ref{cl:a+bound}. For each $r_1\in [\underline{r},\snr^2]$, let $\lambda(r_1)$ be the unique solution on $\mathbb{R}_+$ to 
		$A\exp\left(-\frac{\snr^2}{c^2r_1^2}\frac{\lambda^2}{r_1+\lambda}u^2\right)=1-a_-^*(u;r_1).$ 
		Note that $\lambda(r_1)$ is well-defined for all $r_1\in [\underline{r},\snr^2]$ because the LHS is strictly decreasing in $\lambda$ while the RHS is independent of $\lambda$.\footnote{Moreover, when $\lambda=0$, the LHS is equal to $A>1-a^*_-(u;r_1)$; when $\lambda\to\infty$, the LHS converges to $0<1-a^*_-(u;r_1)$.} Hence,   we have
		\begin{equation}\label{eq:a-bound}
			A\exp\left(-\frac{\snr^2}{c^2r_1^2}\frac{\lambda^2}{r_1+\lambda}u^2\right)<1-a_-^*(u;r_1),\forall \lambda>\lambda(r_1).
		\end{equation}
		Also, $\lambda(r_1)$ is continuous in $r_1$ by the implicit function theorem. Let $\lambda_3'':=\max_{r\in [\underline{r},\snr^2]}\lambda(r)$ and $\lambda_3:=\max\left\{\lambda_3',\lambda_3''\right\}$. Combining \eqref{eq:a+bound} and \eqref{eq:a-bound}, we have 
		\begin{equation*}
			a_+^*(u;r_1,\lambda)>1-A\exp\left(-\frac{\snr^2}{c^2r_1^2}\frac{\lambda^2}{r_1+\lambda}u^2\right)>a_-(u;r_1),\text{ for all } \lambda>\lambda_3\text{ and }r_1\in [\underline{r},\snr^2].
		\end{equation*}
		Then, by Claim \ref{cl:v*<u} and Corollary \ref{cor:EPmonotone2}, we conclude that $EP(\cdot)$ is non-monotone if $\lambda>\lambda_3$ and $r_1\in [\underline{r},\snr^2]$.
	\end{proof}
	
	\begin{proof}[Proof of Theorem \ref{t:zigzag-shape}]
		Let $\lambda_1$ be defined in \eqref{eq:lambda1}, and let $\lambda_2$ and $\lambda_3$ be delivered by Lemmas \ref{lem:lambda2} and \ref{lem:lambda3}, respectively. Define  $\bar{\lambda}:=\max\{\lambda_1,\lambda_2,\lambda_3\}$. Lemmas \ref{lem:lambda1} through \ref{lem:lambda3} imply that if $\lambda>\bar{\lambda}$ and $r_1<r^*(\lambda)$, then $EP(\cdot)$ is non-monotone. Lemma \ref{lem:EPpossibilities} then leads to the conclusion of the theorem.
	\end{proof}

	\subsection{Effects of Better Transparency: Toward a Proof of Theorem \ref{t:noisehelps}}\label{app:transparency}
	In this section, we prove Theorem \ref{t:noisehelps} which is about the convergence of the principal's equilibrium value function when the signal-to-noise ratio $\snr$ grows without bound.
	
	Take any sequence $\{\snr_n\}_n$ such that $\lim_n\snr_n=+\infty$. For each \(n\in \mathbb{N}\), take the unique Markov equilibrium \(\left( a_{n}, b_{n}\right)\) associated with the signal-to-noise ratio $\snr_n$. 
	Let \(V_{n}\left( \cdot \right)\) be the agent's value function in the equilibrium \(\left( a_{n},b_{n}\right)\) and \(W_{n}\left(\cdot \right)\) be the principal's value function. We will often use $z\equiv \log(p/1-p)$ as state variable when analyzing the agent's behavior. When doing so, we denote by $v_n(z):=V_n(p(z))$ the agent's value function in the $z-$space. Write \(z_{n}^{\ast }\) for the principal's equilibrium cutoff. Write \(z_{L,n}\) for the infimum belief \(z\) at which the agent plays \(a_{n}\left( z\right) >0\) and write \(z_{R,n}\) for the supremum. Write $\mathbb{T}$ for the equilibrium stopping time that stops the play of the game. \textit{Without labeling explicitly, we note that the distribution of $\mathbb{T}$ depends on $n$ and the current state $z$.}
	For \(i=1,2\), let \(\mathbb{E}_{n}^{\theta}\left\{ e^{-r_{i}\mathbb{T}}\right\}\) be the expected discount factor when the stopping action is taken in the equilibrium \(\left( a_{n}, b_{n}\right)\) discounted at rate \(r_{i}\) and given the equilibrium strategy of type \(\theta\in \left\{NI,I\right\}\). When the game starts at state \(z\), let 
	\[
	\mathbb{E}_{n}\left\{e^{-r_{i}\mathbb{T}}\right\} :=p(z) \mathbb{E}_{n}^{NI}\left\{ e^{-r_{i}\mathbb{T}}\right\} +(1-p(z))\mathbb{E}_{n}^{I}\left\{ e^{-r_{i}\mathbb{T}}\right\}.
	\]
	
	\subsubsection{Case 1: $\lambda<r_1\left(\frac{c}{u}\right)$, i.e., $u< \frac{r_1}{r_1+\lambda}(u+c)$
	}
	\begin{proof}[Proof of Part 1 of Theorem \ref{t:noisehelps}]
		We first show that there is $\varepsilon>0$ such that $a_n(z_n^*)<1-\varepsilon$ for all $n$. Without loss, assume that $a_n(z_n^*)>0$ for all $n$. Then condition \eqref{eq:a+a'} implies that
		\begin{equation*}
			a_{n}(z_n^*) = 1-2\left( \frac{v_{n}(z_n^*)-u}{c}\right)-a_{n}^{\prime }\left( z_{n-}^*\right)
			< 1-2\left(\frac{ \frac{r_1}{r_1+\lambda}(u+c)-u}{c}\right),
		\end{equation*}
		where the inequality follows from $v_n(\cdot)\geq \frac{r_1}{r_1+\lambda}(u+c)$ and $a_n'(z_{n-}^*)>0$. By the assumption of Case 1, we can set $\varepsilon=2\left(\frac{ \frac{r_1}{r_1+\lambda}(u+c)-u}{c}\right)>0$.
		
		Since $a_n(\cdot)$ is maximized at $z_n^*$,  the above implies that $a_n(\cdot)$ is uniformly bounded away from $1$. Consequently, as $\snr_n\to \infty$, the principal learns the agent's type almost immediately, and thus the principal's equilibrium value function $W_n(\cdot)$ converges uniformly to her full-information value function $\overline{W}(\cdot)$.
	\end{proof}
	
	\subsubsection{Case 2: $\lambda>r_1\left(\frac{c}{u}\right)$, i.e., $u> \frac{r_1}{r_1+\lambda}(u+c)$
	}

	\begin{claim}\label{cl:ahump}
		There exists $N\in \mathbb{N}$ such that whenever $n\geq N$, $a_n(\cdot)$ is hump-shaped.
	\end{claim}
	
	\begin{proof}
		From condition \eqref{eq:r*}, it is easily verified that $\lim_n r^*_n = \lambda\left(\frac{2u}{c}+1\right)$. The assumption of Case 2 implies that $r_1<\lambda\left(\frac{u}{c}\right)< \lambda\left(\frac{2u}{c}+1\right)$. By Theorem \ref{uniqueness}, the result follows.
	\end{proof}
	
	We assume $n\geq N$ for the rest of the proofs in Case 2.
	\begin{claim}\label{cl:vbound}
		Take any compact set $[p_{1},p_{2}]\subset (0,1)$.  We have
		$\limsup_{n\to\infty} \left[\max_{z\in \left[ z(p_{1}),z(p_{2})\right] }v_{n}\left(z\right)\right] \leq u$.
	\end{claim}
	
	\begin{proof}
		Because $p^*\in [p^{**},p_H]$ (by Lemma \ref{lemma:principalbestreply}), we can without loss assume that $p_{n}^{\ast }\in \left[ p_{1},p_{2}\right].$ Since $v_{n}\left( \cdot \right) $
		is decreasing, it suffices to show that $\limsup v_{n}\left( z\left( 
		\frac{p_{1}}{2}\right) \right) \leq u.$
		
		By Claim \ref{cl:ahump}, $a_n(\cdot)$ is hump-shaped for all $n$. First assume that $a_{n}\left( z_{n}^{\ast }\right) \rightarrow 1$. Take any $\varepsilon >0$, and let $z_{n}^{\varepsilon }$ be the smallest $z$ such that $a_{n}\left( z\right)=1-\varepsilon $, which is well-defined for every large $n$ such that $a_{n}\left( z_{n}^{\ast }\right)>1-\varepsilon .$ Consider the stochastic process  $Z_{t}$ in the equilibrium $(a_n,b_n)$ under the noninvestible-type strategy and the initial condition $Z_{0}=z\left( \frac{p_{1}}{2}\right)$. Let $\mathbb{T}_{n}^{\dag }$ be the stopping time that stops the game at the first time that $Z_{t}\geq z_{n}^{\varepsilon }.$ From the law of motion \eqref{eq:dZt}, as $\snr_n\to \infty$, we have $\mathbb{E}_{n}^{NI}\left[e^{-r_{1}\mathbb{T}_{n}^{\dag }}\right] \rightarrow 1$ and hence $v_{n}\left( z\left( \frac{p_{1}}{2}\right) \right) \rightarrow v_{n}\left(z_{n}^{\varepsilon }\right) .$ Moreover, since $v_n(\cdot)$ is decreasing and concave to the left of $z^n_\varepsilon$ (by Corollary \ref{cor:convexity}), we have
		\begin{eqnarray*}
			r_{1}v_{n}\left( z_{\varepsilon }^{n}\right)  &=&r_{1}\left[ u+\left( 1-a_{n}\left( z_{\varepsilon}^{n}\right) \right) c\right]+\frac{1}{2}\snr^2\left[ 1-a_{n}\left(z_{\varepsilon }^{n}\right) \right] ^{2}\left[ v^{\prime }\left( a_{n}\left(z_{\varepsilon- }^{n}\right) \right) +v^{\prime \prime }\left( a_{n}\left(z_{\varepsilon- }^{n}\right) \right) \right]  \\
			&\leq &r_{1}\left[ u+\left( 1-a_{n}\left( z_{\varepsilon}^{n}\right) \right) c\right] ,
		\end{eqnarray*}
		which implies $v_{n}\left( z_{\varepsilon }^{n}\right) \leq u+\varepsilon c$, delivering the result as $\varepsilon 
		$ is arbitrary.
		
		Next assume that  $\lim \inf a_{n}\left( z_{n}^{\ast }\right) <1$. Take an $\epsilon>0$ such that  $u>\frac{r_1}{r_1+\lambda}(u+c)+\epsilon$; such an $\epsilon$ exists because of the assumption of Case 2. Notice
		that we can find a $z^{\dag }$ sufficiently large such that 
		$   v_{n}\left( z^{\dag}\right) <\frac{r_1}{r_1+\lambda}(u+c)+\epsilon, \forall n.$
		\footnote{To see this, note first that
			$p_{n}^{*}$ is bounded above by $p_{H}<1$ (Lemma \ref{lemma:principalbestreply}). Next observe that the posterior is a submartingale according to the
			strategy of the noninvestible type. This implies that for every $\eta >0$ we can find 
			$p_{\eta }<1$ such that, conditional on the noninvestible-type strategy, $p\in \left( p_{\eta },1\right)$ implies that the
			posterior goes below $p_{H}$  with probability less than
			$\eta$. This immediately easily implies the existence of said $z^{\dagger}$.}
		Let $\mathbb{T}_{n}^{\dag }$ be the stopping time that stops the game at the
		first time that $Z_{t}=z^{\dag }$. As $\snr_n\to \infty$, we have $v_{n}\left( z\left( 
		\frac{p_{1}}{2}\right) \right) \rightarrow v_{n}\left( z^{\dag
		}\right) <\frac{r_1}{r_1+\lambda}(u+c)+\epsilon<u$.  
	\end{proof}

	\begin{claim}\label{cl:zLn->-infty}
		\(\lim_{n \to \infty} z_{L,n} =-\infty\). 
	\end{claim}
	\begin{proof}
		Recall that $v_n(z_{L,n})=v_{L,n}$ whose expression is given by \eqref{eq:vL}, i.e., $v_{L,n}= u+c\left(1-\frac{\sqrt{1+8r_{1}/\snr_n^2}+1}{4}\right)$. It is easy to see that $\lim_{n}v_n(z_{L,n})=u+\frac{c}{2}$. Assume toward a contradiction, taking a subsequence if necessary, that $\lim_{n} z_{L,n} = \underline{z}>-\infty$. Then for any $\varepsilon>0$, we have $z_{L,n}\in [\underline{z}-\varepsilon,\underline{z}+\varepsilon]$ and $v_n(z_{L,n})\geq u+\frac{c}{4}$ when $n$ is sufficiently large. Take any $\varepsilon<(0,c/4)$. By Claim \ref{cl:vbound} and the monotonicity of $v_n(\cdot)$, there exists $n^*$ such that for every $n>n^*$ and for every $z\in [\underline{z}-\varepsilon,\underline{z}+\varepsilon]$, we have $v_n(z)< u+\varepsilon<u+\frac{c}{4}$, a contradiction to $v_{L,n}\to u+\frac{c}{2}$ and $z_{L,n}\to \underline{z}$.
	\end{proof}

	\begin{claim}\label{cl:alimitleft}
		For any $\kappa>0$, we have $\lim_{n\to\infty} a_{n}(z_n^*-\kappa)=1.$
	\end{claim}
	\begin{proof}
		Assume toward a contradiction that, taking a subsequence if necessary, $\lim a_{n}(z_n^*-\kappa)=1-2\varepsilon$ for some $\varepsilon>0$.
		Take any large $M>0$ and notice that Claim \ref{cl:zLn->-infty} tells us that  $z_{L,n}<z_{n}^*-\kappa-M$ for large $n$. Since $a_n(\cdot)$ is increasing on $\left[ z_{n}^*-\kappa-M,z_{n}^*-\kappa\right]$ and $\lim a_{n}(z_n^*-\kappa)=1-2\varepsilon$, we know that $a_{n}(z)<1-\varepsilon $ (infinitely
		often) for all $z\in \left[ z_{n}^*-\kappa-M,z_{n}^*-\kappa\right] $. Recall from condition \eqref{eq:a+a'} that for $z\in \left[ z_{n}^*-\kappa-M,z_{n}^*-\kappa\right] $,
		\begin{equation*}
			a_{n}^{\prime }\left( z\right) = 1-a_{n}(z)-2\left( \frac{v_{n}(z)-u}{c}\right),
		\end{equation*}
		implying, in light of Claim \ref{cl:vbound}, that for $n$ sufficiently large, we have $a_{n}'(z)>\frac{\varepsilon }{2}$ for all $z\in \left[ z_{n}^*-\kappa-M,z_{n}^*-\kappa\right]$. But then, we can take $M$ large enough such that $a_{n}\left(
		z_{n}^*-\kappa-M\right) <0$, a contradiction.
	\end{proof}
	
	\begin{claim}\label{cl:vlimitleft}
		For any $\kappa>0$, we have $\lim_{n\to\infty} v_{n}(z_n^*-\kappa)=u.$
	\end{claim}
	\begin{proof}
		From Claim \ref{cl:vbound} and Lemma \ref{lemma:principalbestreply} , we know $\limsup v_n(z_n^*-\kappa)\leq u$. Assume toward a contradiction, taking a subsequence if necessary, that $\lim_{n\to\infty} v_{n}(z_n^*-\kappa)=u-2\varepsilon$ for some $\varepsilon>0$. This implies that $v_{n}(z_n^*-\kappa)<u-\varepsilon$ for $n$ sufficiently large. Note also that Claim \ref{cl:zLn->-infty} tells us that $z_n^*-\kappa>z_{L,n}$ for $n$ sufficiently large. From condition \eqref{eq:a+a'}, 
		$a_{n}^{\prime }\left( z\right) = 1-a_{n}(z)-2\left(\frac{v_{n}(z)-u}{c}\right)$
		for all $z\in [z_n^*-\kappa,z_n^*-\frac{\kappa}{2}]$. Then by Claim \ref{cl:alimitleft} and the monotonicity of $a_n(\cdot)$, we know that for $n$ sufficiently large, $a_n'(z)>\frac{\varepsilon}{c}$ for all $z\in [z_n^*-\kappa,z_n^*-\frac{\kappa}{2}]$. But then, we have $a_n(z_n^*-\kappa)<1-\left(\frac{\kappa}{2}\right)\left(\frac{\varepsilon}{c}\right)$, a contradiction to Claim \ref{cl:alimitleft}.
	\end{proof}

	\begin{claim}\label{lem:learninglemma2}
		Fix a prior \(p_{0}\in \left( 0,1\right)\) and some \(\bar{p}\in \left( p_{0},1\right)\). For each \(\snr >0\), consider an adapted Markov function \(\alpha_{\snr}\left( \cdot \right)\) and a belief process defined by substituting $\alpha_\snr(\cdot)$ into \eqref{eq:dptagent}. Take \(\varepsilon >0\) and let \(\mathbb{\bar{T}}\) be the random time that stops the play in the first time that \(p\geq \bar{p}\). Then we have:
		\[
		\limsup_{\snr\uparrow \infty}\mathbb{E}^{NI}\left\{ r_1\int_{0}^{\mathbb{\bar{T}}}e^{-r_1t}\mathbb{I}_{\left\{ \alpha _{\snr}\left( p_{t}\right) \leq 1-\varepsilon \right\} }dt\right\} =0.
		\]
	\end{claim}
	\begin{proof}
		See Online Appendix (Section \ref{onlineapp:B3}, and Lemma 
		OA.5 in Section \ref{app:patientlimit}). In words, this lemma says that if the noninvestible type does not mimic too often, then as the noise in the signal vanishes, the principal can learn the agent's type almost immediately.
	\end{proof}
	
	\begin{lemma}\label{lem:noregimechange}
		Fix any $\kappa >0$ and, for each $n\in \mathbb{N}$,  assume that the game starts at the prior $z_{n}^{\ast }-\kappa.$ Let $
		\mathbb{T}_{n}^{\ast }$ be the stopping time that stops the play in the
		first time that a posterior reaches $[z_{n}^{\ast },\infty).$ We have:
		\begin{equation*}
			\limsup_{n\to\infty} \mathbb{E}_{n}^{NI}\left[ e^{-r_{1}\mathbb{T}_{n}^{\ast }}
			\right] =0.
		\end{equation*}
	\end{lemma}
	
	\begin{proof}
		Let $\hat{z}_{n}:=z_{n}^{\ast }-\kappa $ and $\underline{z}_{n}:=z_{n}^{\ast
		}-2\kappa .$ Let $\mathbb{T}_{n}\left( \underline{z}_{n}\right) $ be the
		stopping time that stops the play in the first time that the posterior
		reaches $\underline{z}_{n}$ and $\mathbb{T}_{n}\left( z_{n}^{\ast }\right) $
		be the stopping time that stops the play in the first time that the
		posterior reaches $z_{n}^{\ast }.$ Observe that for any $n$, $\mathbb{P}_{n}^{NI }\left[ \mathbb{T}_{n}\left( \underline{z}_{n}\right) <\infty \right] +\mathbb{P}_{n}^{NI }\left[ \mathbb{T}_{n}\left( z_{n}^{\ast }\right)<\infty \right] =1$.
		
		By Claim \ref{cl:alimitleft}, we know that for any $\varepsilon>0$, there exists $n_{1}\in \mathbb{N}$ such that $n>n_{1}$ implies that $v_n(\hat{z}_n^*)$ is bounded above by:
		\begin{equation}\label{eq:vhatpbound}
			\begin{aligned}
				& \mathbb{P}_{n}^{NI}\left[ \mathbb{T}_{n}\left( \underline{z}%
				_{n}\right) <\infty \right] \mathbb{E}_{n}^{NI}\left[ \int_{0}^{%
					\mathbb{T}_{n}\left( \underline{z}_{n}\right) }ue^{-r_{1}t}dt+e^{-r_{1}%
					\mathbb{T}_{n}\left( \underline{z}_{n}\right) }v_{n}\left( 
				\underline{z}_{n} \right) \mid \mathbb{T}_{n}\left( \underline{z}%
				_{n}\right) <\infty \right] \\
				&\quad +\mathbb{P}_{n}^{NI}\left[ \mathbb{T}_{n}\left( z_{n}^{\ast
				}\right) <\infty \right] \mathbb{E}_{n}^{NI}\left[ \int_{0}^{%
					\mathbb{T}_{n}\left( z_{n}^{\ast }\right) }ue^{-r_{1}t}dt+e^{-r_{1}%
					\mathbb{T}_{n}\left( z_{n}^{\ast }\right) }v_{n}\left(  z_{n}^{\ast
				} \right) \mid \mathbb{T}_{n}\left( z_{n}^{\ast }\right) <\infty %
				\right] +\varepsilon .
			\end{aligned}
		\end{equation}
		
		Next we obtain an upper bound for $v_{n}(z_{n}^{\ast } ).$
		For that, we let $\mathbb{T}_{\lambda }$ be the random time of the next
		Poisson shock. Note that
		\begin{align*}
			v_{n}(z_{n}^{\ast } ) =&\phantom{.}\mathbb{P}_{n}^{NI}\left[
			z_{\mathbb{T}_{\lambda }}>z_{n}^{\ast }\right] \mathbb{E}_{n}^{\theta _{S}}%
			\left[ \int_{0}^{\mathbb{T}_{\lambda }}e^{-r_{1}t}\left[ u+\left( 1-a_{n}(z_t)\right)c\right] dt\mid z_{\mathbb{T}_{\lambda
			}}>z_{n}^{\ast }\right] \\
			&+\mathbb{P}_{n}^{NI}\left[ z_{\mathbb{T}_{\lambda }}\leq
			z_{n}^{\ast }\right] \mathbb{E}_{n}^{\theta _{S}}\left[ \int_{0}^{\mathbb{T}%
				_{\lambda }}e^{-r_{1}t}\left[ u+\left( 1-a_{n}(z_t)\right)c\right]dt+e^{-r_{1}\mathbb{T}_{\lambda }}v_{n}\left( z_{\mathbb{T}%
				_{\lambda }} \right) \mid z_{\mathbb{T}_{\lambda }}\leq z_{n}^{\ast }%
			\right].
		\end{align*}
		
		Now we use the following facts to bound the expected value above:
		
		i) Because $a_n(z_t)\geq 0$, 
		\begin{equation*}
			\mathbb{E}_{n}^{NI}\left[ \int_{0}^{\mathbb{T}_{\lambda
			}}e^{-r_{1}t}\left[ u+\left( 1-a_{n}(z_t)\right) c\right]
			dt\mid z_{\mathbb{T}_{\lambda }}>z_{n}^{\ast }\right] \leq \mathbb{E}%
			_{n}^{NI}\left[ \int_{0}^{\mathbb{T}_{\lambda
			}}e^{-r_{1}t}(u+c)dt\mid z_{\mathbb{T}_{\lambda }}>z_{n}^{\ast }\right].
		\end{equation*}
		
		ii) From Claim \ref{cl:vbound},
		\begin{equation*}
			\limsup \mathbb{E}_{n}^{NI}\left[ \mathbb{I}_{\left\{ z_{\mathbb{T%
					}_{\lambda }}\leq z_{n}^{\ast }\right\} }v_{n}\left( z_{\mathbb{T}%
				_{\lambda }} \right) \right] \leq \limsup \mathbb{E}_{n}^{NI}\left[ \mathbb{I}_{\left\{ z_{\mathbb{T}_{\lambda }}\leq z_{n}^{\ast
				}\right\} }u\right].
		\end{equation*}
		
		iii) For every $\varepsilon >0,$ from Claim \ref{lem:learninglemma2}, 
		\begin{equation*}
			\limsup \mathbb{P}_{n}^{NI}\left[ \left\{ z_{\mathbb{T}_{\lambda
			}}\leq z_{n}^{\ast }\right\} \cap \left\{ \int_{0}^{\mathbb{T}_{\lambda
			}}e^{-r_{1}t}(1-a(z_t))dt > \varepsilon \right\} \right] =0.
		\end{equation*}
		Conditions (i), (ii) and (iii) above imply that for every $\epsilon >0$, we can find $n_{2}\in \mathbb{N}$ with $n_{2}>n_{1}$ such that $n>n_{2}$ implies %
		$v_{n}(z_{n}^{\ast } )\leq \mathbb{P}_{n}^{NI}\left[z_{\mathbb{T}_{\lambda }}>z_{n}^{\ast }\right] \mathbb{E}_{n}^{NI}\left[ \int_{0}^{\mathbb{T}_{\lambda}}e^{-r_{1}t}(u+c)dt\mid z_{\mathbb{T}_\lambda}>z_{n}^{\ast }\right]
		+\mathbb{P}_{n}^{\theta _{S}}\left[ z_{\mathbb{T}_{\lambda }}\leq
		z_{n}^{\ast }\right] u+\epsilon.$  
		Since Poisson shocks are independent of the Brownian motion, $\mathbb{E}%
		_{n}^{NI}\left[ \int_{0}^{\mathbb{T}_{\lambda
		}}e^{-r_{1}t}(u+c)dt\mid z_{\mathbb{T}_{\lambda }}>z_{n}^{\ast }\right]
		=\left( \frac{r}{r+\lambda }\right) (u+c)<u$. Moreover, notice that the $z_t$ (and $p_t$) are submartingales conditional on $\theta=NI$, so $\mathbb{P}%
		_{n}^{NI}\left[ z_{\mathbb{T}_{\lambda }}>z_{n}^{\ast }\right] \geq 
		\frac{1}{2}.$ Therefore, the last two observations imply%
		\begin{equation*}
			v_{n}(z_{n}^{\ast })\leq \frac{1}{2}\left( 
			\frac{r_1}{r_1+\lambda }\right) (u+c)+ \frac{1}{2}u+\epsilon.
		\end{equation*}
		Since $\left( 
		\frac{r_1}{r_1+\lambda }\right) (u+c)<u$ and $\epsilon$ is arbitrary, we conclude that $\limsup v_n(z_n^*)<u$.
		
		Going back to the upper bound \eqref{eq:vhatpbound} for $v_n(\hat{z}_n^*)$, we have shown that $\limsup v_n(z_n^*)<u$, and from Claim \ref{cl:vlimitleft}, we know that $\lim v_n(\hat{z}_n^*)= \lim v_n(\underline{z}_n^*)=u$. So for \eqref{eq:vhatpbound} to be an valid upper bound of $v_n(\hat{z}_n^*)$ for arbitrary $\varepsilon$, we must have $\limsup_{n\to\infty} \mathbb{E}_{n}^{NI}\left[ e^{-r_{1}\mathbb{T}_{n}(z_n^*)}\right] =0$, i.e., $\limsup_{n\to\infty} \mathbb{E}_{n}^{NI}\left[ e^{-r_{1}\mathbb{T}_{n}^{\ast }}\right] =0$, as desired.
	\end{proof}
	
	\begin{lemma}\label{lem:znconvergence2}
		$\lim_{n\to \infty} z_n^* = z^{**}$, where $z^{**}$ is the myopic cutoff satisfying $R(p(z^{**}))=0$
	\end{lemma}
	\begin{proof}
		Recall that we always have $z_n^*\geq z^{**}$.
		Hence assume (toward a contradiction) that we can find some $\varepsilon>0$ such that, taking a subsequence if necessary, $z_{n}^{\ast
		}>z^{**}+\varepsilon$ for every $n$. For every $n$, consider the game starts at $z_n^*-\frac{\varepsilon}{2}>z^{**}+\frac{\varepsilon}{2}$. Let $\mathbb{T}_{n}^{\ast }$ be the stopping time that stops the play in the
		first time that a posterior reaches $[z_{n}^{\ast },\infty).$ By Lemma \ref{lem:noregimechange}, we have $\mathbb{E}_{n}^{NI}\left[ e^{-r_{1}\mathbb{T}_{n}^{\ast }}\right] \rightarrow 0$, implying that $\limsup W_n\left(p\left(z_n^*-\frac{\varepsilon}{2}\right)\right)\leq 0$. But since $z_{n}^{\ast
		}-\frac{\varepsilon}{2}>z^{**}+\frac{\varepsilon}{2}$, the principal can get a strictly positive payoff by terminating the relationship when the next stopping opportunity arrives. So the principal has a profitable deviation at $z_n^*-\frac{\varepsilon}{2}$ when $n$ is sufficiently large, a contradiction.
	\end{proof}
	
	\begin{proof}[Proof of Part 2 of Theorem \ref{t:noisehelps}]
		In light of Corollary \ref{cor:convexity}, we extend each $W_n$ continuously from $(0,1)$ to $[0,1]$ by setting $W_n(0)=0$ and $W_n(1)=\frac{\lambda}{r_2+\lambda}w_{NI}$.
		
		We first show that $W_n(\cdot)$ converges to $\underline{W}(p^{**})=0$. By Lemmas \ref{lem:noregimechange} and \ref{lem:znconvergence2}, it is easy to see that $\lim_{n\to\infty}W_n(p)=0$ for all $p<p^{**}$.\footnote{Note that Lemma \ref{lem:noregimechange} enables us to conclude that $\limsup_{n} \mathbb{E}_{n}\left[ e^{-r_{2}\mathbb{T}}\right] =0$. This is because the principal only derives positive payoff from terminating against the noninvestible type. As a result, if $\limsup_{n} \mathbb{E}^{NI}_{n}\left[ e^{-r_{2}\mathbb{T}}\right]=0$, then we must have $\limsup_{n} \mathbb{E}^{I}_{n}\left[ e^{-r_{2}\mathbb{T}}\right]=0$, for otherwise the principal's equilibrium payoff would be negative.}  
		Suppose toward a contradiction, taking a subsequence if necessary, that  $\lim_{n\to \infty}W_n(p^{**})=\delta>0$. Let $\epsilon=\frac{\delta(1-p^{**})}{2w_{NI}}$. For $n$ large enough, we have $W_n(p^{**}-\epsilon)<\frac{\delta}{2}$. Since $W_n(\cdot)$ is convex (by Corollary \ref{cor:convexity}), we have
		$   \frac{W_n(1)- W_n(p^{**})}{1-p^{**}}\geq \frac{W_n(p^{**})-W_n(p^{**}-\epsilon)}{\epsilon}\geq \frac{w_{NI}}{1-p^{**}},$
		which implies $W_n(1)>w_{NI}$, a contradiction.
		
		Next, for each $n$, because $W_n(0)=\underline{W}(0)$, $W_n(1)=\underline{W}(1)$, and $W_n(\cdot)$ is increasing and convex, we have $0\leq W_n'(\cdot)\leq \frac{\lambda}{r_1+\lambda}$. But then, we always have $p^{**}=\argmax_{p\in [0,1]}|W_n(p)-\underline{W}(p)|$. Hence, uniform convergence of $W_n$ follows immediately from its pointwise convergence at $p^{**}$.
	\end{proof}
	

	\newpage
	\section{Online Appendix}
	\setcounter{lemma}{0}
	\setcounter{claim}{0}
	\renewcommand{\thelemma}{OA.\arabic{lemma}}
	\renewcommand{\theclaim}{OA.\arabic{claim}}
	\subsection{Omitted Proofs for Theorem \ref{uniqueness}}\label{onlineapp:B1}
	\subsubsection{Proofs of Lemmas \ref{lem:fullspan} and \ref{lemma:principalbestreply}}
	Consider a Markovian equilibrium, $(\alpha,\beta)$, and the underlying probability space $\left( \Omega ,\mathfrak{F,}%
	\mathbb{P}\right) $.  For each $p\in \left( 0,1\right) ,$ we
	define $\Phi \left( p\right) :=\left\{ \omega \in \Omega :\exists t\leq 
	\mathbb{T}\text{ such that }p_{t}(\omega )=p\right\} $, where $\mathbb{T}$ is the random stopping time of the relationship induced by $(\alpha,\beta)$. The belief span, $\text{SP}(\alpha)$,
	is the set of all $p$ such that $\mathbb{P}\left( \Phi \left( p\right)
	\right) >0$. Clearly, $\text{SP}(\alpha)$ is a connected set because the sample path of $X_t$ is almost surely continuous. Let $\underline{p}:=\inf \text{SP}(\alpha)$, and $\bar{p}:=\sup \text{SP}(\alpha)$. Define the principal's value function $W$ as in the main text on the domain of $\text{SP}(\alpha)$. The function $W$ is continuous because the agent's equilibrium policy function $a\in \mathcal{P}$.
	
	\begin{claim}\label{cl:spanopen}
		$SP(\alpha)$ is an open interval. That is, $SP(\alpha)=(\underline{p},\bar{p})$.
	\end{claim}
	\begin{proof}
		Since $SP(\alpha)$ is a connected set, we only need to show that $\bar{p},\underline{p}\notin SP(\alpha)$. Suppose, toward a contradiction, that $\bar{p}\in SP(\alpha)$. Then, consider a history that leads to the belief $\bar p$ and the continuation play starting from this history. Since the belief process is a martingale, we must have $%
		p_{t}=\bar p$ for all $t\leq \mathbb{T}$ and almost all sample paths. Agent's optimality then
		implies $a(\bar p)=0$,
		and thus the diffusion coefficient of the belief process at $\bar{p}$ is strictly positive. This contradicts $p_{t}=\bar p$ for all $t\leq \mathbb{T}$ and almost all sample paths. The same argument proves that $\underline{p}\notin SP(\alpha)$.
	\end{proof}
	
	\begin{claim}\label{cl:cutoffstructure}
		The principal's equilibrium policy function $b$ has a cutoff structure on $SP(\alpha)$. That is, there exists a unique $p^*\in [\underline{p},\bar{p}]$ such that $p\in(\underline{p},p^*)$ implies $b(p)=0$ and $p\in (p^*,\bar{p})$ implies $%
		b(p)=1$.
	\end{claim}
	\begin{proof}
		Recall that $R(p):=pw_{NI}+(1-p)w_I$ is the principal's expected payoff if the relationship is terminated when her belief is $p$. For any $p\in SP(\alpha)$, define 
		$F(p):=W(p)-R(p).$ 
		At any time $t$ such that the stopping opportunity arrives, given her belief $p_t=p\in SP(\alpha)$, if the principal terminates the relationship, her expected payoff is $R(p)$; if the principal continues the relationship, her continuation value is $W(p)$.	Thus, principal's optimality requires that $b(p)=1$ if $F(p)<0$ and that $b(p)=0$ if $F(p)>0$. 
		
		We first establish two useful properties of $F$. 
		
		\noindent\textbf{Property 1:} If $F(\tilde{p})>0$ at some $\tilde{p}\in SP(\alpha)$, then $F(p)>0$ for all $p<\tilde{p}$. 
		
		To see this, suppose that $F(\tilde{p})>0$ at some $\tilde{p}\in SP(\alpha)$. Let $(p_a,p_b)$ be the largest interval containing $\tilde{p}$ such that $F(p)>0$ for all $p\in (p_a,p_b)$. We want to show that $p_a=\underline{p}$. Suppose, toward a contradiction, that $p_a>\underline{p}$. Since $W$ is continuous, we have $F(p_a)=0$, i.e., $W(p_a)=R(p_a)$. Moreover, principal's optimality requires that $b(p)=0$ for all $p\in (p_a,p_b)$. We consider two cases, and will reach a contradiction in each of these cases.
		
		\noindent\textbf{Case 1:} $p_b<\bar{p}$. 
		
		In this case, continuity of $W$ also implies that $F(p_b)=0$, i.e., $W(p_b)=R(p_b)$. Consider now a history that leads to the belief $\tilde{p}$ and the continuation play starting from this history. Let  $\mathbb{T}^{\dag }$ be the first time that the posterior belief reaches $\left( p_{a},p_{b}\right)
		^{c}$ (setting $\mathbb{T}^{\dag }=\infty $ if this event is not reached in
		finite time). Let $\varphi $ represent the probability measure (from the principal's perspective) induced by
		the distribution of $p_{\mathbb{T}^{\dag }}$. Then,
		\begin{align*}
			R(\tilde{p}) <W(\tilde{p}) &= \int_{p_{a}}^{p_{b}}W(p)\mathbb{E}\left[ e^{-r_2\mathbb{T}^{\dag }}\mid p_{%
				\mathbb{T}^{\dag }}=p\right] \varphi \left( dp\right)  \\
			&=\int_{p_{a}}^{p_{b}}R(p)\mathbb{E}\left[ e^{-r_2\mathbb{T}^{\dag }}\mid p_{%
				\mathbb{T}^{\dag }}=p\right] \varphi \left( dp\right) 
			<\int_{p_{a}}^{p_{b}}R(p)\varphi \left( dp\right)
			= R(\tilde{p}).
		\end{align*}%
		The second equality trivially holds if $\mathbb{E}\left[ e^{-r_2\mathbb{T}%
			^{\dag }}\mid p_{\mathbb{T}^{\dag }}=p\right] =0$, and otherwise if $ \mathbb{T}%
		^{\dag }<\infty$, then it holds because $W(p_{\mathbb{T}^{\dag }})=R(p_{\mathbb{T}%
			^{\dag }})$. The last inequality uses the facts that  $0\leq W(p_a)=R(p_{a})$
		implies  $R(p)>0$ for all $p>p_{a}$, and that $\mathbb{T}^{\dag }>0$ almost surely. The final equality holds because $p_t$ is a bounded martingale and $R(\cdot)$ is an affine function. But then, we have an obvious contradiction.
		
		\noindent\textbf{Case 2:} $p_b=\bar{p}$. 
		
		In this case, we have $b(p)=0$ for all $p\in (p_a,\bar{p})$. Consider again a history that leads to the belief $\tilde{p}$ and the continuation play starting from this history. Let  $\mathbb{T}^{\dag }$ be the first time that the posterior reaches $p_a$ (setting $\mathbb{T}^{\dag }=\infty $ if this event is not reached in
		finite time). Let $\varphi $ represent the probability measure (from the principal's perspective) induced by
		the distribution of $p_{\mathbb{T}^{\dag }}$. Then,
		\[
		R(\tilde{p}) <W(\tilde{p})=W(p_a)\mathbb{E}\left(e^{-r\mathbb{T}^{\dag }}\right)<R(p_a),
		\]
		which contradicts $R$ being increasing.
		
		\noindent\textbf{Property 2:} Let $p^*:=\sup\{p\in SP(\alpha): F(p)>0\}$.\footnote{By convention, if $\{p\in SP(\alpha): F(p)>0\}=\emptyset$, we set $\sup\{p\in SP(\alpha): F(p)>0\}=\underline{p}$.} Then, $F(p)<0$ for all $p>p^*$.
		
		By definition of $p^*$, we know that $F(p)\leq 0$ for all $p\geq p^*$, i.e., $W(p)\leq R(p)$ for all $p\geq p^*$, so it is weakly optimal for the principal to terminate the relationship whenever $p\in (p^*,\bar{p})$. Suppose, toward a contradiction, that $F(\tilde{p})=0$ for some $\tilde{p}>p^*$. Consider a history that leads to the belief $\tilde{p}$ and the continuation play starting from this history. Let  $\mathbb{T}^{\dag }$ be the first time that the stopping opportunity arrives or the posterior reaches $p^*$ (setting $\mathbb{T}^{\dag }=\infty $ if this event is not reached in
		finite time). Let $\varphi $ represent the probability measure (from the principal's perspective) induced by
		the distribution of $p_{\mathbb{T}^{\dag }}$. Then,
		\begin{align*}
			R(\tilde{p}) =W(\tilde{p}) 
			&=\int_{p^*}^{\bar{p}}W(p)\mathbb{E}\left[ e^{-r_2\mathbb{T}^{\dag }}\mid p_{%
				\mathbb{T}^{\dag }}=p\right] \varphi \left( dp\right)  \\
			&\leq\int_{p^*}^{\bar{p}}R(p)\mathbb{E}\left[ e^{-r_2\mathbb{T}^{\dag }}\mid p_{%
				\mathbb{T}^{\dag }}=p\right] \varphi \left( dp\right)  
			<\int_{p^*}^{\bar{p}}R(p)\varphi \left( dp\right)
			=R(\tilde{p}).
		\end{align*}
		The first equality follows from the contradiction assumption that $F(\tilde{p})=0$, the first inequality follows from the definition of $p^*$ such that $F(p)\leq 0$ for all $p\geq p^*$, the last inequality holds because $0\leq W(p^*)\leq R(p^*)$ implies that $R(p)>0$ for all $p>p^*$, and the final equality holds because $p_t$ is a bounded martingale and $R(\cdot)$ is an affine function. But then, we have an obvious contradiction, establishing Property 2.
		
		These two properties of $F$ immediately deliver our result. Specifically, let $p^*:=\sup\{p\in SP(\alpha): F(p)>0\}$. Then, Property 1 implies that $F(p)>0$ (and thus $b(p)=0$) for all $p\in (\underline{p},p^*)$, and Property 2 implies that $F(p)<0$ (and thus $b(p)=1$) for all $p\in (p^*,\bar{p})$.
	\end{proof}

	We continue with a technical result that will be later used.
	\begin{claim}
		\label{claim_3}Fix a positive integer $T.$ For any $\varepsilon >0$ there
		exists $\eta >0$ satisfying the following property: 
		Take  any pair of adapted processes $dY_{t}^{1}=\mu _{1,t}dt+ \sigma dB_{t}$
		and $dY_{t}^{2}=\mu _{2,t}dt+ \sigma dB_{t}$   such that $\mu _{j,t}\in \left[
		0,1\right] $ for $j=1,2$ and for every $t.$ Let $\mathbb{P}_{1}$ and $%
		\mathbb{P}_{2}$ be the probability distributions over $\left( C([0,T]),%
		\mathbb{B}\left( C([0,T])\right) \right)$\footnote{$\mathbb{B}$ stands for
			the Borel sigma-field, and $C([0,T])$ is the set of continuous functions over $[0,T]$.} generated by such stochastic processes. If $A\in 
		\mathbb{B}\left( C([0,T])\right) $ is such that $\mathbb{E}_{\mathbb{P}_{1}}%
		\left[ \mathbb{I}_{A}\right] <\eta $ then $\mathbb{E}_{\mathbb{P}_{2}}\left[ 
		\mathbb{I}_{A}\right] <\varepsilon$. 
	\end{claim}
	
	\begin{proof}
		Dividing both processes by $\sigma $ and subtracting the same drift from
		both processes if necessary we may assume that $dY_{t}^{1}=dB_{t}$ and $%
		dY_{t}^{2}=\mu _{2,t}dt+dB_{t}$ with $\mu _{2,t}\in \left[ -\sigma
		^{-1},\sigma ^{-1}\right] .$ Since the drift is bounded we can invoke
		Girsanov's theorem to obtain%
		\begin{equation*}
			\mathbb{E}_{\mathbb{P}_{2}}[ \mathbb{I}_{A}] =\mathbb{E}_{\mathbb{%
					P}_{1}}[ \mathbb{I}_{A}M_{T}] ,
		\end{equation*}%
		where $M_{T}=\exp \left( \int_{0}^{T}\mu _{2,t}dB_{t}-\frac{1}{2}%
		\int_{0}^{T}\mu _{2,t}^{2}dt\right) .$ Notice that $M_{T}\leq F_{\mu
			_{2}}:=\exp \left( \int_{0}^{T}\mu _{2,t}dB_{t}\right) $. Since this class
		of processes is uniformly integrable we can take $n^{\ast }\in \mathbb{N}$ such that $\mathbb{E}_{\mathbb{P}_{1}}\left[ F_{\mu _{2}}\mathbb{I}%
		_{\left\{ F^{\mu _{2}}>n^{\ast }\right\} }\right] <\frac{\varepsilon }{2}$
		(holding for every process in this class) and consequently%
		$   \mathbb{E}_{\mathbb{P}_{1}}\left[ M_{T}\mathbb{I}_{\left\{ M_{T}>n^{\ast
			}\right\} }\right]\leq\mathbb{E}_{\mathbb{P}_{1}}\left[ F_{\mu _{2}}\mathbb{I}_{\left\{
			F_{\mu _{2}}>n^{\ast }\right\} }\right] <\frac{\varepsilon }{2}.$
		
		Therefore, taking $\eta =\frac{\varepsilon }{2n^{\ast }}$, we obtain that
		\begin{align*}
			\mathbb{E}_{\mathbb{P}_{2}}\left[ \mathbb{I}_{A}\right]  =\mathbb{E}_{%
				\mathbb{P}_{1}}\left[ \mathbb{I}_{A}M_{T}\mathbb{I}_{\left\{ M_{T}\leq
				n^{\ast }\right\} }\right] +\mathbb{E}_{\mathbb{P}_{1}}\left[ \mathbb{I}%
			_{A}M_{T}\mathbb{I}_{\left\{ M_{T}>n^{\ast }\right\} }\right] 
			&\leq \mathbb{E}_{\mathbb{P}_{1}}\left[ \mathbb{I}_{A}M_{T}\mathbb{I}%
			_{\left\{ M_{T}\leq n^{\ast }\right\} }\right] +\frac{\varepsilon }{2} \\
			&<n^{\ast }\mathbb{E}_{\mathbb{P}_{1}}\left[ \mathbb{I}_{A}\right] +\frac{%
				\varepsilon }{2}
			<\varepsilon. \qedhere
		\end{align*}%
	\end{proof}

	\begin{claim}
		\label{claim_2} $\bar{p}=1$
	\end{claim}
	\begin{proof}
		Assume towards a contradiction that $\bar{p}<1$.
		
		\noindent{\textbf{Case 1: }} $\bar p>p^{\ast }$.
		
		The belief process $p_t$ is a martingale, so for every $\varepsilon >0$ there exists an $\epsilon>0$ such that if $p_{t}>\bar p-\epsilon $, then $%
		\mathbb{P}\left( \inf_{s>t}p_{s}>p^{\ast }+\varepsilon \phantom{.} | \phantom{.} \theta=NI\right)
		>1-\varepsilon$. This implies that $\mathbb{P}\left( \mathbb{%
			T=T}_{\lambda } \phantom{.} | \phantom{.} \theta=NI \right) >1-\varepsilon $ where $\mathbb{T}_{\lambda }$ is the
		arrival of the next Poisson-shock. Notice that for every $\eta >0$ we can
		take $\varepsilon _{\eta }>0$ such that the agent's payoff at $p_{t}$ is no more than $(u+c)\left( \frac{r_{1}}{r_{1}+\lambda }%
		\right) +\eta .$ This implies that for every $\nu >0$ we can take $\eta $
		small enough (taking $\varepsilon _{\eta }$ to satisfy the condition above) 
		so that $\mathbb{E}\left( \int_{t}^{\min \left\{ t+1,\mathbb{T}%
			\right\} }\mathbb{(}a(p_t))dt \phantom{.} | \phantom{.} \theta=NI \right) <\nu $. Hence, there exists $\varpi>0 $ such that $\mathbb{E} \left( \int_{t}^{\min
			\left\{ t+1,\mathbb{T}\right\} }(1-a(t))dt \phantom{.} | \phantom{.} \theta=NI \right) >\varpi$.

		Consider the law of motion \eqref{eq:dptagent} when $p_{t}\in \left[ p^{\ast },\bar p \right]$. Observe that the instantaneous variance of the belief process when the (noninvestible type) agent
		plays $a(\cdot)$ is bounded below by a positive constant times $(1-a_t)^{2}\min
		\left\{ p^{\ast }(1-p^{\ast }),\bar p(1-\bar p)\right\} ^{2}>0$. Because $\bar p<1$ and because  $p_{t}-p_{0}=\int_{0}^{t}dp_{t}$, we obtain  that $%
		\mathbb{E}\left[ \left\vert p_{\min \left\{ t+1,\mathbb{T}\right\}
		}-p_{t}\right\vert ^{2}\right] \geq \delta$ for some positive constant $%
		\delta$, hence $\mathbb{E}\left[ \left\vert p_{\min \left\{ t+1,\mathbb{T}%
			\right\} }-p_{t}\right\vert \right] \geq \delta$. Because $\left(
		p_{\min \left\{ t+1,\mathbb{T}\right\} }-p_{t}\right) $ has mean zero, we obtain that  $%
		\mathbb{E}\left[ \left( p_{\min \left\{ t+1,\mathbb{T}\right\}
		}-p_{t}\right) ^{+}\right] \geq \frac{\delta}{2}.$ Taking $\varepsilon <%
		\frac{\delta}{4}$ we conclude that $\mathbb{P}\left( p_{\min \left\{ t+1,%
			\mathbb{T}\right\} }>\bar p+\frac{\varepsilon }{2}\right) >0$, which is a
		contradiction.

		\noindent{\textbf{Case 2: }} $\bar{p}\leq p^*$.
		
		Assume $\bar{p}\leq p^{\ast }.$ Then, $b(p)=0$ for all $p\in SP(\alpha,\beta)$. Claim \ref*{claim_3} implies that, for every $T>0,$
		if the noninvestible agent plays $a_{t}=0$ for every $t\in \left[ 0,T\right]$, 
		then the relationship terminates before $T$ with probability zero, which
		implies that the agent's best response must satisfy $a_{t}=0$  for every $t>0$. This
		contradicts the assumption that $\bar{p}$ is never reached.
	\end{proof}
	
	\begin{claim}
		\label{claim_4}$\underline{p}=0$.
	\end{claim}
	
	\begin{proof}
		Assume towards a contradiction that $%
		\underline{p}>0$. 
		
		\noindent{\textbf{Case 1:}} $\underline{p}<p^{\ast }$.

		\noindent{\textbf{Step 1:}} For every $\eta \in \left( 0,\frac{p^{\ast
			}-\underline{p}}{2}\right) $ there exists  $\epsilon >0$ such that if $%
		p_{t}< \underline{p}+\epsilon$ then
		$\mathbb{P}\left(
		\sup_{s>t}p_{s}\geq p^{\ast } \phantom{.} | \phantom{.} \theta=NI\right) <\eta, $
		and consequently  $\mathbb{P}%
		\left( \sup_{s>t}p_{s}\geq p^{\ast }\right) <\eta$.
		
		This follows from the martingale property of the belief process.
		
		\noindent{\textbf{Step 2:}} For every $\varepsilon >0$ and $T\in 
		\mathbb{N}
		$ there exists $\epsilon >0$  such that if $p_{t}<\underline{p}+\epsilon $ and if the agent plays a strategy $\tilde \sigma$ that plays $a_t=0$  for all
		$t>0$,  then
		$\mathbb{P}^{\tilde \sigma}\left( \left\{ \mathbb{T}<T\right\} \right) <\varepsilon$.
		
		This follows from Step 1 and Claim 
		OA.3.
		
		\noindent{\textbf{Step 3:}} There exists $T^{\ast }\in 
		\mathbb{N}
		$  and $\varepsilon >0$  such that if $\mathbb{P}^{\tilde \sigma
		}\left( \left\{ \mathbb{T}<T^*\right\} \right) <\varepsilon $, then 
		
		\[ \mathbb{E}\left( \int_{t}^{\min \left\{
			t+T^{\ast },\mathbb{T}\right\} }(1-a_t)dt \phantom{.}|\phantom{.} \theta=NI\right) \geq 1.
		\]
		
		To see this, take an arbitrary $T\in \mathbb{N}.$ If $\mathbb{E}\left( \int_{t}^{\min \left\{ t+T,\mathbb{T}\right\}
		}\left( 1-a_t\right) dt\mid \theta =NI\right) <1$, then agent gets no
		more than%
		\begin{equation*}
			(u+c)\left( \int_{0}^{1}r_1e^{-r_1t}ds\right) +u\left(
			\int_{1}^{T}r_1e^{-r_1t}ds\right) +(u+c)\left( \int_{T}^{\infty
			}r_1e^{-r_1t}ds\right) <(u+c)\left( \int_{0}^{T}re^{-r_1t}ds\right) 
		\end{equation*}%
		for $T$ large enough. We can thus let $T^{\ast }$ be the smallest positive
		integer satisfying the inequality above and then pick $\varepsilon $ so that%
		\begin{equation*}
			(u+c)\left( \int_{0}^{1}r_1e^{-r_1t}ds\right) +u\left( \int_{1}^{T^{\ast
			}}r_1e^{-r_1t}ds\right) +(u+c)\left( \int_{T^{\ast }}^{\infty }r_1e^{-r_1t}ds\right)
			<\left( 1-\varepsilon \right) (u+c)\left( \int_{0}^{T^{\ast
			}}r_1e^{-r_1t}ds\right) ,
		\end{equation*}
		in which case $\mathbb{E}\left( \int_{t}^{\min \left\{ t+T^*,\mathbb{T}%
			\right\} }\left( 1-a_t\right) dt\mid \theta =NI\right) <1$ and $\mathbb{P}^{\tilde{%
				\sigma}}\left( \left\{ \mathbb{T}<T^{\ast }\right\} \right) <\varepsilon $
		would imply that $\tilde{\sigma}$ is a profitable deviation.
		
		\noindent{\textbf{Step 4:}} There exists $\epsilon ^{\ast }>0$ and $%
		\varepsilon ^{\ast }>0$  such that if:
		
		\begin{enumerate}
			\item $p_{t}<\underline{p}+\epsilon ^{\ast }$,
			
			\item $\mathbb{P}\left( \sup_{s>t}p_{s}\geq p^{\ast }\right)
			<\varepsilon ^{\ast }$,
			
			\item $\mathbb{E}\left( \int_{t}^{\min \left\{
				t+T^{\ast },\mathbb{T}\right\} }(1-a(t))dt \phantom{.} | \phantom{.} \theta=NI\right) \geq 1$,
			
		\end{enumerate}
		then $\mathbb{P}\left( \inf_{s>t}p_{s}<\underline{p}\right) >0$.
		
		The argument is analogous to that used in the Case 1 of Claim
		\ref*{claim_2}'s proof, and is thus omitted.
		
		\noindent{\textbf{Step 5:}} Step 3 guarantees that we can find $\varepsilon ^{\ast }>0$
		\ and $T^{\ast }$ such that if $\mathbb{P}^{\tilde{\sigma}}\left( \left\{ \mathbb{T}%
		<T^{\ast }\right\} \right) <\varepsilon ^{\ast }$ then   $\mathbb{E}\left( \int_{t}^{\min \left\{
			t+T^{\ast },\mathbb{T}\right\} }\left( 1-a_t\right) dt\mid \theta
		=NI\right) \geq 1.$ Steps 1 and 2 guarantee that we can find $\epsilon
		^{\ast }>0$ such that if $p_{t}<$ $\underline{p}+\epsilon ^{\ast }$ then $\mathbb{P}^{%
			\tilde{\sigma}}\left( \left\{ \mathbb{T}<T^{\ast }\right\} \right)
		<\varepsilon ^{\ast }.$ Therefore, using Step 4, we conclude that  $\mathbb{P}\left( \inf_{s>t}p_{s}<\underline{p}\right) >0$, which contradicts
		the definition of $\underline{p}.$
		
		\noindent{\textbf{Case 2:}} $p^{\ast }\leq \underline{p}$.
		This case is analogous to the Case 2 of Claim  \ref*{claim_2}'s proof, and is thus omitted.
	\end{proof}
	
	\begin{claim}\label{claim_5}
		In any Markovian equilibrium policy profile $(a,b)$, $\lim_{p \to 1}a(p)<1.$
	\end{claim}

	\begin{proof}
		Assume towards a contradiction that $\lim_{p \to 1}a(p)=1.$ Fix an $\varepsilon >0$. Take $T_{\varepsilon }\in\mathbb{N}$ such that $e^{-r_{1}T_{\varepsilon }}<\varepsilon $ and $\mathbb{P}\left( 
		\mathbb{T}_{\lambda }>T_{\varepsilon }\right) <\varepsilon $ where $\mathbb{T%
		}_{\lambda }$ is the random time that the next stopping opportunity arrives. Let $Z(p):=\ln(p/(1-p))$, for $p\in (0,1)$. Under the contradiction assumption, take $\hat{Z}>Z(p^{*})$ such that $z>\hat{Z}$ implies $a (z)>\frac{1}{2}$.\footnote{%
			Here, $p^*$ is the principal's equilibrium termination threshold delivered by Claim \ref*{cl:cutoffstructure}. It satisfies
			$p^*<1$, for otherwise the principal would never stop the game and the agent would choose $a_t=0$, to which the principal's best reply would be to stop when $p$ is close to 1.}  Recall from the law of motion \eqref{eq:dZt} that $Z_t$ has bounded drift. So there exists $Z_{\varepsilon }>\hat{Z}$ such that $Z_{0}\geq Z_{\varepsilon }$
		implies $\mathbb{P}\left( \inf_{t<T_{\varepsilon }}Z_{t}<\hat{Z}\mid \theta
		=NI\right) <\varepsilon.$
		
		Under the contradiction assumption, if $Z_{0}\geq
		Z_{\varepsilon }$, the payoff of the agent is no greater than 
		\begin{eqnarray*}
			&&\mathbb{P}\left( \mathbb{T}_{\lambda }\leq T_{\varepsilon }\right) \left[ 
			\begin{array}{c}
				\mathbb{P}\left( \inf_{t<T_{\varepsilon }}Z_{t}<\hat{Z}\mid \theta
				=NI\right) (u+c) \\ 
				+\mathbb{P}\left( \inf_{t<T_{\varepsilon }}Z_{t}>\hat{Z}\mid \theta
				=NI\right) \int_{0}^{T_{\varepsilon }}\lambda e^{-\lambda t}\left(
				1-e^{-r_1t}\right) \left[ \frac{1}{2}u+\frac{1}{2}(u+c)\right] dt%
			\end{array}%
			\right]  \\
			&&+\mathbb{P}\left( \mathbb{T}_{\lambda }>T_{\varepsilon }\right) (u+c) \\
			&<&\left( 1-\varepsilon \right) ^{2}\left( \frac{r_1}{r_1+\lambda }\right)
			\left( \frac{1}{2}u+\frac{1}{2}(u+c)\right) +\left[ 1-\left(
			1-\varepsilon \right) ^{2}\right] (u+c) \\
			&<&\left( \frac{r_1}{r_1+\lambda }\right) (u+c)
		\end{eqnarray*}%
		for $\varepsilon $ small, where the expression in the first line uses the fact that $a(z)>\frac{1}{2}$ for $z>\hat Z$. This leads to a contradiction as the agent can guarantee him a payoff of $\left( \frac{r_1}{r_1+\lambda }\right) (u+c)$ by never mimicking. 
	\end{proof}
	
	\begin{claim}\label{claim_6}
		In any Markovian equilibrium policy profile $(a,b)$, $\lim_{p \to 0}a(p)<1$.
	\end{claim}
	
	\begin{proof}
		Assume towards a contradiction that $\lim_{p\to 0}a(p)=1.$ Fix an $\varepsilon >0.$ Under the contradiction assumption, take $\hat{Z}<Z(p^*)$ such
		that $Z<\hat{Z}$ implies $a (Z)>\frac{1}{2}$.\footnote{%
			Note that $p^*\geq p^{**}>0$.} By Claim \ref*{claim_3} and because $Z_t$ has bounded drift, there exists $Z_{\varepsilon }<\hat Z$ such that  $
		Z_{0}<Z_{\varepsilon }$ implies $\mathbb{P}\left( \sup_{t<T_{\varepsilon
		}}Z_{t}\geq \hat{Z}\mid \theta =NI\right) <\varepsilon $ and $\mathbb{P}^{%
			\tilde{\sigma}}\left( \sup_{t<T_{\varepsilon }}Z_{t}\geq \hat{Z}\mid \theta
		=NI\right) <\varepsilon $ where $\tilde{\sigma}$ is the strategy that
		prescribes $a_{t}=0$ for every $t>0.$ 
		
		Under the contradiction assumption, if $Z_{0}<Z_{%
			\varepsilon }$, the payoff of the agent is no greater than%
		\begin{eqnarray*}
			&&\mathbb{P}\left( \sup_{t<T_{\varepsilon }}Z_{t}<\hat{Z}\mid \theta
			=NI\right) \left( \varepsilon \left[ \left( \frac{1}{2}\right) (u+c)+\left( 
			\frac{1}{2}\right) u \right] +\left( 1-\varepsilon \right) (u+c) \right) 
			\\
			&&+\mathbb{P}\left( \sup_{t<T_{\varepsilon }}Z_{t}\geq \hat{Z}\mid \theta
			=NI\right) (u+c) \\
			&<&\left( 1-\varepsilon \right) \left( \varepsilon \left[ \left( \frac{1}{2}%
			\right) (u+c)+\left( \frac{1}{2}\right) u\right] +\left( 1-\varepsilon
			\right) (u+c)\right) +\varepsilon (u+c).
		\end{eqnarray*}
		On the other hand, the strategy  $\tilde{\sigma}$ yields a payoff at least
		as large as $\left( 1-\varepsilon \right) ^{2}(u+c).$ So taking $%
		\varepsilon $ such that 
		\begin{equation*}
			\left( 1-\varepsilon \right) \left( \varepsilon \left[ \left( \frac{1}{2}%
			\right) (u+c)+\left( \frac{1}{2}\right) u\right] +\left( 1-\varepsilon
			\right) (u+c)\right) +\varepsilon (u+c)<\left( 1-\varepsilon \right)
			^{2}(u+c),
		\end{equation*}%
		we conclude that the agent can profitably deviate by playing $\tilde{\sigma}.$
	\end{proof}

	\begin{proof}[Proof of Lemma \ref{lem:fullspan}] 
		That $SP(\alpha)=(0,1)$ follows directly from Claims \ref*{claim_2} and \ref*{claim_4}. Consequently, $a(p)<1$ for all $p\in (0,1)$, for otherwise there would be an absorbing state, contradicting $SP(\alpha)=(0,1)$. By definition of an Markovian equilibrium, $a(\cdot)$ is piecewise Lipschitz, so $a(\cdot)$ only has finite discontinuities and its one-sided limit always exists. Hence, $\sup_{p\in (0,1)}a(p)<1$ if and only if $\lim_{p\to 1}a(p)<1$ and $\lim_{p\to 0}a(p)<1$, so by Claims \ref*{claim_5} and \ref*{claim_6} we are done.
	\end{proof}
	
	\begin{proof}[Proof of Lemma \ref{lemma:principalbestreply}]
		Take any Markov equilibrium $(a,b)$. By Claim \ref*{cl:cutoffstructure} and Lemma \ref{lem:fullspan}, $b$ has a cutoff structure on $(0,1)$. So we only need to argue that the cutoff belief $p^*$ satisfies $0<p^*<1$. First, since $R(p)<0$ for all $p\in (0,p^{**})$, principal's optimality requires that $b(p)=0$ for all such $p$, and so $p^*\geq p^{**}>0$. Moreover, since $W$ is bounded above by $\frac{\lambda}{r_2+\lambda}w_{NI}$,\footnote{Recall that $W(p)$ is the principal's value at $p$ conditional on the stopping opportunity \textit{not} arriving.} $R(p)>W(p)$ for all $p\in (p_H,1)$, so principal's optimality requires that $b(p)=1$ for all $p\in (p_H,1)$, and thus $p^*\leq p_H<1$.
	\end{proof}
	
	\subsubsection{Properties of $a_+^*(x)$ and $a_-^*(x)$}
	\begin{proof}[Proof of Claim \ref{cl:v*unique}]
		Recall the definition of $a_-^*(\cdot)$ and $a_+^*(\cdot)$ in \eqref{eq:a*-} and \eqref{eq:a*+}. In this proof, we focus on the properties of $a_+^*(\cdot)$, which is the more difficult case. The properties of $a_-^*(\cdot)$ can be established analogously.
		
		For ease of notation, define
		$q:=\frac{x-\frac{r_1}{r_1+\lambda}u}{\sqrt{\kappa_R}}$ and $q_{R}:=\frac{v_{R}-\frac{r_{1}}{r_{1}+\lambda}u}{\sqrt{\kappa_{R}}}.$
		Consequently, we can rewrite $a_+^*(\cdot)$ as
		$$a_+^*(x(q))=1-\frac{\frac{\sqrt{2(r_1+\lambda)}}{\snr}\phi(q)}{\frac{\sqrt{2(r_1+\lambda)}}{\snr}\phi(q_R)+\Phi(q_R)-\Phi(q)}.$$
		
		First, notice that
		\begin{equation}\label{eq:qR>mu/r}
			q_R=\sqrt{\frac{2}{r_{1}+\lambda}}\left(1-\frac{r_{1}+\lambda}{\xi_{R}\snr^{2}}\right)\snr=\sqrt{\frac{2}{r_{1}+\lambda}}\left(1-\frac{1+\xi_R}{2}\right)\snr>\frac{\snr}{\sqrt{2(r_1+\lambda)}},
		\end{equation}
		where the first equality uses the definition of $v_R$ in \eqref{eq:vR} and the fact that $\kappa_R=\frac{r_1^2c^2}{2(r_1+\lambda)\snr^2}$, and the rest follows from the fact that $\xi_R$ is the positve root of $\xi^2+\xi=\frac{2(r_1+\lambda)}{\snr^2}$ so that $\frac{r_1+\lambda}{\xi_R\snr^2}=\frac{1+\xi_R}{2}<\frac{1}{2}$.
		
		Next, notice that given \eqref{eq:qR>mu/r}, we have
		\begin{equation}\label{eq:decreasing}
			\left[\frac{\sqrt{2\left(r_{1}+\lambda\right)}}{\snr}\phi\left(q\right)+\Phi\left(q\right)\right]^{\prime}=\phi\left(q\right)\left[1-\frac{\sqrt{2\left(r_{1}+\lambda\right)}}{\snr}q\right]<0,\phantom{.} \forall q\in [q_R,\infty).
		\end{equation}
		
		Now, let us take derivative of $1-a_+^*(x(q))$ with respect to $q$. Using the fact that $\phi'(q)=-q\phi(q)$, we have
		\begin{align*}
			\left[\frac{\frac{\sqrt{2\left(r_{1}+\lambda\right)}}{\snr}\phi\left(q\right)}{\frac{\sqrt{2\left(r_{1}+\lambda\right)}}{\snr}\phi\left(q_{R}\right)+\Phi\left(q_{R}\right)-\Phi\left(q\right)}\right]^{\prime} & =\frac{-\frac{\sqrt{2\left(r_{1}+\lambda\right)}}{\snr}q\phi\left(q\right)\left[\frac{\sqrt{2\left(r_{1}+\lambda\right)}}{\snr}\phi\left(q_{R}\right)+\Phi\left(q_{R}\right)-\Phi\left(q\right)\right]+\frac{\sqrt{2\left(r_{1}+\lambda\right)}}{\snr}\phi^{2}\left(q\right)}{\left[\frac{\sqrt{2\left(r_{1}+\lambda\right)}}{\snr}\phi\left(q_{R}\right)+\Phi\left(q_{R}\right)-\Phi\left(q\right)\right]^{2}}\\
			& =\frac{\sqrt{2\left(r_{1}+\lambda\right)}}{\snr}\phi\left(q\right)\frac{-q\left[\frac{\sqrt{2\left(r_{1}+\lambda\right)}}{\snr}\phi\left(q_{R}\right)+\Phi\left(q_{R}\right)-\Phi\left(q\right)\right]+\phi\left(q\right)}{\left[\frac{\sqrt{2\left(r_{1}+\lambda\right)}}{\snr}\phi\left(q_{R}\right)+\Phi\left(q_{R}\right)-\Phi\left(q\right)\right]^{2}}.
		\end{align*}
		So the sign of $a_+^{*\prime}(x(q))$ is the same as that of 
		$$S(q):=q\left[\frac{\sqrt{2\left(r_{1}+\lambda\right)}}{\snr}\phi\left(q_{R}\right)+\Phi\left(q_{R}\right)-\Phi\left(q\right)\right]-\phi\left(q\right).$$
		Note that, given condition \eqref{eq:qR>mu/r}, we have 
		$
		S\left(q_{R}\right)  =\left[\frac{\sqrt{2\left(r_{1}+\lambda\right)}}{\snr}q_{R}-1\right]\Phi\left(q_{R}\right)>0.
		$
		Moreover, given condition \eqref{eq:decreasing}, we have 
		$
		S'(q)=\left[\frac{\sqrt{2\left(r_{1}+\lambda\right)}}{\snr}\phi\left(q_{R}\right)+\Phi\left(q_{R}\right)-\Phi\left(q\right)\right]>0.
		$
		Therefore, $S(q)>0$ for all $q\geq q_R$, implying $a_+^*(\cdot)$ is strictly increasing on $[v_R,v_L]$. 
		
		Finally, inspecting \eqref{eq:a*+} we have $a_+^*(v_R)=0$; applying condition \eqref{eq:decreasing} we have $0<a_+^*(v_R)<1$.
	\end{proof}

	\subsubsection{Translation Invariance}
	
	The analysis in Claim \ref{cl:ashaper} implies that value function and pseudo-best reply of the agent is \textit{translation-invariant}: Should the principal displace her threshold from \(z^*\) to say \(z^* + \epsilon\), then the agent's previous value and mixing behavior at \(z\) would coincide with the new value and mixing behavior at \(z+\epsilon\). More formally, this an implication of the (more general) translation invariance of agent's payoff function in the \(z\)-space.
	\begin{lemma}[Translation Invariance]
		\label{lemma:translation}
		Fix an arbitrary strategy profile \(\left(\alpha,\beta\right)\) and some \(\epsilon \in \mathbb{R}\). Consider a new profile \(\left(\alpha',\beta'\right)\), defined by \(\alpha'_t :=\left.\alpha_t\middle| \left\{Z_s - \epsilon\right\}_{s \leq t}\right.\) and \(\beta'_t:=\left.\beta_t \middle| \left\{Z_s - \epsilon\right\}_{s \leq t}\right.\). Then, the payoff of the noninvestible agent satisfies
		\[
		\mathbb{E}\left\{U_1(t,\alpha,\beta)\middle| Z_0 = z\right\} = \mathbb{E}\left\{U_1(t,\alpha',\beta')\middle| Z_0 = z+\epsilon \right\},
		\]
		almost surely, for all \(t \geq 0\) and \(z \in \mathbb{R}\).
	\end{lemma}
	
	\begin{proof}
		The law of motion of the process \(\{Z_t\}_{t \geq 0}\) from the perspective of the noninvestible agent is:
		\[
		dZ_t = \tfrac{1}{2} \snr^2 (1-\alpha_t)^2 dt + \snr (1-\alpha_t) dB_t.
		\]
		This means that, if the principal perturbs her strategy using a constant displacement, the agent can maintain his distribution of payoffs intact by imitating the perturbation.
	\end{proof}

	The payoff of the principal is not translation-invariant in the \(z\)-space. Note, however, that we can write her payoff conditioning on the type of the agent:
	\[
	U_2(t,\alpha,\beta) = p_t \mathbb{E}\left\{U_2(t,\alpha,\beta) \middle| \theta=NI \right\} + (1-p_t)\mathbb{E}\left\{U_2(t,\alpha,\beta) \middle| \theta=I \right\}.
	\]
	The conditional payoff inside the outer expectation satisfy a \textit{conditional} translation invariance property.
	\begin{lemma}[Conditional Translation Invariance]
		\label{lemma:conditionaltranslation}
		Fix an arbitrary strategy profile \(\left(\alpha,\beta\right)\) and some \(\epsilon \in \mathbb{R}\). Consider a new profile \(\left(\alpha',\beta'\right)\), defined by \(\alpha'_t :=\left.\alpha_t\middle| \left\{Z_s - \epsilon\right\}_{s \leq t}\right.\) and \(\beta'_t:=\left.\beta_t \middle| \left\{Z_s - \epsilon\right\}_{s \leq t}\right.\) for all \(t \geq 0\). Then, the conditional payoffs of the principal satisfy
		\[
		\mathbb{E}\left\{U_2(t,\alpha,\beta)\middle| \theta=NI, Z_t = z\right\} = \mathbb{E}\left\{U_2(t,\alpha',\beta')\middle| \theta=NI, Z_t = z+\epsilon \right\},
		\]
		\[
		\mathbb{E}\left\{U_2(t,\alpha,\beta)\middle| \theta=I, Z_t = z\right\} = \mathbb{E}\left\{U_2(t,\alpha',\beta')\middle|  \theta=I, Z_t = z+\epsilon \right\},
		\]
		almost surely, for all \(t \geq 0\) and \(z \in \mathbb{R}\).
	\end{lemma}
	
	\begin{proof}
		In the case of conditioning on \(\theta=NI\), the law of motion of \(\{Z_t\}_{t \geq 0}\) is as in the proof of Lemma~\ref*{lemma:translation}. In the case of conditioning on \(\theta=I\), the dynamics of \(\{Z_t\}_{t \geq 0}\) satisfies
		\begin{equation*}
			dZ_t = -\tfrac{1}{2} \snr^2 (1-\alpha_t)^2 dt + \snr (1-\alpha_t) dB_t.
		\end{equation*}
		In both cases the dynamics are linear given \(\alpha\), so if the agent perturbs his strategy using a constant displacement in \(z-\)space, the principal can maintain her payoff distribution intact by imitating the perturbation.
	\end{proof}
	
	\subsubsection{Monotonicity and Curvature of Value Functions}
	\begin{proof}[Proof of Corollary \ref{cor:convexity}]
		We first establish the claimed properties of $W(\cdot)$. From the law of motion of $p_t$ given by \eqref{eq:law-of-p}, we know that the diffusion coefficient converges to $0$ as $p\to 0$ or $p\to 1$. So it is easy to verify that $\lim_{p\to 0}W(p)=0$ and $\lim_{p\to 1}W(p)=\frac{\lambda}{r_2+\lambda}w_{NI}$. Also, since the principal can always ignore any information, $W(p)$ is bouneded below by $\underline{W}(p)\equiv \frac{\lambda}{r_2+\lambda}\max\{0,R(p)\}$. Then, from the principal's HJB given by \eqref{eq:HJBprincipal}, we always have $W''(p)\geq 0$ no matter whether $W(p)>R(p)$ or $W(p)<R(p)$; that is, $W(\cdot)$ is convex on $(0,1)$. Since $\lim_{p\to 0}W(p)=0$ and $W(\cdot)\geq 0$, $W(\cdot)$ must be (weakly) increasing at $0$, and because it is convex,  $W(\cdot)$ is increasing on $(0,1)$.
		
		Now we turn to $v(\cdot)$. Suppose first that $r_1\geq r^*$, so that in equilibrium $a(\cdot)\equiv 0$. From Claim \ref{cl:La=1} and conditions \eqref{eq:A2a=1} and \eqref{eq:A1a=1}, $v(\cdot)$ is strictly decreasing and concave on $(-\infty,z^*)$, with $\lim_{z\to-\infty}v(z)=u+c$. From Claim \ref{cl:Ra=1} and conditions \eqref{eq:B2a=1} and \eqref{eq:B1a=1}, $v(\cdot)$ is strictly decreasing and convex on $(z^*,\infty)$, with $\lim_{z\to\infty}v(z)=\frac{r_1}{r_1+\lambda}(u+c)$. Suppose now that $r_1< r^*$, so that in equilibrium $a(\cdot)$ is hump-shaped. In light of Claims \ref{cl:La=1} and \ref{cl:Ra=1} and conditions \eqref{eq:A2}, \eqref{eq:A1}, \eqref{eq:B2} and \eqref{eq:B1}, it suffices to show that $v(\cdot)$ is strictly decreasing and concave on $(z_L,z^*)$, and strictly decreasing and convex on $(z^*,z_R)$. But these properties follow immediately from Claim \ref{cl:aoptimal} and the fact that $a(\cdot)$ is hump-shaped with $0<a(z)<1$ for $z\in (z_L,z_R)$.
	\end{proof}

	\subsection{Omitted Proofs for Theorem \ref{t:zigzag-shape}}\label{onlineapp:B2}
	\begin{proof}[Proof of Lemma \ref{lem:lambda2}]
		Recall from the proof of Claim \ref{cl:ashaper} that $r_1<r^*$ implies $v_{R}<v_{L}$. Then by Claim \ref{cl:v*<u} and Corollary \ref{cor:EPmonotone2}, we are done if we can find an $\lambda_2\geq \lambda_1$ and an $\underline{r}$ such that $\lambda>\lambda_2$ and $r_1<\underline{r}$ imply that $a_-^*(u;r_1)<a_+^*(u;r_1,\lambda)$.
		
		Using \eqref{eq:a*-} and \eqref{eq:a*+}, we have
		\begin{align}\label{eq:a*-<a*+}
			& a_{-}^{*}\left(u;r_1\right)<a_{+}^{*}\left(u;r_1,\lambda\right)\nonumber\\
			\iff & \frac{\frac{\sqrt{2}}{\snr}\phi\left(0\right)}{\frac{\sqrt{2r_{1}}}{\snr}\phi\left(\frac{v_{L}-u}{\sqrt{\kappa_{L}}}\right)+\Phi\left(\frac{v_{L}-u}{\sqrt{\kappa_{L}}}\right)-\Phi\left(0\right)}>\frac{\frac{1}{\sqrt{r_{1}}}\phi\left(\sqrt{\frac{2}{r_{1}+\lambda}}\frac{\lambda\snr}{r_{1}c}u\right)}{\phi\left(\frac{v_{R}-\frac{r_{1}}{r_{1}+\lambda}u}{\sqrt{\kappa_{R}}}\right)+\frac{\snr}{\sqrt{2\left(r_{1}+\lambda\right)}}\left[\Phi\left(\frac{v_{R}-\frac{r_{1}}{r_{1}+\lambda}u}{\sqrt{\kappa_{R}}}\right)-\Phi\left(\sqrt{\frac{2}{r_{1}+\lambda}}\frac{\lambda\snr}{r_{1}c}u\right)\right]}
		\end{align}
		Since $\phi\left(\frac{v_{L}-u}{\sqrt{\kappa_{L}}}\right)\leq\phi\left(0\right)$
		and $\Phi\left(\frac{v_{L}-u}{\sqrt{\kappa_{L}}}\right)\leq1$, we
		can find a lower bound for the LHS of \eqref{eq:a*-<a*+} whenever $r_1<1$:
		\begin{equation*}
			\frac{\frac{\sqrt{2}}{\snr}\phi\left(0\right)}{\frac{\sqrt{2r_{1}}}{\snr}\phi\left(\frac{v_{L}-u}{\sqrt{\kappa_{L}}}\right)+\Phi\left(\frac{v_{L}-u}{\sqrt{\kappa_{L}}}\right)-\Phi\left(0\right)}\geq\frac{\frac{\sqrt{2}}{\snr}\phi\left(0\right)}{\frac{\sqrt{2}}{\snr}\phi\left(0\right)+1-\Phi\left(0\right)}.
		\end{equation*}
		
		Now let us find an upper bound for the RHS of \eqref{eq:a*-<a*+}. First, when $\lambda\geq 1$,
		we know 
		\begin{equation}\label{eq:first}
			\frac{1}{\sqrt{r_{1}}}\phi\left(\sqrt{\frac{2}{r_{1}+\lambda}}\frac{\lambda\snr}{r_{1}c}u\right)\leq\frac{1}{\sqrt{r_{1}}}\phi\left(\sqrt{\frac{2}{r_{1}+1}}\frac{\snr}{r_{1}c}u\right).
		\end{equation}
		Second, by direct calculation we have 
		\[
		\frac{v_{R}-\frac{r_{1}}{r_{1}+\lambda}u}{\sqrt{\kappa_{R}}}=\frac{\sqrt{2}\snr}{4}\left(\frac{3}{\sqrt{\lambda+r_{1}}}+\sqrt{\frac{1}{r_{1}+\lambda}+\frac{8}{\snr^{2}}}\right),
		\]
		and when $\lambda\geq 1$, we have 
		\[
		\frac{v_{R}-\frac{r_{1}}{r_{1}+\lambda}u}{\sqrt{\kappa_{R}}}\leq\frac{\sqrt{2}\snr}{4}\left(\frac{3}{\sqrt{1+0}}+\sqrt{\frac{1}{1+0}+\frac{8}{\snr^{2}}}\right)=\frac{\sqrt{2}}{4}\left(3\snr+\sqrt{\snr^{2}+8}\right).
		\]
		Then 
		\begin{equation}\label{eq:second}
			\phi\left(\frac{v_{R}-\frac{r_{1}}{r_{1}+\lambda}u}{\sqrt{\kappa_{R}}}\right)\geq\phi\left(\frac{\sqrt{2}}{4}\left(3\snr+\sqrt{\snr^{2}+8}\right)\right).
		\end{equation}
		Third, 
		\[
		\frac{\snr}{\sqrt{2\left(r_{1}+\lambda\right)}}\left[\Phi\left(\frac{v_{R}-\frac{r_{1}}{r_{1}+\lambda}u}{\sqrt{\kappa_{R}}}\right)-\Phi\left(\sqrt{\frac{2}{r_{1}+\lambda}}\frac{\lambda\snr}{r_{1}c}u\right)\right]\geq-\frac{\snr}{\sqrt{2\lambda}}\left(1-\Phi\left(0\right)\right).
		\]
		Apparently, there exists $\underline{\lambda}\geq\lambda_1$, such
		that for all $\lambda>\underline{\lambda}$, we have 
		\begin{equation}\label{eq:third}
			\phi\left(\frac{\sqrt{2}}{4}\left(3\snr+\sqrt{\snr^{2}+8}\right)\right)-\frac{\snr}{\sqrt{2\lambda}}\left(1-\Phi\left(0\right)\right)\geq\frac{1}{2}\phi\left(\frac{\sqrt{2}}{4}\left(3\snr+\sqrt{\snr^{2}+8}\right)\right).
		\end{equation}
		
		Defining $\lambda_{2}=\max\left\{ 1,\underline{\lambda}\right\} $ and applying conditions \eqref{eq:first}, \eqref{eq:second} and \eqref{eq:third}, we know that whenever $\lambda>\lambda_{2}$, we have the following upper bound for the RHS of \eqref{eq:a*-<a*+}: 
		\[
		\frac{\frac{1}{\sqrt{r_{1}}}\phi\left(\sqrt{\frac{2}{r_{1}+\lambda}}\frac{\lambda\snr}{r_{1}c}u\right)}{\phi\left(\frac{v_{R}-\frac{r_{1}}{r_{1}+\lambda}u}{\sqrt{\kappa_{R}}}\right)+\frac{\snr}{\sqrt{2\left(r_{1}+\lambda\right)}}\left[\Phi\left(\frac{v_{R}-\frac{r_{1}}{r_{1}+\lambda}u}{\sqrt{\kappa_{R}}}\right)-\Phi\left(\sqrt{\frac{2}{r_{1}+\lambda}}\frac{\lambda\snr}{r_{1}c}u\right)\right]}\leq\frac{\frac{1}{\sqrt{r_{1}}}\phi\left(\sqrt{\frac{2}{r_{1}+1}}\frac{\snr}{r_{1}c}u\right)}{\frac{1}{2}\phi\left(\frac{\sqrt{2}}{4}\left(3\snr+\sqrt{\snr^{2}+8}\right)\right)}.
		\]
		
		Now we compare the lower bound for the LHS of \eqref{eq:a*-<a*+} with the upper bound for the RHS of \eqref{eq:a*-<a*+}. Taking limit $r_{1}\rightarrow0$, we have 
		\[
		\lim_{r_{1}\rightarrow0}\frac{\frac{1}{\sqrt{r_{1}}}\phi\left(\sqrt{\frac{2}{r_{1}+1}}\frac{\snr}{r_{1}c}u\right)}{\frac{1}{2}\phi\left(\frac{\sqrt{2}}{4}\left(3\snr+\sqrt{\snr^{2}+8}\right)\right)}=0<\frac{\frac{\sqrt{2}}{\snr}\phi\left(0\right)}{\frac{\sqrt{2}}{\snr}\phi\left(0\right)+1-\Phi\left(0\right)},
		\]
		recalling that $\phi$ is the pdf of the standand normal distribution. So there exists $\underline{r}'>0$ such that  for all $r<\underline{r}'$,
		\[
		\frac{\frac{1}{\sqrt{r_{1}}}\phi\left(\sqrt{\frac{2}{r_{1}+1}}\frac{\snr}{r_{1}c}u\right)}{\frac{1}{2}\phi\left(\frac{\sqrt{2}}{4}\left(3\snr+\sqrt{\snr^{2}+8}\right)\right)}<\frac{\frac{\sqrt{2}}{\snr}\phi\left(0\right)}{\frac{\sqrt{2}}{\snr}\phi\left(0\right)+1-\Phi\left(0\right)}.
		\] 
		Letting $\underline{r}=\min\left\{ 1,\underline{r}',r^*\right\} $,
		we have 
		$
		a_{-}^{\star}\left(u;r_1\right)<a_{+}^{\star}\left(u;r_1,\lambda\right)
		$
		whenever $\lambda>\lambda_2$ and $r_1<\underline{r}$, as desired.
	\end{proof}
	
	\begin{proof}[Proof of Claim \ref{cl:a+bound}]
		Note that
		\begin{equation}\label{eq:a*+expand}
			a_{+}^{*}\left(u;r_{1},\lambda\right)  =1-\frac{1}{\phi\left(\frac{v_{R}-\frac{r_{1}}{r_{1}+\lambda}u}{\sqrt{\kappa_{R}}}\right)+\frac{\snr}{\sqrt{2\left(r_{1}+\lambda\right)}}\left[\Phi\left(\frac{v_{R}-\frac{r_{1}}{r_{1}+\lambda}u}{\sqrt{\kappa_{R}}}\right)-\Phi\left(\sqrt{\frac{2}{r_{1}+\lambda}}\frac{\lambda\snr}{r_{1}c}u\right)\right]}\frac{1}{\sqrt{2\pi}}e^{-\frac{\lambda^{2}}{\lambda+r_{1}}\frac{\snr^{2}}{r_{1}^{2}c^{2}}u^{2}}.
		\end{equation}
		Note also that, for all $\lambda>\lambda_1$,
		\[
		\frac{v_{R}-\frac{r_{1}}{r_{1}+\lambda}u}{\sqrt{\kappa_{R}}}=\frac{\sqrt{2}\snr}{4}\left(\frac{3}{\sqrt{\lambda+r_{1}}}+\sqrt{\frac{1}{\lambda+r_{1}}+\frac{8}{\snr^{2}}}\right)\leq\frac{\sqrt{2}\snr}{4}\left(\frac{3}{\sqrt{\lambda_{1}}}+\sqrt{\frac{1}{\lambda_{1}}+\frac{8}{\snr^{2}}}\right)
		\]
		Thus for all $\lambda>\lambda_1$,
		\begin{align}\label{eq:ineq1}
			&\phi\left(\frac{v_{R}-\frac{r_{1}}{r_{1}+\lambda}u}{\sqrt{\kappa_{R}}}\right)+\frac{\snr}{\sqrt{2\left(r_{1}+\lambda\right)}}\left[\Phi\left(\frac{v_{R}-\frac{r_{1}}{r_{1}+\lambda}u}{\sqrt{\kappa_{R}}}\right)-\Phi\left(\sqrt{\frac{2}{r_{1}+\lambda}}\frac{\lambda\snr}{r_{1}c}u\right)\right]\nonumber\\
			\geq\quad&\phi\left(\frac{\sqrt{2}\snr}{4}\left(\frac{3}{\sqrt{\lambda_{1}}}+\sqrt{\frac{1}{\lambda_{1}}+\frac{8}{\snr^{2}}}\right)\right)-\frac{\snr}{\sqrt{2\lambda}}.
		\end{align}
		Let $A'$ to be such that 
		\[
		\frac{1}{\sqrt{2\pi}A'}=\frac{1}{2}\phi\left(\frac{\sqrt{2}\snr}{4}\left(\frac{3}{\sqrt{\lambda_{1}}}+\sqrt{\frac{1}{\lambda_{1}}+\frac{8}{\snr^{2}}}\right)\right).
		\]
		Since $\lim_{\lambda\rightarrow\infty}\frac{\snr}{\sqrt{2\lambda}}=0$,
		there must exist $\lambda_{3}^{\prime}\geq \lambda_{1}$, such that  for all $\lambda>\lambda_{3}^{\prime}$,
		\begin{equation}\label{eq:ineq2}
			\phi\left(\frac{\sqrt{2}\snr}{4}\left(\frac{3}{\sqrt{\lambda_{1}}}+\sqrt{\frac{1}{\lambda_{1}}+\frac{8}{\snr^{2}}}\right)\right)-\frac{\snr}{\sqrt{2\lambda}}>\frac{1}{\sqrt{2\pi}A'}.
		\end{equation}
		Finally, let $A=\max\{A',2\}$. Conditions \eqref{eq:a*+expand}, \eqref{eq:ineq1} and \eqref{eq:ineq2} then tell us that whenever $\lambda>\lambda_{3}^{\prime}$,
		we have 
		\[
		a_{+}^{*}\left(u;r_{1},\lambda\right) >1-A'e^{-\frac{\lambda^{2}}{\lambda+r_{1}}\frac{\snr^{2}}{r_{1}^{2}c^{2}}u^{2}}\geq1-Ae^{-\frac{\lambda^{2}}{\lambda+r_{1}}\frac{\snr^{2}}{r_{1}^{2}c^{2}}u^{2}},
		\]
		as desired.
	\end{proof}
	
	\subsection{Omitted Proofs for Theorem \ref{t:noisehelps}}\label{onlineapp:B3}
	\begin{proof}[Proof of Claim \ref{lem:learninglemma2}]
		The proof is almost identical to Lemma \ref*{lemma_a7}'s proof in the next section, and is thus omitted.
	\end{proof}
	
	\subsection{Patient Limit: Toward a Proof of Theorem \ref{t:patientlimit}}\label{app:patientlimit}
	In this section, we prove Theorem \ref{t:patientlimit} which is about the convergence of equilibrium value functions when players get arbitrarily patient at comparable rates.

	For each \(n\in \mathbb{N}\), take the unique Markov equilibrium \(\left( a_{n}, b_{n}\right)\) associated with the discount factor \(r_{i,n}\) for \(i=1,2\). 
	Assume that \(\lim_{n \to \infty}  r_{i,n}=0\) and \(\lim_{n \to \infty} \left(r_{2,n}/r_{1,n}\right) :=\chi \in (0,\infty)\). Let \(V_{n}\left( \cdot \right)\) be the agent's value function in the equilibrium \(\left( a_{n},b_{n}\right)\) and \(W_{n}\left(\cdot \right)\) be the principal's value function. We will often use $z\equiv \log(p/1-p)$ as state variable when analyzing the agent's behavior. When doing so, we denote by $v_n(z):=V_n(p(z))$ the agent's value function in the $z-$space. Write \(z_{n}^{\ast }\) for the principal's equilibrium cutoff. Write \(z_{L,n}\) for the infimum belief \(z\) at which the agent plays \(a_{n}\left( z\right) >0\) and write \(z_{R,n}\) for the supremum. Write $\mathbb{T}$ for the equilibrium stopping time that stops the play of the game. \textit{Without labeling explicitly, we note that the distribution of $\mathbb{T}$ depends on $n$ and the current state $z$.}
	For \(i=1,2\), let \(\mathbb{E}_{n}^{\theta}\left\{ e^{-r_{i,n}\mathbb{T}}\right\}\) be the expected discount factor when the stopping action is taken in the equilibrium \(\left( a_{n}, b_{n}\right)\) discounted at rate \(r_{i,n}\) and given the equilibrium strategy of type \(\theta\in \left\{NI,I\right\}\).\footnote{Here, we interpret the strategy of the investible type as always setting $a=1$.} When the game starts at state \(z\), let 
	\[
	\mathbb{E}_{n}\left\{e^{-r_{i,n}\mathbb{T}}\right\} :=p(z) \mathbb{E}_{n}^{NI}\left\{ e^{-r_{i,n}\mathbb{T}}\right\} +(1-p(z))\mathbb{E}_{n}^{I}\left\{ e^{-r_{i,n}\mathbb{T}}\right\}.
	\]
	
	
	\begin{claim}
		\label{claim_a0} 
		Take \(\left( r_{1},r_{2}\right) \in \mathbb{R}_{++}^{2}\) and let $\tau$ be any stopping time. Assume that \(\mathbb{E}\left\{ e^{-r_{1}\tau}\right\}
		=\xi \in \left(0,1\right)\).
		
		i) If \(r_{2}\leq r_{1}\) then
		\[
		\xi \leq \mathbb{E}\left\{ e^{-r_{2}\tau}\right\} \leq \xi ^{\left({r_{2}}/{r_{1}}\right) }.
		\]
		
		ii) If $r_{2}>r_{1}$ then
		\[
		\xi ^{\left( {r_{2}}/{r_{1}}\right) }\leq \mathbb{E}\left\{e^{-r_{2}\tau}\right\} \leq \xi .
		\]
		Moreover, for each \(\xi \in (0,1)\) and any inequality above, there exists a distribution over stopping times for which this inequality is tight.
	\end{claim}
	
	\begin{proof}
		Assume that \(\mathbb{E}\left\{e^{-r_{1}\tau}\right\} =\xi \in \left(0,1\right)\). Let $F$ be the CDF of $\tau$. Let \(y=e^{-r_{1}\tau}\) and \(H\) be its CDF. We have \(\tau=-(\log y)/r_1\) and hence 
		$\int_{0}^{\infty }e^{-r_{2}t}dF(t)=\int_{0}^{1}y^{\left({r_{2}}/{r_{1}}\right) }dH(y).$ 
		
		i) Suppose that $r_2\leq r_1$. On the one hand, since \(y^{\left({r_{2}}/{r_{1}}\right)}\) is concave, we have
		\[
		\int_{0}^{1}y^{\left({r_{2}}/{r_{1}}\right) }dH(y)\leq \left(\int_{0}^{1}y dH(y)\right) ^{\left( {r_{2}}/{r_{1}}\right) }=\xi ^{\left({r_{2}}/{r_{1}}\right) },
		\]
		with equality if \(H\) has an atom of mass one.
		
		On the other hand, \(y^{\left({r_{2}}/{r_{1}}\right) }\geq y\) for every \(y\in \left[ 0,1\right]\) and hence \(\int_{0}^{1}y^{\left({r_{2}}/{r_{1}}\right) }dH(y)\geq \int_{0}^{1}ydH(y)\), with equality if the support of \(H\) is \(\left\{ 0,1\right\}\).
		
		ii) Suppose that $r_2\geq r_1$. On the one hand, since \(y^{\left({r_{2}}/{r_{1}}\right)}\) is convex, we have 
		\[
		\int_{0}^{1}y^{\left({r_{2}}/{r_{1}}\right) }dH(y)\geq \left(\int_{0}^{1}ydH(y)\right) ^{\left( {r_{2}}/{r_{1}}\right) }=\xi ^{\left({r_{2}}/{r_{1}}\right) },
		\]
		with equality if \(H\) has an atom of mass one.
		
		On the other hand, \(y^{\left({r_{2}}/{r_{1}}\right) }\leq y\) for every \(y\in \left[ 0,1\right]\) and hence \(\int_{0}^{1}y^{\left({r_{2}}/{r_{1}}\right) }dH(y)\leq \int_{0}^{1} y dH(y)\), with equality if the support of \(H\) is \(\left\{ 0,1\right\}\).
	\end{proof}

	\begin{claim}
		\label{claim_a2}
		For every \(\varepsilon >0\) there exists \(z^{\dag }\in \mathbb{R}\) and $\tilde{n}_{1}\in \mathbb{N}$ such that, if \(z\geq z^{\dag }\) and \(n\geq \tilde{n}_{1}\), then the continuation payoff of the agent at \(z\) is less than \(\varepsilon\) in the equilibrium $(a_n,b_n)$.
	\end{claim}
	
	\begin{proof}
		Otherwise we can find a sequence of equilibria \(\left( a _{n},b_{n}\right)\) starting at \(\left( z_{n}\right) \to +\infty\) in which the agent obtains a payoff weakly greater than \(\varepsilon\). Since the agent's equilibrium payoff is bounded above by $(u+c)\left[ 1-\mathbb{E}_{n}^{NI}\left( e^{-r_{1,n}\mathbb{T}}\right) \right]$, we have
		\[
		(u+c)\left[ 1-\mathbb{E}_{n}^{NI}\left( e^{-r_{1,n}\mathbb{T}}\right) \right] \geq \varepsilon \hspace{10pt}\Rightarrow \hspace{10pt} \mathbb{E}_{n}^{NI}\left( e^{-r_{1,n}\mathbb{T}}\right) \leq \left( 1-\tfrac{\varepsilon}{u+c}\right),
		\]
		and hence the principal's payoff in equilibrium $\left( a_{n},b_{n}\right) $ at $z_{n}$ is at most
		\[
		\max \left\{ \left( 1-\tfrac{\varepsilon }{u+c}\right) ,\left( 1-\tfrac{\varepsilon }{u+c}\right) ^{\left({r_{2,n}}/{r_{1,n}}\right) }\right\} p\left( z_n\right) w_{NI},
		\]
		which is always strictly less than $w_{NI}$. Meanwhile, since $r_{2,n}\to 0$ and $z_n\to \infty$ as $n\to \infty$, the principal's payoff at $z_n$ by terminating the relationship in the first opportunity satisfies
		\begin{equation*}
			\lim_{n\to \infty}\left(\tfrac{\lambda}{r_{2,n}+\lambda}\right)\left[(1-p(z_n))w_{I}+p(z_n)w_{NI}\right]= w_{NI}.
		\end{equation*}
		So the principal  has a profitable deviation when \(n\) is sufficiently large.
	\end{proof}
	We assume that \(n\geq \tilde{n}_{1}\) for the remainder of this proof.
	
	\begin{claim}
		\label{lemma_a4}
		For every fixed \(z_0\), we have \(\limsup_{n \to \infty} v_n\left(z_0\right) \leq u\). 
	\end{claim}
	
	\begin{proof}
		Take any small \(\varepsilon \in \left( 0,u/2\right)\). For each \(n\in \mathbb{N}\), let 
		$z_{\varepsilon }^{n}:=\inf \left\{ z \middle| a_{n}\left( z\right)=1-\varepsilon \right\}.$ 
		There are two cases to consider. Let \(z^{\dag }\) be defined and delivered by Claim \ref*{claim_a2}. Every sequence can be
		split into (at most) two subsequence, each one of them satisfying one of the cases below.
		
		\textbf{Case 1 }\(z_{\varepsilon }^{n}\leq z^{\dag }\) for every \(n\in \mathbb{N}\). 
		
		In this case, take \(m\in \mathbb{N}\) such that \(z^{\dag }-m<z_{0}\) and let \(z_{0}^{n}:=z_{\varepsilon }^{n}-m\). Since \(v_{n}(\cdot)\) is decreasing, it suffices to show that \(\limsup_{n \to \infty} v_{n}\left(z_{0}^{n}\right) \leq u\).
		
		Take any \(\zeta >0\). Suppose that the game starts at \(z_{0}^{n}\) and consider the stopping time \(\mathbb{\hat{T}}_{n}\) that stops the play of the game at the first time \(Z_n(t) =z_{\varepsilon }^{n}\) (setting \(\mathbb{\hat{T}}_{n}=+\infty\) if this event does not happen in finite time). Note that \(Z_{n}\left( t\right)\) is a submartingale under the strategy of the noninvestible type and that  \(a_{n}\left( z\right) \leq 1-\varepsilon\) with probability one before \(\mathbb{\hat{T}}_{n}\). Using this observation and $Z_n(t)$'s law of motion \eqref{eq:dZt}, it is straightforward to show that \(\mathbb{\hat{T}}_{n}<+\infty\) with probability one under the strategy of the noninvestible type and that  \(\mathbb{E}_{n}^{NI}\left[e^{-r_{1,n}\mathbb{\hat{T}}_{n}}\right] \rightarrow 1\). Take \(n^{\ast \ast
		}\in \mathbb{N}\) for which \(n>n^{\ast \ast }\) implies \(\mathbb{\ E}_{n}^{NI}\left[e^{-r_{1,n}\mathbb{\hat{T}}_{n}}\right] >1-\varepsilon\). Next notice that, at the state \(z_{\varepsilon }^{n}\), $v_n$ is decreasing and concave (by Corollary \ref{cor:convexity}), and hence
		\begin{align*}
			r_{1,n} v_{n}\left( z_{\varepsilon }^{n}\right) &= r_{1,n}\left[ u+\left( 1-a_{n}\left( z_{\varepsilon}^{n}\right) \right) c\right]+\frac{1}{2}\snr^2\left[ 1-a_{n}\left(z_{\varepsilon }^{n}\right) \right] ^{2}\left[ v^{\prime }\left( a_{n}\left(z_{\varepsilon }^{n}\right) \right) +v^{\prime \prime }\left( a_{n}\left(z_{\varepsilon }^{n}\right) \right) \right]\\
			&\leq r_{1,n}\left[ u+\left( 1-a_{n}\left( z_{\varepsilon}^{n}\right) \right) c\right],
		\end{align*}
		which implies  
		$v_{n}\left( z_{n}^\varepsilon\right) \leq u+\varepsilon c,$ 
		because $a_n(z_n^\varepsilon)=1-\varepsilon$. 
		It follows that the payoff of the noninvestible type converges to a number not greater than
		$\left( 1-\varepsilon \right) \left( u+\varepsilon c\right) +\varepsilon (u+c),$ 
		which proves the result as \(\varepsilon\) is arbitrary.
		
		\textbf{Case 2 }\(z_{\varepsilon }^{n}>z^{\dag }\) for every \(n\in \mathbb{N}\).
		
		We may assume that \(z_{0}<z^{\dag }\) for every \(n\) as otherwise the claim follows from Claim \ref*{claim_a2}. Suppose that the game starts at \(z_{0}\) and consider the stopping time \(\mathbb{\hat{T}}_{n}\) that stops the play of the game at \(z^{\dag }\). As in Case 1, we have \(\mathbb{\hat{T}}_{n}<+\infty\) with probability one under the noninvestible-type's strategy and \(\mathbb{E}_{n}^{NI}\left\{e^{-r_{1,n}\mathbb{\hat{T}}_{n}}\right\}\rightarrow 1\). Since \(\limsup_{n \to \infty} v_{n}\left( z^{\dag }\right) \leq \varepsilon < u/2\), the rest of the proof follows the same argument as in Case 1.
	\end{proof}
	
	\begin{claim}
		\label{lemma_a3}
		\(\lim_{n \to \infty} z_{L,n} =-\infty\). 
	\end{claim}
	\begin{proof}
		The proof follows verbatim from Claim \ref{cl:zLn->-infty}'s proof.
	\end{proof}

	\begin{lemma}
		\label{lemma_a5}
		For every $z_0< \liminf z_n^*$, we have \(\lim_{n \to \infty} a_{n}\left( z_{0}\right) =1\).
	\end{lemma}
	\begin{proof}
		By Claim \ref*{lemma_a3}, $z_0\in (z_{L,n},z_n^*)$ for $n$ sufficiently large. 
		Then from condition \eqref{eq:a+a'}, we know that \(a_{n}(\cdot )\) eventually satisfies the following differential equation 
		\begin{equation}
			\label{ode}
			a_{n}^{\prime }\left( z\right) = 1-a_{n}(z)-2\left( \tfrac{v_{n}(z)-u}{c}\right).
		\end{equation}
		
		Assume toward a contradiction that we can find a subsequence such that \(\lim_{n\rightarrow \infty }a_{n}\left( z_{0}\right) =\bar{a}<1\). Take \(m\in\mathbb{N}\) such that \(\frac{1-\bar{a}}{4}m>2\). Claim \ref*{lemma_a3} implies that $[z_0-m,z_0]\subset (z_{L,n},z_n^*)$ for $n$ sufficiently large.  Claim \ref*{lemma_a4} and the monotonicity of $v_n(\cdot)$ imply that we can find \(n^{\dag }\in \mathbb{N}\) such that for every \(n\geq n^{\dag}\), for every \(z\in \left[ z_{0}-m,z_{0}\right]\), we have \(2\left( \frac{v_{n}(z)-u}{c}\right) <\frac{1-\bar{a}}{4}\). Given the contradiction assumption, we can find \(n^{\dag }\in \mathbb{N}\) such that for every \(n\geq n^{\dag }\) we have \(a_{n}(z_{0})<\frac{1+\bar{a}}{2}\). Since \(a_{n}\left( \cdot \right)\) is strictly increasing on $[z_0-m,z_0]$, this implies \(a_{n}\left( z\right) <\frac{1+\bar{a}}{2}\) for all \(z\in \left[ z_{0}-m,z_{0}\right]\). So (\ref{ode}) implies \(a_{n}^{\prime }\left( z\right) >\frac{1-\bar{a}}{4}\) for all \(z\in \left[ z_{0}-m,z_{0}\right]\) and hence 
		\[
		a_{n}\left( z_{0}-m\right) <a_{n}\left( z_{0}\right) - \tfrac{1-\bar{a}}{4}m< a_{n}\left( z_{0}\right) -2 < 0,
		\]
		which leads to a contradiction as \(a_{n}\) is bounded below by \(0\).
	\end{proof}

	\begin{lemma}
		\label{lemma_a6}
		For every \(z_{0}<\liminf_{n \to \infty} z_{n}^{\ast}\), we have \(\lim_{n \to \infty} v_{n}\left( z_{0}\right) =u\).
	\end{lemma}
	
	\begin{proof}
		By Claim \ref*{lemma_a3}, $z_0\in (z_{L,n},z_n^*)$ for $n$ sufficiently large.  Take \(\vartheta >0\) such that \(z_{0}+2\vartheta <\lim \inf z_{n}^{\ast }\) and, taking a subsequence if necessary, assume
		that \(z_{0}+\vartheta <z_{n}^{\ast }\) for each one of its elements.
		
		Assume toward a contradiction, taking a subsequence if necessary, that \(\lim_{n \to \infty} v_{n}\left( z_{0}\right) <u-\varepsilon\), for some \(\varepsilon >0\). Because $v_n(\cdot)$ is strictly decreasing, we may take \(n^{\ast }\) such that \(n\geq n^{\ast }\) implies \(v_{n}\left( z\right) <u-\frac{\varepsilon }{2}\) for all \(z\in \left[ z_{0},z_{0}+\frac{\vartheta }{2}\right]\). In this case, we have 
		\[
		a_{n}^{\prime }\left( z\right) =1-a_{n}(z)-2\left( \tfrac{v_{n}(z)-u}{c}\right) \geq \tfrac{\varepsilon}{c}
		\]
		for every \(z\in \left[ z_{0},z_{0}+\frac{\vartheta }{2}\right]\). This implies that \(\limsup_n a_{n}\left( z_{0}\right) \leq 1-\left( \frac{\vartheta}{2}\right) \frac{\varepsilon }{c}\), contradicting Lemma~\ref*{lemma_a5}.
	\end{proof}

	\begin{lemma}
		\label{lemma_a7}
		Fix a prior \(p_{0}\in \left( 0,1\right)\) and some \(\bar{p}\in \left( p_{0},1\right)\). For each \(r>0\), consider an adapted Markov function \(\alpha_{r}\left( \cdot \right)\) and a belief process defined by substituting $\alpha_r(\cdot)$ into \eqref{eq:dptagent}. Take \(\varepsilon >0\) and let \(\mathbb{\bar{T}}\) be the random time that stops the play in the first time that \(p\geq \bar{p}\). Then we have:
		\[
		\limsup_{r\downarrow 0}\mathbb{E}^{NI}\left\{ r\int_{0}^{\mathbb{\bar{T}}}e^{-rt}\mathbb{I}_{\left\{ \alpha _{r}\left( p_{t}\right) \leq 1-\varepsilon \right\} }dt\right\} =0.
		\]
	\end{lemma}
	
	\begin{proof}
		Take a small \(\epsilon >0\). Next, take \(\zeta >0\) and let \(\mathbb{T}^{\zeta }\) be the stopping time that stops the play in the first time that the
		posterior reaches \(\left( \zeta ,\bar{p}\right)^{c}\). Using the margingale property of beliefs whose law of motion is given by \eqref{eq:dptagent}, it is straightforward to show that we can take \(\zeta\) small enough so that
		\[
		\mathbb{P}^{NI}\left\{ \mathbb{T}^{\zeta }<\infty ,p(\mathbb{T}^{\zeta})=\zeta \right\} <\frac{\epsilon }{2}.
		\]
		Therefore, we have:
		\begin{align*}
			\mathbb{E}^{NI}\left\{ r\int_{0}^{\mathbb{\bar{T}}}e^{-rt}\mathbb{I}_{\left\{ \alpha _{r}\left( p_{t}\right) \leq 1-\varepsilon \right\}}dt\right\} = \phantom{.}&\mathbb{E}^{NI}\left\{ r\int_{0}^{\mathbb{\bar{T}}}e^{-rt}\mathbb{I}_{\left\{ \inf_{t\leq \mathbb{\bar{T}}}p_{t}\leq \zeta \right\} }\mathbb{I}_{\left\{\alpha _{r}\left( p_{t}\right) \leq 1-\varepsilon \right\}}dt\right\}\\
			& +\mathbb{E}^{NI}\left\{ r\int_{0}^{\mathbb{\bar{T}}}e^{-rt}\mathbb{I}_{\left\{ \inf_{t\leq \mathbb{\bar{T}}}p_{t}>\zeta \right\} }\mathbb{I}_{\left\{ \alpha _{r}\left( p_{t}\right) \leq 1-\varepsilon \right\}}dt\right\}\\
			\leq \phantom{.}&\frac{\epsilon }{2}+\mathbb{E}^{NI}\left\{ r\int_{0}^{\mathbb{\bar{T}}}e^{-rt}\mathbb{I}_{\left\{ \inf_{t\leq \mathbb{\bar{T}}}p_{t}>\zeta \right\} }\mathbb{I}_{\left\{ \alpha _{r}\left( p_{t}\right)\leq 1-\varepsilon \right\} }dt\right\}\\
			\leq\phantom{.}& \frac{\epsilon }{2}+\mathbb{E}^{NI}\left\{ r\int_{0}^{\mathbb{T}^{\zeta }}e^{-rt}\mathbb{I}_{\left\{ \alpha _{r}\left( p_{t}\right)\leq 1-\varepsilon \right\} }dt\right\}\\
			\leq\phantom{.}& \frac{\epsilon }{2}+\mathbb{E}^{NI}\left\{ r\int_{0}^{\mathbb{T}^{\zeta }}\mathbb{I}_{\left\{ \alpha _{r}\left( p_{t}\right) \leq 1-\varepsilon \right\} }dt\right\}.
		\end{align*}
		We must then show that
		$\limsup_{r\downarrow 0}\mathbb{E}^{NI}\left\{ r\int_{0}^{\mathbb{T}^{\zeta }}\mathbb{I}_{\left\{ \alpha _{r}\left( p_{t}\right) \leq 1-\varepsilon \right\} }dt\right\} <\frac{\epsilon }{2}.$ 
		Let \(\xi_r (t)\) be a function that is \(1\) whenever \(a_r \left( p_{t}\right)\leq 1-\varepsilon\) and \(0\) otherwise. It suffices to show that 
		$\limsup_{r\downarrow 0}r\mathbb{E}^{NI}\left\{ \int_{0}^{\mathbb{T}^{\zeta }}\xi_r (t)dt\right\} <\frac{\epsilon }{2}.$ 
		
		For that we will consider a different stopping time \(\mathbb{T}^{\ast }\) and a new process \(\xi_r ^{\ast }(t)\) which are built from \(\mathbb{T}^{\zeta }\)
		and \(\xi_r (t)\) in the following way. 
		Whenever \(\mathbb{T}^{\zeta }<\infty\) and \(\int_{0}^{\mathbb{T}^{\zeta }}\xi_r(t)dt\in \left( m-1,m\right)\) for some \(m\in \mathbb{N}\), we will set 
		\[
		\xi_r ^{\ast }(t) :=
		\begin{cases}
			\xi_r (t) & t\leq \mathbb{T}^{\zeta },\\ 
			1 & t>\mathbb{T}^{\zeta }.
		\end{cases}
		\]
		We will also set \(\mathbb{T}^{\ast }:=\mathbb{T}^{\zeta }+\tilde{t}\), where \(\tilde{t}\) is defined by \(\int_{0}^{\mathbb{T}^{\zeta }}\xi_r (t)dt+\tilde{t} =m\). Whenever \(\mathbb{T}^{\zeta }<+\infty\) and \(\int_{0}^{\mathbb{T}^{\zeta }}\xi_r(t)dt = m-1\) for some \(m\in \mathbb{N}\), we set \(\xi_r ^{\ast }(t):=\xi_r (t)\) and \(\mathbb{T}^{\ast }:=\mathbb{T}^{\zeta}\). 
		Clearly it suffices to show that 
		\[
		\limsup_{r\downarrow 0}\mathbb{E}^{NI}\left\{ r \int_{0}^{\mathbb{T}^{\ast }}\xi_r ^{\ast }(t)dt\right\} <\frac{\epsilon }{2}.
		\]
		Next, we build a family of stochastic processes \(\left\{ \xi _{r,m}^{\ast}(t)\right\} _{m\in \mathbb{N}}\) from \(\xi_r ^{\ast }(t)\) by setting
		\[
		\xi _{r,m}^{\ast }(t):=
		\begin{cases}
			\xi_r ^{\ast }(t) & \int_{0}^{t}\xi ^{\ast }(t)dt\in (m-1,m],\\ 
			0 & \text{ otherwise}.
		\end{cases}
		\]
		This immediately implies that 
		$r\mathbb{E}^{NI}\left\{ \int_{0}^{\mathbb{T}^{\ast }}\xi_r ^{\ast}(t)dt\right\} =r\sum_{m=1}^{\infty }\mathbb{E}^{NI}\left\{\int_{0}^{\mathbb{T}^{\ast }}\xi _{r,m}^{\ast }(t)dt\right\}.$
		
		Next, observe that, conditional on \(\theta = NI\), \(p_{t}\) is a bounded submartingale. Thus, for any adapted function \(\tilde{\xi}\left( t\right)\in \left\{0,1\right\}\) and any stopping time \(\mathbb{\tilde{T}}\), we have 
		\begin{align*}
			1 \geq p_{\mathbb{T}}-p_{0}=\mathbb{E}^{NI}\left[ \int_{0}^{\mathbb{T}}dp_{t}\right] &= \mathbb{E}^{NI}\left[ \int_{0}^{\mathbb{T}}\tilde{\xi}%
			\left( t\right) dp_{t}\right] +\mathbb{E}^{NI}\left[ \int_{0}^{\mathbb{T}%
			}\left( 1-\tilde{\xi}\left( t\right) \right) dp_{t}\right] \\
			&\geq \mathbb{E}^{NI}\left[ \int_{0}^{\mathbb{T}}\tilde{\xi}\left( t\right)
			dp_{t}\right] 
		\end{align*}
		because $\left( 1-\tilde{\xi}\left( t\right) \right) $ being an adapted
		process and $p_t$ being a
		submartingale jointly imply that $\mathbb{E}^{NI}\left[ \int_{0}^{\mathbb{T}%
		}\left( 1-\tilde{\xi}\left( t\right) \right) dp_{t}\right] \geq 0$. As a result, we have
		\begin{equation}\label{ineq1}
			1\geq \mathbb{E}^{NI}\left[ \int_{0}^{\mathbb{T}}\xi _{r}^{\ast }\left(
			t\right) dp_{t}\right] =\sum_{m=1}^{\infty }\mathbb{E}^{NI}\left[ \int_{0}^{%
				\mathbb{T}}\xi _{r,m}^{\ast }\left( t\right) dp_{t}\right].
		\end{equation}
		
		Next, since $0<\zeta <\bar{p}<1,$ from condition \eqref{eq:dptagent} it is straightforward to show that there
		exists a positive constant $\vartheta >0$ such that, for any $m\in \mathbb{N}$, we have
		\begin{equation}
			\label{ineq2}
			\mathbb{E}^{NI}\left\{ \int_{0}^{\mathbb{T}^{\ast }}\xi _{r,m}^{\ast}(t)dp_{t}\right\} \geq \vartheta \mathbb{E}^{NI}\left\{ \int_{0}^{\mathbb{T}^{\ast }}\xi _{r,m}^{\ast }(t)dt\right\}.
		\end{equation}
		
		Therefore, combining (\ref{ineq1}) and (\ref{ineq2}) we have 
		\[
		\sum_{m=1}^{\infty }\mathbb{E}^{NI}\left\{ \int_{0}^{\mathbb{T}^{\ast }}\xi _{r,m}^{\ast }(t)dt\right\} \leq \frac{1}{\vartheta }\sum_{m=1}^{\infty }\mathbb{E}^{NI}\left\{ \int_{0}^{\mathbb{T}^{\ast }}\xi _{r,m}^{\ast }(t)dp_{t}\right\} \leq \frac{1}{\vartheta },
		\]
		implying that \(\sum_{m=1}^{\infty }r\mathbb{E}^{NI}\left\{\int_{0}^{\mathbb{T}^{\ast }}\xi _{r,m}^{\ast }(t)dt\right\} \leq \frac{r}{\vartheta }\), which is smaller than \(\frac{\epsilon }{2}\) when $r$ is sufficiently small.
	\end{proof}

	\begin{lemma}
		\label{lemma_a8}
		\(\lim_{n \to \infty} z_{n}^{\ast }=z^{**}\).
	\end{lemma}
	
	\begin{proof}
		Suppose toward a contradiction that we can find a subsequence for which \(\lim_{n \to \infty} z_{n}^{\ast }:=\bar{z}>z^{**}\). Let \(z^{m}\) be the midpoint between \(\bar{z}\) and \(z^{**}\). Take \(\varepsilon >0\). Consider the game starting at \(z^{m }\). Notice that Lemma~\ref*{lemma_a6} implies that \(\lim_{n \to \infty} v_{n}\left( z^{m }\right)=u\), while Lemma~\ref*{lemma_a7} implies that, for each \(\nu >0\), we have 
		\[
		\lim_{n \to \infty} \mathbb{E}_{n}^{NI}\left\{ r_{1,n}\int_{0}^{\mathbb{T}}e^{-r_{1,n}t}\mathbb{I}_{\left\{ a _{n}\left( z_{t}\right) \leq 1-\nu\right\} }dt\right\} =0.
		\]
		These two observations imply that \(\lim_{n \to \infty}\mathbb{E}_{n}^{NI}\left(e^{-r_{1,n}\mathbb{T}}\right) =0\), which, by the same argument as Claim \ref*{claim_a0}'s proof, implies that \(\lim_{n \to \infty} \mathbb{E}^{NI}_n\left( e^{-r_{2,n}\mathbb{T}}\right) =0\); that is, conditional on $\theta=NI$, the principal derives zero discounted payoff from the game. It follows that the principal obtains a limit payoff bounded above by zero at \(z^{m}\). But then, for $n$ sufficiently large, if the stopping opportunity arrives at $z=z^m$, the principal can profitably deviate by stopping the game to obtain \(p\left( z^{m }\right)w_{NI}+\left( 1-p\left( z^{m }\right) \right) w_{I}>0\), a contradiction.
	\end{proof}

	\begin{lemma}
		\label{lemma_a9}
		For every \(z_0>z^{**}\) and \(i=1,2\), we have \(\lim_{n \to \infty} \mathbb{E}_{n}\left\{e^{-r_{i,n}\mathbb{T}}\right\} =1\).
	\end{lemma}
	
	\begin{proof}
		Fix $z_0>z^{**}$. By Claim \ref*{claim_a0}, it suffices to show that \(\lim_{n \to \infty} \mathbb{E}_{n}\left\{ e^{-r_{2,n}\mathbb{T}}\right\} =1\). Taking a subsequence if necessary, assume toward a contradiction that \(\lim_{n \to \infty} \mathbb{E}_{n}\left\{e^{-r_{2,n}\mathbb{T}}\right\} <1\).
		
		Let \(\tilde{\tau}\) be the stopping time that stops the play in the first time that either the state reaches \(\left[ 0,p\left( z_{n}^{\ast }\right)\right]\) or when \(\mathbb{T}\) happens. Let $x=e^{-r_{2,n}t}$. Let \(Q_{n}\) be the distribution of \(p_{\tilde{\tau}}\) and  \(H_{n}\left( \cdot \mid p_{\tilde{\tau}}\right)\) be the conditional distribution of $x$ given $p_{\tilde{\tau}}$.
		
		\medskip
		\noindent{\textbf{Step 1.}} We show that the contradiction assumption implies that, the discounted amount of time that the relationship continues with beliefs close to $p(z_n^*)$ is nonnegligible (i.e., condition \eqref{limit*} holds).
		
		Note that
		\begin{equation*}
			\label{wn}
			W_{n}\left( p\left( z_0\right) \right) =\int_{p\left( z_{n}^{\ast }\right)}^{1}\int_{0}^{1} x \left[ \mathbb{I}_{\left\{ p_{\tilde{\tau}}>p\left(z_{n}^{\ast }\right) \right\} }R\left( p_{\tilde{\tau}}\right) +\mathbb{I}_{\left\{ p_{\tilde{\tau}}\leq p\left( z_{n}^{\ast }\right) \right\}}W_{n}\left( p\left( z_{n}^{\ast }\right) \right) \right] H_{n}\left( dx\mid p_{\tilde{\tau}}\right) dQ_{n}(dp_{\tilde{\tau}}).
		\end{equation*}
		Because \(\lim_{n \to \infty} W_n\left( p\left( z_{n}^{\ast }\right) \right) = \lim_{n\to\infty}R \left(p\left( z_n^*\right) \right) =0\), we have
		\begin{equation}
			\label{wn}
			\limsup_{n\to\infty}W_n(p(z_0))= \limsup_{n\to\infty} \int_{p\left( z_{n}^{\ast }\right) }^{1}\int_{0}^{1} x R\left( p_{\tilde{\tau}}\right) H_{n}\left( dx\mid p_{\tilde{\tau}}\right) Q_{n}(dp_{\tilde{\tau}}).
		\end{equation}
		Moreover, since $R(p^{**})=0$ and $p(z_n^{*})\to p^{**}$ (by Lemma \ref*{lemma_a8}), for every \(\varepsilon >0\) there exists \(\zeta >0\) such that when $n$ is sufficiently large, \(R\left( p\right) >\zeta\) for every \(p>p\left( z_{n}^{\ast }\right)+\varepsilon\). Combining this observation with condition \eqref{wn}, it is easy to show that, for every $\varepsilon>0$, if
		\[
		\limsup_{n \to \infty} \int_{p\left( z_{n}^{\ast }\right) +\varepsilon}^{1}\int_{0}^{1}\left( 1-x\right) H_{n}\left( dx\mid p_{\tilde{\tau}}\right) Q_{n}(dp_{\tilde{\tau}})>0,
		\]
		then we would have \(\limsup_{n \to \infty} W_n\left( p\left( z_0\right) \right) <R\left( p\left(z_0\right) \right)\), which contradicts $b_n(z_0)=1$ (principal's optimality) when $n$ is sufficiently large. Hence, for every \(\varepsilon >0\), we have\begin{equation}
			\label{limit}
			\limsup_{n \to \infty}\int_{p\left( z_{n}^{\ast }\right) +\varepsilon}^{1}\int_{0}^{1}\left( 1-x\right) H_{n}\left( dx\mid p_{\tilde{\tau}}\right) Q_{n}(dp_{\tilde{\tau}})=0.
		\end{equation}
		Therefore, the assumption that \(\lim_{n \to \infty} \mathbb{E}_{n}\left( e^{-r_{2,n}\mathbb{T}}\right) <1\) implies that, for every \(\varepsilon >0\), we have 
		\begin{equation}
			\label{limit*}
			\limsup_{n \to \infty} \int_{p\left( z_{n}^{\ast }\right) }^{p\left( z_{n}^{\ast }\right)+\varepsilon }\int_{0}^{1}\left( 1-x\right) H_{n}\left( dx\mid p_{\tilde{\tau}}\right) Q_{n}(dp_{\tilde{\tau}})>0.
		\end{equation}
		For the remainder of this
		proof, we take \(\varepsilon >0\) such that \(p\left( z_{n}^{\ast }\right)+\varepsilon <\left( \frac{z_0+z^{** }}{2}\right)\). 
		\medskip
		
		\noindent\textbf{Step 2.} We show that condition \eqref{limit*} implies that, the noninvestible type has a profitable deviation by fully mimicking the investible type.
		
		Lemma~\ref*{lemma_a7} implies that if we let \(\mathbb{\bar{T}}_{m}\) be the random time that stops the play in the first time that the
		posterior leaves \(\left( m^{-1},1-m^{-1}\right)\) or that \(\mathbb{T}\) happens, then for each \(\upsilon >0\), we have 
		\[
		\lim_{n \to \infty} \mathbb{E}_{n}^{NI}\left\{ r_{1,n}\int_{0}^{\mathbb{\bar{T}}_{m}}e^{-r_{1,n}t}\mathbb{I}_{\left\{ a_{n}\left( p_{t}\right) \leq 1-\upsilon\right\} }dt\right\} =0.
		\]
		By the martingale property of beliefs we can take \(m\in \mathbb{N}\) large enough to make \(\limsup_{n \to \infty} \mathbb{P}^{NI}\left\{ \inf_{t\leq \mathbb{T}} p_{t}\leq m^{-1}\right\}\) as small as we want. Analogously, we can take \(m\) large enough to guarantee that whenever the posterior starts at \(\left( 1-m^{-1},1\right)\) then \(\limsup_{n \to \infty} \mathbb{P}^{NI}\left\{\inf_{t\leq \mathbb{T}}p_{t}\leq p\left( z_{n}^{\ast }\right) +\varepsilon\right\}\) is as small as we want. These two observations then imply that 
		\begin{equation}
			\label{limit2}
			\limsup_{n \to \infty} \mathbb{E}_{n}^{NI}\left\{ r_{1,n}\int_{0}^{\mathbb{T}}e^{-r_{1,n}t}(1-a_{n}\left( p_t\right)) dt\right\} =0.
		\end{equation}
		
		Next, let $y=e^{-r_{1,n}t}$. For \(\theta\in \{NI,I\}\), let \(Q_{n}^{\theta}\) stand for the distribution of \(p_{\mathbb{T}}\) (not \(p_{\tilde{\tau}}\) as above) given the strategy of type \(\theta\) and let \(H_{n}^{\theta}\left( \cdot \mid p_{\mathbb{T}}\right)\) stand for the conditional distribution of $y$ given \(p_{\mathbb{T}}\) and the strategy of type \(\theta\). On the one hand, using (\ref{limit}) and (\ref{limit2}), it is straightforward to see that, taking a subsequence if necessary, the limit payoff of the noninvestible type from following his equilibrium strategy is given by:
		\begin{equation}
			\label{limit3}
			\lim_{n \to \infty} u\int_{0}^{p\left( z_{n}^{\ast }\right) +\varepsilon}\int_{0}^{1}\left( 1-y\right) H_{n}^{NI}\left( dy\mid p_{\mathbb{T}}\right)Q_{n}^{NI}(dp_{\mathbb{T}})>0,
		\end{equation}
		where the positive sign follows from (\ref{limit*}). On the other hand, the limit payoff of the noninvestible type from following the strategy of the investible type (i.e., always boosting performance with probability $1$) is given by:
		\begin{equation}
			\label{limit4}
			\lim_{n \to \infty} u\int_{0}^{p\left( z_{n}^{\ast }\right) +\varepsilon }\int_{0}^{1}\left( 1-y\right) H_{n}^{I}\left( dy\mid p_{\mathbb{T}}\right) Q_{n}^{I}(dp_{\mathbb{T}})>0.
		\end{equation}
		
		Next, a straightforward application of Bayes rule implies that \(H_{n}^{NI}\left( \cdot \mid p_{\mathbb{T}}\right) = H_{n}^{I}\left( \cdot
		\mid p_{\mathbb{T}}\right) \) for every $p_{\mathbb{T}}\in \left(0,1\right)$. Moreover, using \(p\left( z_{n}^{\ast }\right) +\varepsilon <\left( \frac{z_0+z^{**}}{2}\right)\) and Bayes rule, one can find \(\xi >1\) such that \(Q_{n}^{I}(A)\geq \xi Q_{n}^{NI}(A)\) for every (Borel-measurable) \(A\subset \left[0,p\left( z_{n}^{\ast }\right) +\varepsilon \right]\). Hence, subtracting (\ref{limit3}) from (\ref{limit4}) we obtain an expression at least as large as 
		\begin{equation*}
			\label{limit5}
			\lim_{n \to \infty} \left( \xi -1\right) u \int_{0}^{p\left( z_{n}^{\ast }\right)+\varepsilon }\int_{0}^{1}\left( 1-y\right) H_{n}^{NI}\left( dy\mid p_{\mathbb{T}}\right) Q_{n}^{NI}(dp_{\mathbb{T}})>0.
		\end{equation*}
		This implies that the noninvestible type can profitably deviate by fully mimicking, which leads to a contradiction and concludes the proof.
	\end{proof}
	
	\begin{proof}[Proof of Theorem \ref{t:patientlimit}]
		First, for the agent, Lemmas \ref*{lemma_a6} and \ref*{lemma_a8} tell us that $\lim_{n\to\infty}v_n(z)=u$ for all $z<z^{**}$, and Lemma  \ref*{lemma_a9} implies that  $\lim_{n\to\infty}v_n(z)=0$ for all $z>z^{**}$.
		
		Next, for the principal, we first argue  that $W_n(\cdot)$ converges pointwise to $\max\{0,R(\cdot)\}$. In light of Corollary \ref{cor:convexity}, we continuously extend $W_n(\cdot)$ from $(0,1)$ to $[0,1]$ by setting $W_n(0)=0$ and $W_n(1)=\frac{\lambda}{r_{2,n}+\lambda}w_{NI}$. Lemma \ref*{lemma_a9} implies that $\lim_{n\to\infty} W_n(p(z))= R(p(z))$ for all $z>z^{**}$. We now show that $\lim_{n\to\infty}W_n(p(z))=0$ for all $z\leq z^{**}$. To see this, fix any $z\leq z^{**}$ and take any $\varepsilon>0$. Since $R(p(z^{**}))=0$, there exists $\delta>0$ such that $R(p(z^{**})+\delta)<\frac{\varepsilon}{2}$. But then, we can find $n^*$ such that for every $n>n^*$, $W_n(p(z^{**})+\delta)<\varepsilon$. Since $W_n(\cdot)$ is increasing, it follows that $W_n(p(z))<\varepsilon$ for every $n>n^*$. So we must have $\lim_{n\to\infty}W_n(p(z))=0$, because $\varepsilon$ is arbitrary and $W_n(\cdot)$ is bounded below by $0$. 
		
		To show uniform convergence, note that for any fixed $n$, $W_n(\cdot)$ is bounded below by $\frac{\lambda}{r_{2,n}+\lambda}\max\{0,R(\cdot)\}$ such that $W_n(1)=\frac{\lambda}{r_{2,n}+\lambda}R(1)$. Because $W_{n}(\cdot)$ is convex and increasing, $|W_{n}'(\cdot)|$ is bounded above by $(w_{NI}-w_I)$, and hence $\{W_{n}\}_n$ is uniformly equicontinuous. Since $W_n$ converges pointwise to $\max\{0,R\}$, invoking Arzel\`{a}-Ascoli theorem we conclude that the convergence is uniform.
	\end{proof}
	
	
\end{document}